\documentclass[reqno,letterpaper]{amsart}
\usepackage{amsmath,amssymb,amsthm,graphicx,mathrsfs,url}
\usepackage[usenames,dvipsnames]{color}
\usepackage[colorlinks=true,linkcolor=Red,citecolor=Green]{hyperref}
\usepackage{amsxtra}
\usepackage{epstopdf}
\usepackage{wasysym}
\usepackage{bm}
\usepackage{comment}
\usepackage{enumitem}

\usepackage{mathtools}

\def\arXiv#1{\href{http://arxiv.org/abs/#1}{arXiv:#1}}

\usepackage{graphicx}

\def\?[#1]{\textbf{[#1]}\marginpar{\Large{\textbf{??}}}}

\def\smallsection#1{\smallskip\noindent\textbf{#1}.}
\let\epsilon=\varepsilon 

\newcommand{\RR}{{\mathbb R}}
\newcommand{\ZZ}{{\mathbb Z}}

\newcommand{\CC}{{\mathbb C}}

\newcommand{\Z}{\mathbb{Z}}

\newcommand{\C}{\mathbb{C}}

\newcommand{\CI}{C^{\infty}}

\DeclareMathOperator{\re}{Re}
\DeclareMathOperator{\im}{Im}
\makeatletter
\newcommand*{\rom}[1]{\expandafter\@slowromancap\romannumeral #1@}
\makeatother

\newtheorem{theo}{Theorem}
\newtheorem{prop}{Proposition}[section]	

\newtheorem{Assumption}{Assumption}

\newtheorem{lemm}[prop]{Lemma}
\newtheorem{corr}[prop]{Corollary}
\theoremstyle{definition}
\newtheorem{rem}{Remark}

\newtheorem*{rmks}{Remarks}

\numberwithin{equation}{section}

\DeclareMathOperator{\Res}{Res}
\DeclareMathOperator{\Spec}{Spec}

\DeclareMathOperator{\Hom}{Hom}

\let\Im=\Imag

\DeclareMathOperator{\loc}{loc}

\let\Re=\Real

\DeclareMathOperator{\tr}{tr}

\title[Chiral limit of twisted trilayer graphene]{Chiral limit of twisted trilayer graphene}
\author{Simon Becker}
\address[Simon Becker]{ETH Zurich, 
Institute for Mathematical Research, 
Rämistrasse 101, 8092 Zurich, 
Switzerland}
\email{simon.becker@math.ethz.ch}

\author{Tristan Humbert}
\address[Tristan Humbert]{ENS Paris, Département de Mathématiques et Applications, 
Rue d'Ulm, Paris, 
France}
\email{tristan.humbert@ens.psl.eu}

\author{Jens Wittsten}
\address[Jens Wittsten]{Department of Engineering, University of Bor{\aa}s, SE-501 90 Bor{\aa}s, Sweden}
\email{jens.wittsten@hb.se}

\author{Mengxuan Yang}
\address[Mengxuan Yang]{Department of Mathematics, University of California,
Berkeley, CA 94720, USA.}
\email{mxyang@math.berkeley.edu}

\begin{document}

\begin{abstract}
We initiate the mathematical study of the Bistritzer-MacDonald Hamiltonian for twisted trilayer graphene in the chiral limit (and beyond). We develop a spectral theoretic approach to investigate the presence of flat bands under specific magic parameters. This allows us to derive trace formulae that show that the tunnelling parameters that lead to flat bands are nowhere continuous as functions of the twisting angles. 
\end{abstract}
\maketitle

\section{Introduction}
\begin{figure}
\label{fig:magic}
\includegraphics[width=7.5cm]{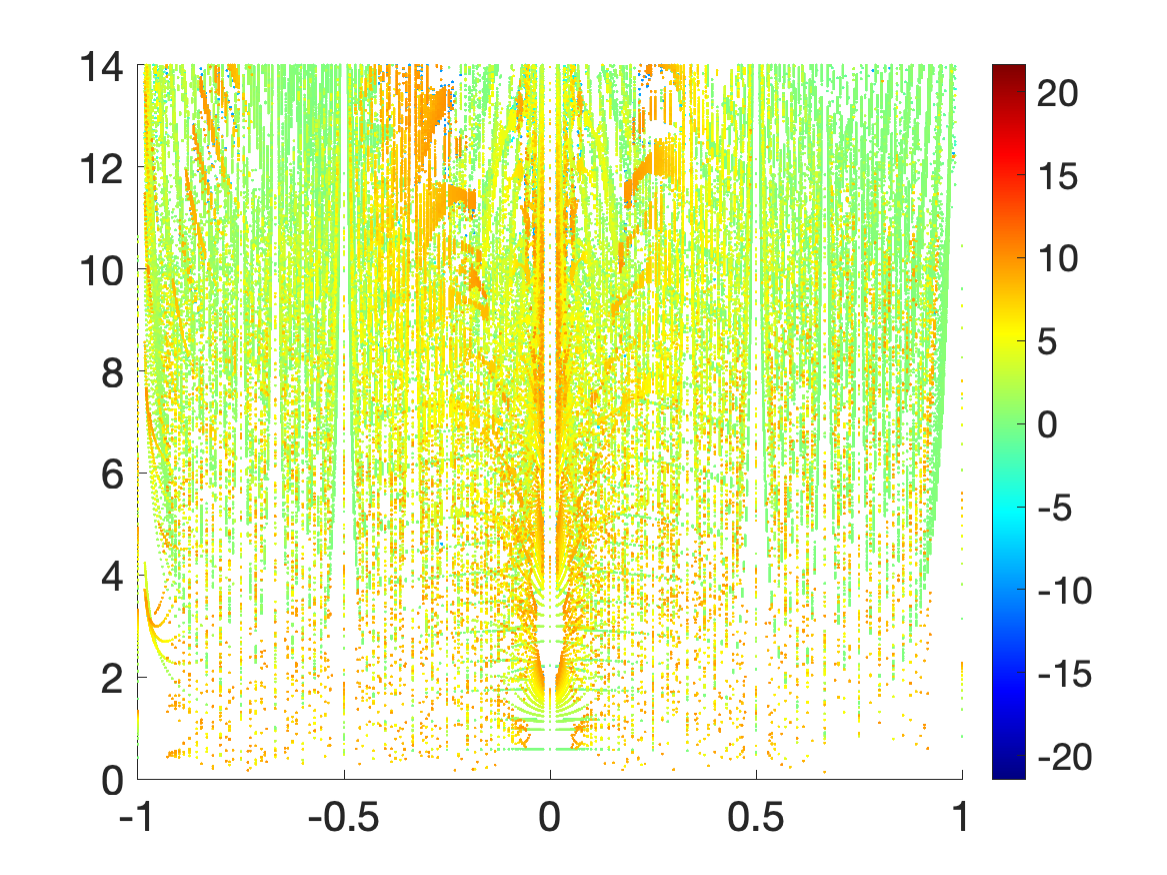}
\includegraphics[width=7.5cm]{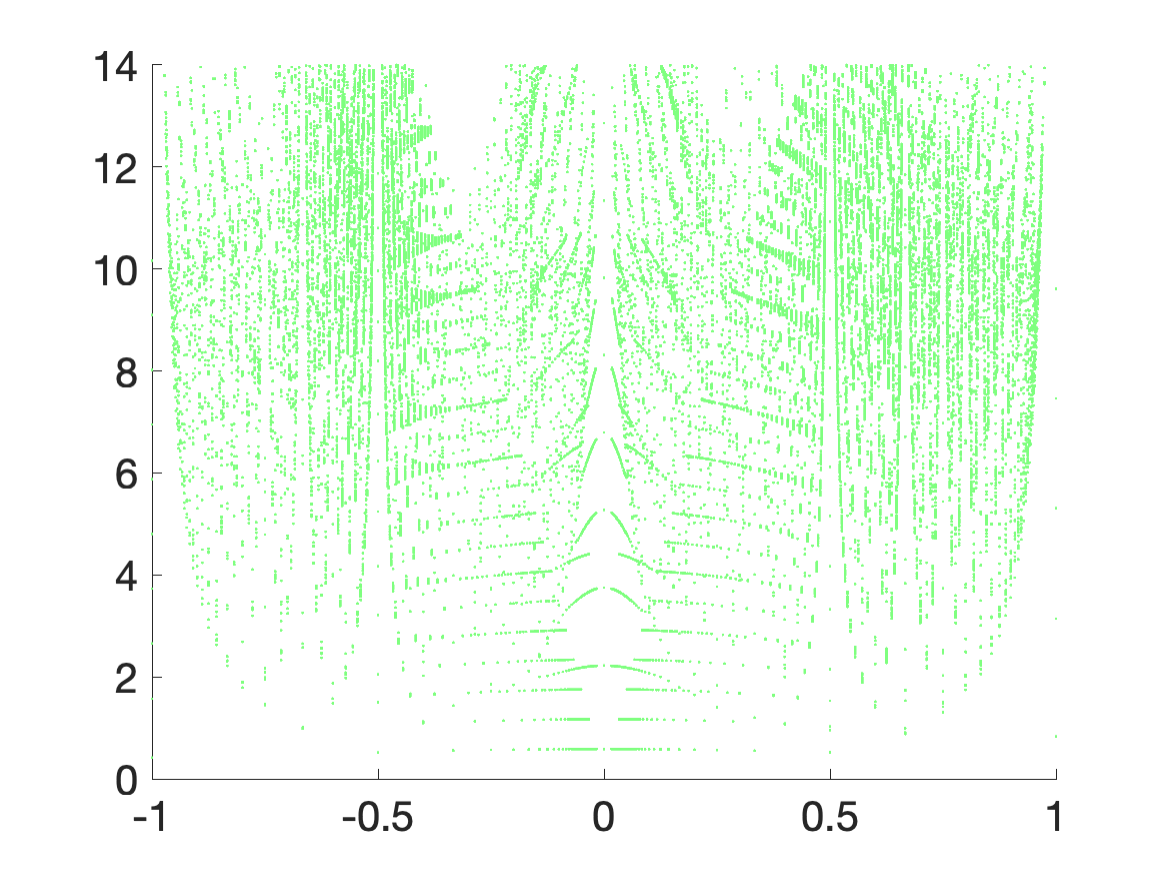}
\caption{(Left): Magic parameters at which the chiral limit of the Hamiltonian \eqref{eq:original} exhibits a flat band. $x$-axis is ratio of twisting angles $\frac{\zeta_2}{\zeta_1}.$ We assume $\alpha_{23}=\alpha_{21}$. Then $y$-axis is real part of magic parameter $\Re(\alpha_{12})$ and color coding is imaginary part $\Im(\alpha_{12}).$ (Right): Only real magic parameters ($\Im (\alpha_{12})=0$).}
\end{figure}

Twisted trilayer graphene (TTG) is a material formed by stacking three sheets of graphene, with slight relative twisting angles $\zeta_{1},\zeta_{2}$ between layers one and two, and two and three, respectively. This stacking arrangement gives rise to a captivating interference pattern known as a moiré pattern (see Figure \ref{fig:moire}). The electron tunnelling described by this pattern leads to significant modifications in the electronic properties of the material, such as flat bands at certain magic parameters, see Figure \ref{fig:magic}.

Each layer of graphene is composed of carbon atoms arranged in a two-dimensional honeycomb lattice, featuring two distinct atom types, A and B, per fundamental domain. At sites where neighboring layers (top and middle or middle and bottom) align, inter-layer interactions between atoms of the same type, referred to as AA sites, occur. Additionally, tunnelling interactions between atoms of different types take place at AB and BA sites, where atoms of type A are stacked over atoms of type B, and vice versa.

The exploration of different graphene layer configurations, including twisted trilayer graphene, expands on the research conducted on twisted bilayer graphene \cite{magic,BEWZ21,BEWZ22,bhz1,bhz2,bhz23}. These multilayer systems offer increased tunability due to a larger set of parameters \cite{khalaf2019magic,LVK22}. We also want to mention recent progress on a helical version of twisted trilayer graphene \cite{Dev23,GMM23a,GMM23b}. The study of the continuum or Bistritzer-MacDonald (BM) model for twisted trilayer graphene exhibits, unlike the BM model for bilayer graphene \cite{BM11}, commensurable and incommensurable twisting angles (see Assumption \ref{ass:angles}). Thus, TTG serves as an example for a whole range of materials where commensurability matters. The initial theoretical analysis of twisted trilayer graphene revealed a similar phenomenon of electronic bands flattening at various \emph{magic} angles. This breakthrough subsequently led to experimental observations of correlated phenomena \cite{PCWT21}.

In this paper we study the BM model of twisted trilayer graphene described by the Hamiltonian
\begin{equation}
\label{eq:original}
 H(\alpha,\tilde \alpha) = \begin{pmatrix} W(\tilde \alpha) & D(\alpha)^* \\ D(\alpha) & W(\tilde \alpha)  \end{pmatrix}
 \end{equation}
with 
\begin{equation}
\begin{split}
\label{eq:original2} 
D(\alpha) &= \begin{pmatrix} 2 D_{\bar z} &  \alpha_{12} U( pz ) & 0 \\ 
  \alpha_{12} U(- pz) &2 D_{\bar z} & \alpha_{23} U(p\tfrac{\zeta_2}{\zeta_1} z)\\ 
0 &   \alpha_{23} U(-p\tfrac{\zeta_2}{\zeta_1} z) &2 D_{\bar z} \end{pmatrix},\\ 
W(\tilde \alpha) &= \begin{pmatrix} 0 & \tilde \alpha_{12} V(p z) & 0 \\  (\tilde \alpha_{12} V(p z))^* & 0  &\tilde \alpha_{23} V(p\frac{\zeta_2}{\zeta_1} z)  \\ 0 & (\tilde \alpha_{23} V(p\frac{\zeta_2}{\zeta_1} z))^* & 0 \end{pmatrix}.
\end{split}
\end{equation}
The parameters $\alpha=(\alpha_{12},\alpha_{23})$ and $\tilde\alpha=(\tilde\alpha_{12},\tilde\alpha_{23})$ describe (rescaled) hopping amplitudes between layers $i$ and $j$ at AB/BA and AA sites, respectively.
We shall mainly focus on the \emph{chiral limit}, which discards the tunnelling at AA sites by setting $\tilde \alpha=0$ in \eqref{eq:original} and gives the Hamiltonian
\begin{equation*}
H(\alpha):=H(\alpha,0) = \begin{pmatrix} 0 & D(\alpha)^* \\ D(\alpha) & 0 \end{pmatrix}.
\end{equation*}
The \emph{anti-chiral limit} is obtained by setting $\alpha=0$, instead. In Section \ref{sec:derivation} we provide a brief discussion and derivation of the Hamiltonian \eqref{eq:original} from the continuum model for twisted trilayer graphene. In \eqref{eq:original}, the parameters $\zeta_1$ and $\zeta_2$ describe the relative (small) twisting angle between layers. In Figure \ref{fig:moire} we show examples of the moir\'e pattern formed in TTG for different $\zeta_1$ and $\zeta_2$, and we will throughout the paper assume that they are commensurable to analyze the electronic band structure:
\begin{Assumption}[Commensurable angles]
\label{ass:angles}
We assume that $\frac{\zeta_2}{\zeta_1} \in \mathbb Q \setminus\{0\}$, then we can write $\frac{\zeta_2}{\zeta_1}=3^j\frac{r_1}{r_2}$ with $r_1,j \in \ZZ$, $r_2 \in \mathbb N$, and both $r_1,r_2\not\equiv 0 \operatorname{mod} 3$.
We then set
\begin{equation}\label{eq:paq}
p:=\begin{cases}  
r_2 & \text{ for }j> 0,\\
3^{-j}r_2  & \text{ for }j\le 0,\\
\end{cases} \quad\text{ and }\quad q:=\begin{cases} 0 & \text{ for }j> 0, \\ 
r_1 & \text{ for }j \le 0.\end{cases}
\end{equation}
\end{Assumption}
Under Assumption \ref{ass:angles}, we have for $j\ne0$ that the factors $p$ and $p\frac{\zeta_2}{\zeta_1}$ in \eqref{eq:original2} satisfy $p\frac{\zeta_2}{\zeta_1} \in 3\ZZ+q$ with precisely one of $p$ and $q$ in $ 3\ZZ$, and therefore also precisely one of $p$ and $p\frac{\zeta_2}{\zeta_1}$ in $3\ZZ$, while for $j=0$ we have neither $p$ nor $q=p\frac{\zeta_2}{\zeta_1}$ in $3\ZZ$. In particular, at least one of $p$ and $q$ does not belong to $3\ZZ$, and if $j\leq 0$ then $p\frac{\zeta_2}{\zeta_1} =q$.

\begin{figure}
\includegraphics[width=.4\textwidth]{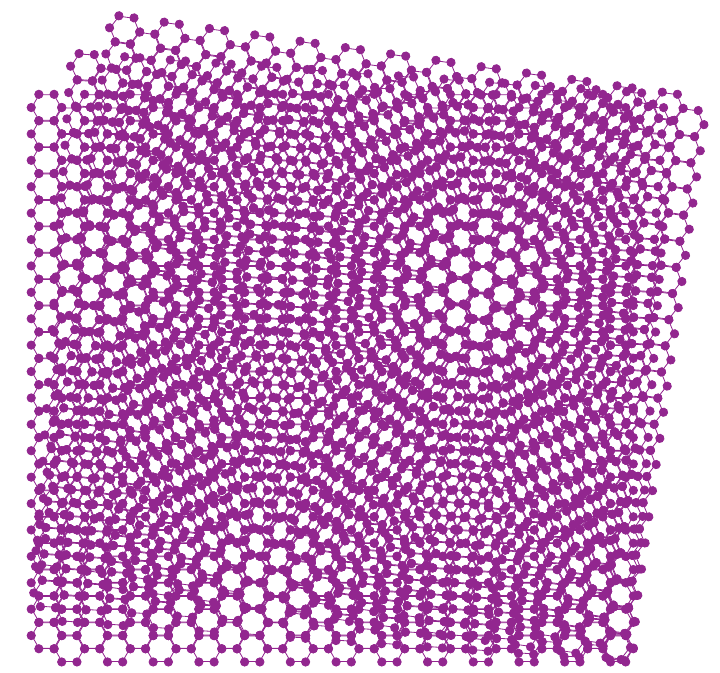}
\quad
\includegraphics[width=.4\textwidth]{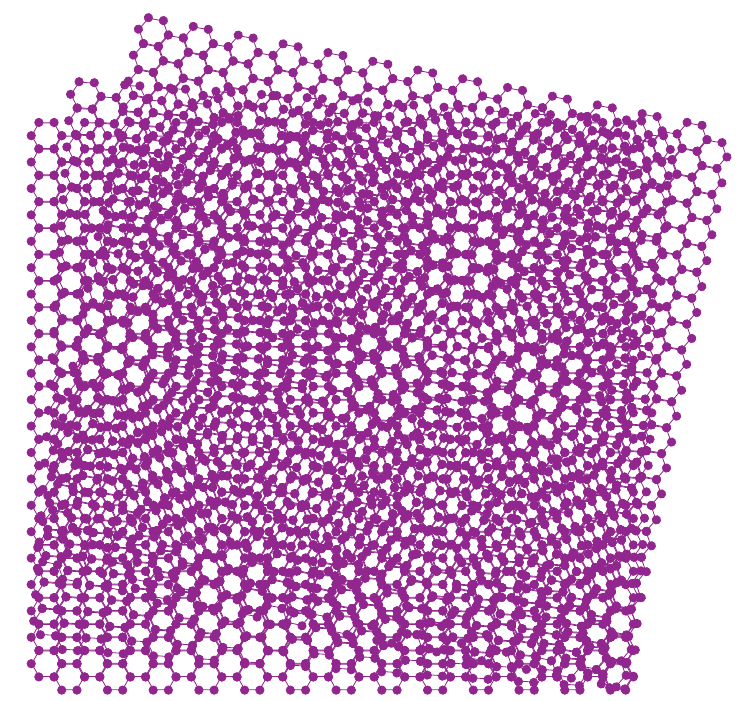}
\caption{Moiré patterns for different twisting configurations. The left panel shows three layers stacked with equal relative twist angles of $4^\circ$ between each layer, while in the right panel, the relative twist angles between layers are $4^\circ$ and $7^\circ$ degrees. \label{fig:moire}}
\end{figure}

We introduce $\Gamma:=4\pi i (\omega\ZZ\oplus \omega^2\ZZ)$, where $\omega=e^{2\pi i/3}$ is a third root of unity.
Then $\Gamma_3:=\Gamma/3$ is the \emph{moir\'e lattice},
and $U, V\in \CI(\CC)$ in \eqref{eq:original2} satisfy
\begin{align}
\label{eq:UaV}
 U(z + na) &= \bar \omega^{n(a_1+a_2)}U(z), \quad U(\omega z) = \omega U(z), \quad \overline{U(\bar z)} = U(z), \\ 
V(z + na) &= \bar \omega^{n(a_1+a_2)}V(z), \quad V(\omega z) = V(z), \quad 
\overline{V(z)} = V(-z)=V(\bar z)
\label{eq:onlyV}
\end{align}
for $n \in \mathbb N$ and $a = \frac{4\pi i}{3}(a_1\omega+ a_2 \omega^2)\in\Gamma_3$.
In particular, $U$ and $V$ are periodic with respect to $\Gamma$.
We provide a complete characterisation of such potentials in Proposition \ref{prop:1}. A standard example of $U$ in \eqref{eq:UaV} is
\begin{equation}
\label{eq:standard_pot}
U_0(z)=\sum_{j=0}^2\omega^j e^{\frac{1}2(z\bar\omega^j-\bar z\omega^j)}
\end{equation}
and it will serve as our reference potential for simulations, if nothing else is mentioned.

\begin{figure}
\includegraphics[width=7.5cm]{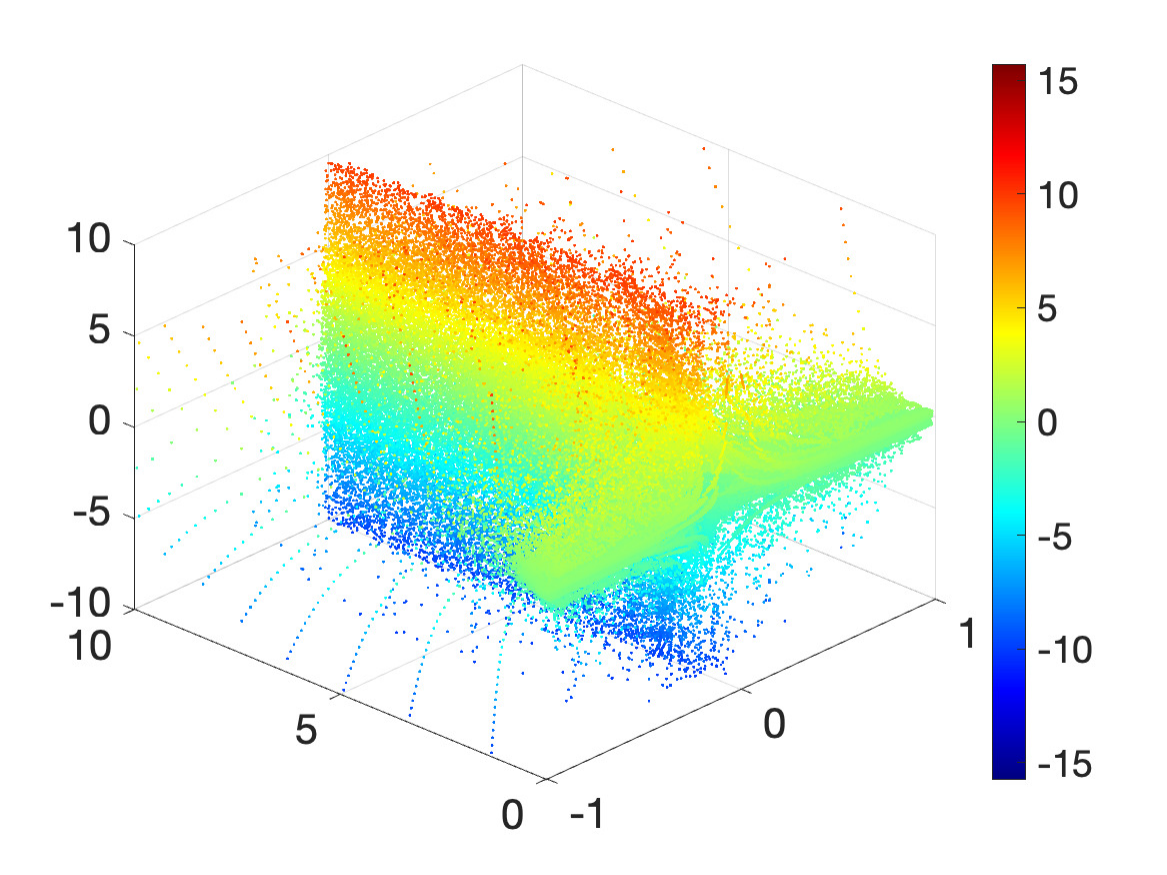}
\includegraphics[width=7.5cm]{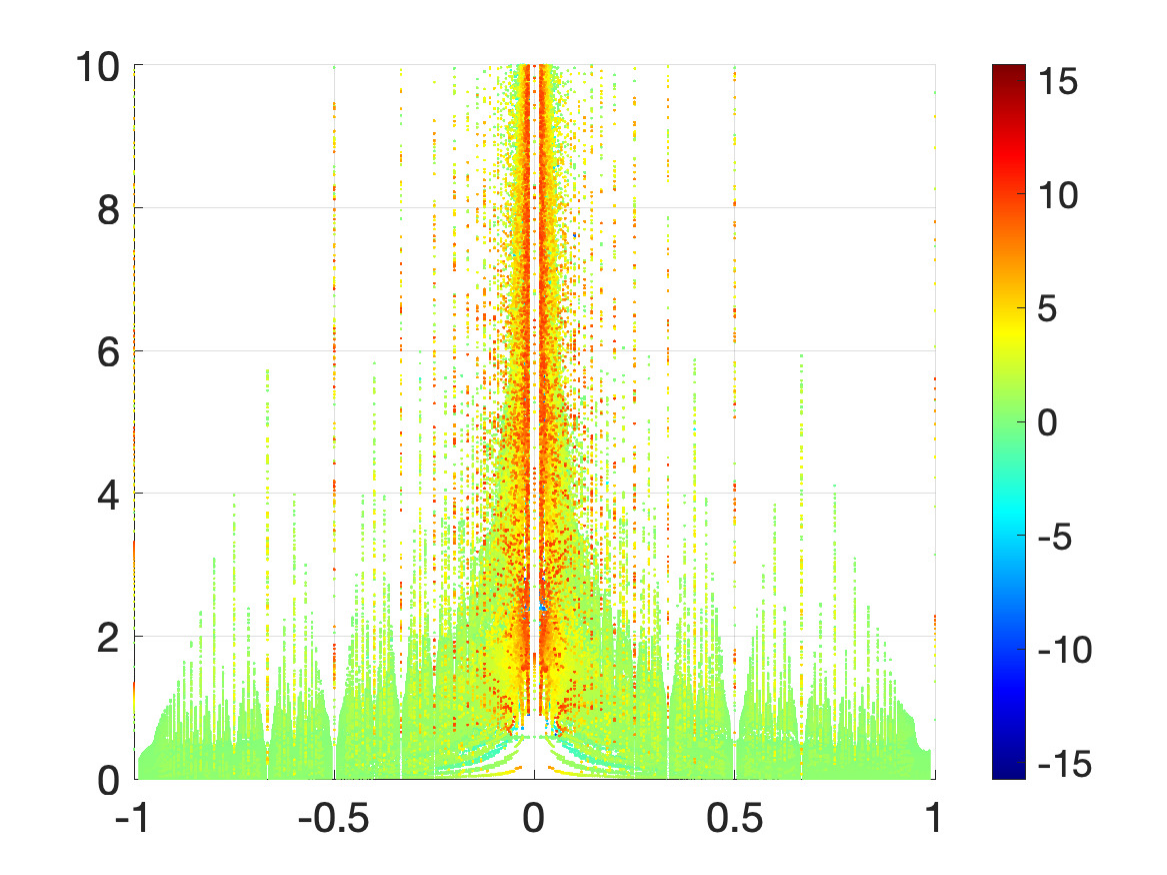}
\caption{Magic angles showing ratio of twisting angles $\zeta_2/\zeta_1$ ($x$-axis) and $\Re(\alpha_{12}/\zeta_2)$ ($y$-axis) with color coding given by $\Im(\alpha_{12}/\zeta_2).$}
\end{figure}

We consider now the chiral Hamiltonian $H(\alpha)$, and the eponymous property 
\begin{equation}\label{eq:chiralsymmetry}
H(\alpha) =- \sigma_3 H(\alpha) \sigma_3.
\end{equation}
We define $\langle k, z \rangle :=\Re( k\bar z) = \frac12(k\bar{z} + \bar{k}z)$.
From Assumption \ref{ass:angles} it follows that the potentials $U(\pm pz)$ and $U(\pm p\frac{\zeta_2}{\zeta_1}z)$ in the Hamiltonian $H(\alpha,\widetilde \alpha)$ in \eqref{eq:original} are periodic up to a power of $\omega$ with respect to the moir\'e lattice $\Gamma_3$.
Since $H(\alpha,\widetilde \alpha)$ commutes with the translation operator 
\[\mathscr L_au(z) := \operatorname{diag}(\omega^{p(a_1+a_2)},1,\bar\omega^{q(a_1+a_2)},\omega^{p(a_1+a_2)},1,\bar \omega^{q(a_1+a_2)}) u( z + a ),  \quad a = \tfrac{4}{3}\pi i (\omega a_1+ \omega^2 a_2)\in \Gamma_3,\]we may enforce Floquet boundary conditions $\mathscr L_a u=e^{i\langle a,k\rangle}u$ with $a\in\Gamma_3$ for arbitrary $k\in\CC$, or equivalently $k\in\CC/\Gamma_3^*$.
We can then conjugate the Hamiltonian to obtain
\begin{equation}
\label{eq:Floquet}
H_{k}(\alpha,\widetilde \alpha) := e^{-i\langle k, z\rangle} H(\alpha,\widetilde \alpha)e^{i\langle k, z\rangle}:H^1_{\text{loc}}(\CC) \to L^2_{\text{loc}}(\CC) 
\end{equation}
so that in the chiral limit
\begin{equation}
\label{eq:Floquet_intro}
H_{k}(\alpha) = \begin{pmatrix} 0 & D(\alpha)^*+\bar k 
\\ D(\alpha) +k 
& 0  \end{pmatrix}.
\end{equation}
The boundary condition for the operator $H_k$ reduces to $\mathscr L_a\psi=\psi$ and we denote the associated $L^2$ space by $L^2_0$ which is a subspace of $L^2_{\text{loc}}$ with the $L^2$ norm on $\CC/\Gamma_3$.
The Floquet-Bloch decomposition of $H(\alpha,\widetilde \alpha)$ then implies that the spectrum can be described as
$$
\Spec_{L^2(\CC)}H(\alpha,\widetilde \alpha)=\bigcup_{k\in\CC/\Gamma_3^*}\Spec_{L^2_0}H_k(\alpha,\widetilde \alpha)
$$
where $\Gamma_3^*$ is the dual lattice of $\Gamma_3$ (see \eqref{eq:dual_lattice} for a definition).

For each $k$, the operator $H_k(\alpha)$ is an elliptic differential operator with discrete spectrum. 
This yields a family of discrete spectra of $\{H_k(\alpha)\}_{k\in\CC/\Gamma_3^*}$ (cf.~\cite[Section 2C]{BEWZ22}): 
\begin{equation}\label{eq:specHk}
\begin{gathered}
\Spec_{L^2_0}H_k(\alpha)=\{E_{\pm j}(k,\alpha)\}_{j=1}^\infty,\quad E_{-j}(\alpha,k)=-E_j(\alpha,k),
\\ E_{j+1}(k,\alpha)\ge E_{j}(k,\alpha)\ge 0,\quad j\ge1.
\end{gathered}    
\end{equation}
We refer to the image of $ k\mapsto E_{j}(k,\alpha)$ as a {\it band}, and say that $H(\alpha)$ has a \emph{flat band} at zero if $E_1(k,\alpha)=0$ for all $k\in\CC$. 
We then introduce the set
\begin{equation}\label{def:magicalphas}
\mathcal A:=\{\alpha\in\CC^2:E_1(k,\alpha)\equiv 0\text{ for all }k\in\CC\},
\end{equation}
which we refer to as the set of \emph{magic parameters} or just \emph{magic} $\alpha$'s.
Note that by the above we have $\alpha\in\mathcal A$ if and only if $0\in \Spec_{L^2_0}H_k(\alpha)$ for all $k\in\CC$, i.e.,
\begin{equation}
\label{eq:magic_small}
\alpha\in\mathcal A\quad\Longleftrightarrow\quad \Spec_{L^2_0} (D(\alpha))=\CC.
\end{equation}

Just like in \cite{BEWZ22}, the key point to obtain a characterisation of magic parameters is to establish the existence of \emph{protected states}. Indeed, using the symmetries of the Hamiltonian, we prove in Proposition \ref{prop:uialpha} the existence of three protected states
\begin{equation}
\label{eq:3protect}
   \varphi(\alpha,z)\in\ker_{L^2_{p,0}}D(\alpha),\ \ \psi(\alpha,z)\in\ker_{L^2_{0,0}}D(\alpha),\ \ \rho(\alpha,z)\in\ker_{L^2_{-q,0}}D(\alpha),
\end{equation}
analytic in $\alpha$ with $\varphi(0,z)=e_1$, $\psi(0,z)=e_2$, $\rho(0,z)=e_3$, where $\{e_1,e_2,e_3\}$ is the standard basis of $\CC^3$.
\begin{figure}
\includegraphics[height=5.5cm]{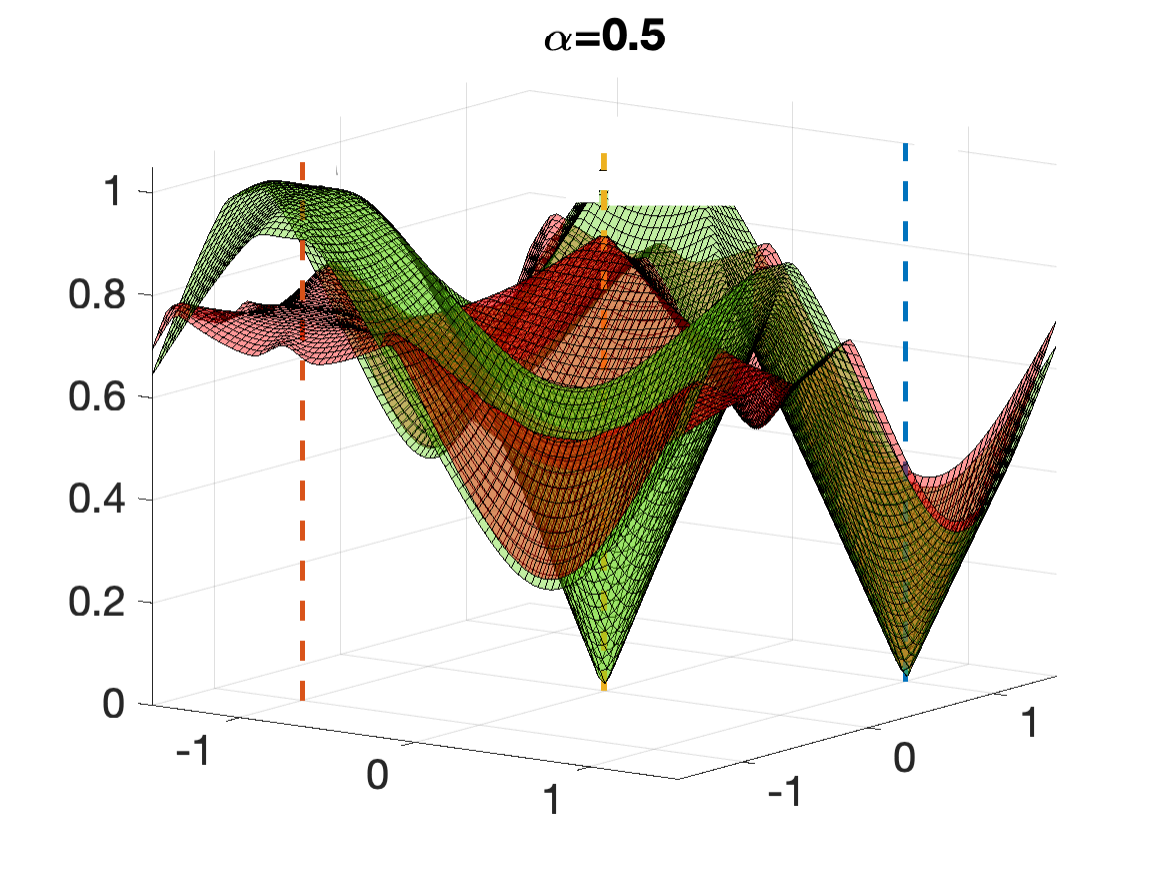}
\includegraphics[height=5.5cm]{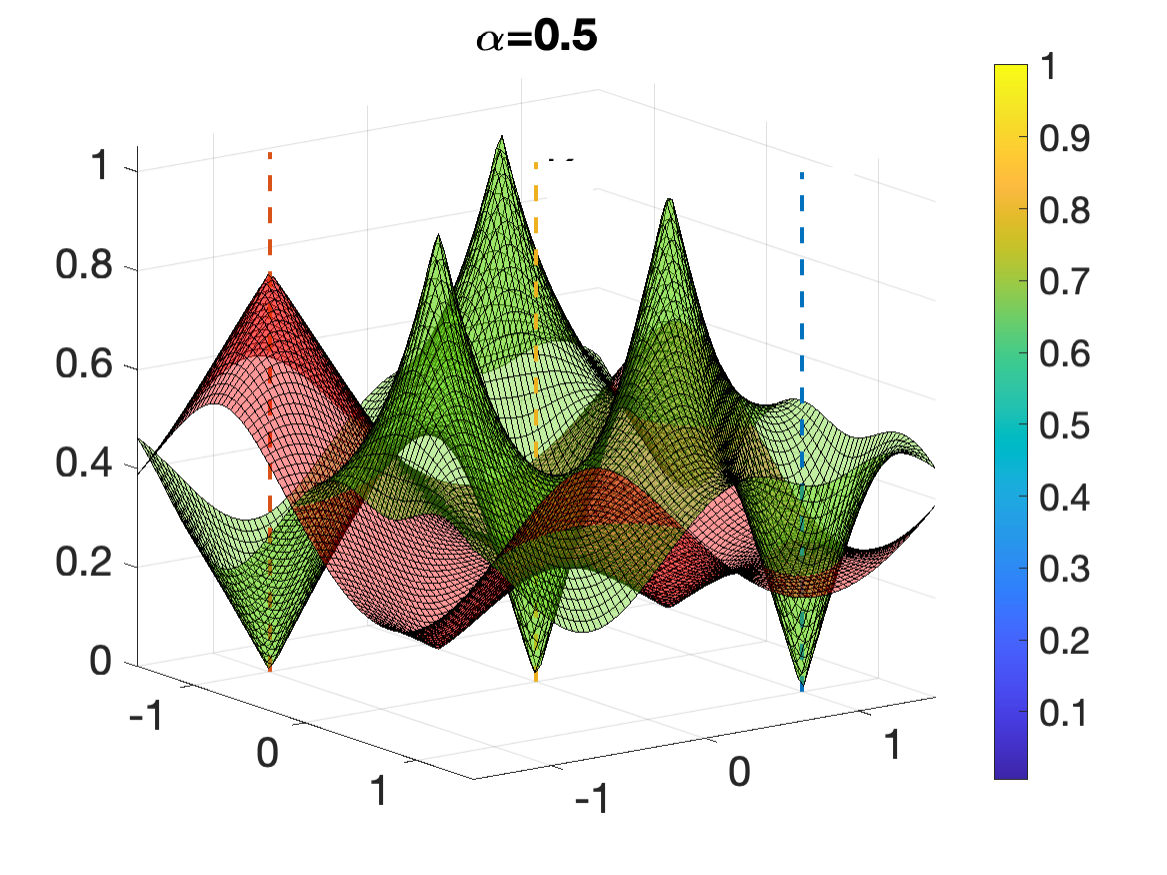}
\caption{Band structure for different $p,q$: Case \ref{case1} (left panel; $p\equiv 2, q\equiv 0$ mod 3), which shares many qualitative properties with TBG, where at one Dirac point there is a two fold degenerated Dirac cone (protected state) and at another Dirac point there is a simple Dirac cone (protected state); Case \ref{case2} (right panel; $p\equiv q \equiv 1$ mod 3) where there exist three Dirac cones (protected states) at different Dirac points.}
\label{fig:protected_cones}
\end{figure}
Here, the subspaces $L^2_{m,j}$ are associated to irreducible representations of the symmetry group, see Subsection \ref{ss:symmetries}.

From Proposition \ref{prop:uialpha} we also obtain the following dichotomy for the three protected states that we illustrate in Figure \ref{fig:protected_cones}:
\begin{enumerate}[label=\text{Case \Roman*:},ref=\Roman*]
    \item \label{case1}  $-q,p,0$ are not mutually different numbers in $\mathbb Z_3$. Then two of the protected states fall into the same subspace. This implies the existence of overlapping Dirac cones in the Brillouin zone $\CC/\Gamma_3^*$, see Figure \ref{fig:protected_cones}. We also observe that, for $p\not\equiv0$ mod 3, we have $-p \in \mathbb Z_3 \setminus \{0,p,-q\}$, which we shall usually assume without loss of generality (see Remark \ref{rem:pnotzeromod3}).  
    \item \label{case2} $-q,p,0$ are mutually different numbers in $\mathbb Z_3$. Then
    \begin{equation*}
        \ker_{L^2_{p,0}}D(\alpha)\neq \{0\}, \ \ \ker_{L^2_{0,0}}D(\alpha)\neq \{0\}, \ \ 
        \ker_{L^2_{-q,0}}D(\alpha)\neq \{0\}.
    \end{equation*}
    This corresponds to three separate Dirac cones in Figure \ref{fig:protected_cones}.
\end{enumerate}

While Case \ref{case1} behaves in many ways similar to twisted bilayer graphene and most of the result of \cite{bhz1,bhz2,bhz23} have an analogue to this case, Case \ref{case2} exhibits many new phenomena. We note that equal twisted trilayer graphene introduced in \cite{PopTar23-1} falls into Case \ref{case2}.

Having these three protected states, we can define their Wronskian
\begin{equation}
\label{eq:det}  
W(\alpha) : = \det [ \varphi, \psi, \rho ] , 
\end{equation}
which vanishes if and only if $\alpha$ is magical. A similar treatment has also been proposed by Popov--Tarnopolsky \cite{Tarno}.
\begin{theo}[Flat bands \& Wronskian]
\label{theo:wronskian_intro}
The parameter $\alpha \in \CC^2 $ is magic if and only if the Wronskian vanishes:
\begin{equation*}
        \alpha\in\mathcal{A} \iff W (\alpha) = 0.
    \end{equation*}
\end{theo}

Following \cite{BEWZ22} and \cite{bhz1}, we will use an alternative characterisation, for $\alpha_{12}\neq 0$, in terms of the compact and non-normal operator which we call the \emph{Birman-Schwinger operator} defined by
\begin{equation}
\label{eq:Bk}
 B_k(\tfrac{\alpha_{23}}{\alpha_{12}}) := R(k) U(- pz)R(k) U( pz )+\Big(\frac{\alpha_{23}}{\alpha_{12}}\Big)^2(R(k) U(p\tfrac{\zeta_2}{\zeta_1} z)R(k) U(-p\tfrac{\zeta_2}{\zeta_1} z))
 \end{equation}
with $R(k) = (2D_{\bar z}+k)^{-1}$. The set of magic parameters is then characterised by
   \begin{theo}[Characterisation via the Birman-Schwinger operator]
   \label{theo:spectral_char}
   The parameter $\alpha \in \CC^2 \setminus  (\{0\} \times \CC)$ is magic if and only if $\frac{1}{\alpha_{12}^2} \in \Spec(B_k(\tfrac{\alpha_{23}}{\alpha_{12}}))$ for $k \notin \Gamma^*.$
   \end{theo}
The above characterisation is explicit up to fixing the hopping ratio $\tfrac{\alpha_{23}}{\alpha_{12}},$ i.e., the set of all magic parameters is a union over all possible ratios $\tfrac{\alpha_{23}}{\alpha_{12}}.$ In practice, one may just assume that $\alpha_{12}=\alpha_{23}$ for standard TTG.
The Birman-Schwinger operator $B_k$ is a Hilbert-Schmidt operator, which implies that for $\ell\geq 2$ all traces $\tr(B_k^{2\ell})$ are well-defined and moreover
$$\tr(B_k^{\ell})=\sum_{\alpha_{12}\in \mathcal A(\frac{\alpha_{23}}{\alpha_{12}})}\frac{1}{\alpha_{12}^{2\ell}}, \quad \ell \ge 2$$
where we denoted by $\mathcal A(\frac{\alpha_{23}}{\alpha_{12}})$ the set of magic parameters for the fixed hopping ratio $\alpha_{23}/\alpha_{12}$ (see \S \ref{ss:spectral_char} for a precise definition).
 We can therefore try to understand the structure of the set of magic parameters $\mathcal A(\frac{\alpha_{23}}{\alpha_{12}})$ by understanding the traces of the operator $B_k$. Similarly to what was done in \cite{bhz1}, we provide, in Theorem \ref{traceresult}, semi-explicit formulas for these traces. These expressions are enough to prove the following result.
 \begin{theo}
     \label{ratio}
  For the potential given in \eqref{eq:potentialU} with finite Fourier expansion with coefficients in $c_{n,m}\in \mathbb Q(\omega)$ and with the additional symmetry $$\overline{U(\bar z)}=U(z)$$, we have for any $\ell\geq 2$ that $\tr(B_k^{\ell})=q_{\ell}\pi/\sqrt 3$, where $q_{\ell}\in \mathbb Q$. As a consequence, if the set of magic parameters $\mathcal A(\frac{\alpha_{23}}{\alpha_{12}})$ is non-empty, it is infinite.
 \end{theo}
We will actually prove statements for more general potentials in Theorems \ref{rat} and Theorem \ref{theo:races_to_infinity}. 
As non-zero non-normal operators may have only zero in their spectrum, the existence of a magic parameter from Theorem \ref{theo:spectral_char} is not trivial. In \cite{BEWZ22}, the existence was established by computing explicitly the first trace and proving it was non-zero. In this article, we provide an explicit calculation of the first trace in Theorem \ref{explicit} and, as a corollary, we get that the set of magic parameters is infinite for any twisting angles satisfying Assumption \ref{ass:angles}. If we are additionally in Case \ref{case1}, then by adapting the argument in \cite{bhz23}, we obtain the existence of infinitely many non-simple magic parameters (see Proposition \ref{prop: multiplicity} for a precise definition of multiplicity).
\begin{corr}
\label{corr:degenerate}
 Consider the potential $U_0$ defined in  \eqref{eq:standard_pot} and assume that $\alpha_{23}/\alpha_{12}\in \mathbb Q$ is rational. Then the set of magic parameters $\mathcal A$ is infinite. Assuming Case \ref{case1}, the set of non-simple magic parameters is infinite.   
\end{corr}
We refer to Theorem \ref{theo:races_to_infinity2} for the most general statement. Numerical experiments, see Figure \ref{fig:allsimple}, suggest that the second part of Corollary \ref{corr:degenerate} does not hold for angles in Case \ref{case2}. {Note that Corollary \ref{corr:degenerate} implies that when the potential is the standard one given by \eqref{eq:standard_pot} and $\alpha_{23}/\alpha_{12}\in \mathbb Q$, then there are infinitely many $\alpha$ for which the corresponding chiral Hamiltonian $H(\alpha)$ has a flat band. This is in contrast to the anti-chiral model which never exhibits flat bands, see Appendix \ref{sec:ACL}.} 

We then investigate the continuity of magic parameters. Our main result shows that the set of magic parameters is maximally discontinuous when changing continuously the twisting angles $\zeta_1,\zeta_2$. By this we mean that the lowest order trace function computed in the preceding section is \emph{discontinuous at any point} in the following sense:

\begin{theo}\label{thm:discontinuity}
Let $\zeta_1$ and $\zeta_2$ with $\zeta_1/\zeta_2\in \mathbb Q$ fixed. For $r=\alpha_{23}/\alpha_{12}$ fixed, set the trace of unrescaled magic parameters to be
$$\mathcal S_4(r,\zeta_1,\zeta_2):=\bigg(\frac{p}{\zeta_1}\bigg)^4\sum_{\alpha_{12}\in \mathcal A(r)}\frac{1}{\alpha_{12}^4}.$$
Then there exists a sequence $(\zeta_1^{(n)},\zeta_2^{(n)})$ such that for any $n$, one has $\zeta_1^{(n)}/\zeta_2^{(n)}\in \mathbb Q$ with $(\zeta_1^{(n)},\zeta_2^{(n)})\to (\zeta_1,\zeta_2)$. Moreover, one has 
$$ \mathcal S_4(r,\zeta_1^{(n)},\zeta_2^{(n)})\to +\infty.$$
In particular, the trace $\mathcal S_4(r,\zeta_1^{(n)},\zeta_2^{(n)})$ is sequentially discontinuous at \emph{any} point.
\end{theo}

In Section \ref{s:gensimp} we show that for a generic choice of potentials \eqref{eq:newU} we can ensure in Case \ref{case1} that all magic angles are simple within each subspace $L^2_{-p,\ell}$ for $\ell \in \ZZ_3.$ If we then choose $V$ as in \eqref{eq:newU2} then the sublattice polarized flat bands will generically have Chern number $-1$ which is shown in Section \ref{sec:Chern}. Here, we find a very general argument that should be useful in general for the study of twisted $N$-layer systems.  It is interesting to compare with twisted bilayer graphene where additional symmetries allowed us to obtain a stronger result under a less restrictive class of perturbations \cite{bhz23}.

While most of our analysis so far mostly focus on the operator $D(\alpha)$, we will get back to the full Hamiltonian $H(\alpha)$ in Section \ref{sec:Theta} by studying its band structure when $\alpha$ is supposed to be simple, i.e., $E_1(k,\alpha)=0$ for any $k$ but $E_2(k_0,\alpha)\neq 0$ for some $k_0$. 
This leads naturally to the question of band touching: namely, does there exist $k_1$ such that $E_2(k_1,\alpha)=0$? In the case of twisted bilayer graphene, it was shown in \cite[Theo.~4]{bhz23} that there was a spectral gap for a simple or doubly-degenerate magic angle. Interestingly, in the case of twisted trilayer graphene, the first two bands always touch. Nevertheless, if $\alpha$ is simple then the point $K_0$ where the two bands touch is unique (see Section \ref{sec:Theta} for a more precise statement about the value of $K_0$).

We close out the paper by studying exponential squeezing of bands for the chiral model:
In Theorem \ref{t:squeezegenpot} in Section \ref{s:exp} we show that as the twisting angles between the two lattices approach zero, a large number of bands accumulate in an exponentially small neighbourhood around zero, determined by the twisting angle ratio. In particular, we observe that even though exponential, the rate of convergence depends significantly on the ratio of the twisting angles (see Figure \ref{fig:decay}).

\smallsection{Notation}
We denote by $\sigma_i$ the $i$-th Pauli matrix. Let $\Gamma\subset\C$ be the lattice $\Gamma=4\pi i (\omega\Z\oplus\omega^2\Z)$ where $\omega=e^{2\pi i/3}$ is a third root of unity. The dual lattice $\Gamma^*$ is defined as the $ k\in\C$ such that $\langle \gamma, k\rangle := \frac12(\gamma\bar k+\bar \gamma k)\in 2\pi\Z$ for all $\gamma\in\Gamma$, and is given by $\Gamma^*=\frac1{\sqrt 3} (\omega\Z\oplus\omega^2\Z)$. We will represent $k\in\Gamma^*$ as $k=k(k_1,k_2)=\frac1{\sqrt 3} (\omega^2k_1-\omega k_2)$. The moir\'e lattice and its dual lattice are given by 
\begin{gather*}
    \Gamma_3: = \tfrac{1}{3}\Gamma = \tfrac{4}{3}\pi i (\omega a_1 + \omega^2 a_2),\ a_1,a_2\in\ZZ;\ \  
    \Gamma_3^* = 3\Gamma^* = \sqrt{3} (\omega^2k_1-\omega k_2),\ k_1,k_2\in\ZZ.
\end{gather*}

\smallsection{Outline of article}
{\color{blue}}
\begin{itemize}
\item In Section \ref{sec:derivation}, we derive the full Hamiltonian $H(\alpha,\tilde\alpha)$ in \eqref{eq:original} and study its basic symmetries.
\item In Section \ref{sec:DiracPts}, we discuss the protection of Dirac points in the chiral limit and its consequences on magic parameters. Then, we state the Birman-Schwinger principle for twisted trilayer graphene and obtain the scalar operator \eqref{eq:Bk} whose eigenvalues characterize \emph{magic parameters}.
\item In Section \ref{sec:TF}, we study structural properties of trace formul\ae \ that we obtain from the Birman-Schwinger operator.
\item In Section \ref{sec:Tracesfor2}, we explicitly work out lowest order trace formul\ae \ that show the existence of infinitely many magic parameters only assuming Assumption \ref{ass:angles}. 
\item In Section \ref{sec:continuity}, we prove the discontinuity of magic parameters with respect to the twisting angle and the H\"older continuity of the spectrum of the Hamiltonian in Hausdorff distance.
\item In Section \ref{sec:gen_simpl}, we study the generic simplicity of flat bands in different representations and state conditions for degenerate magic parameters to exist. 
\item In Section \ref{sec:Theta}, we discuss the theta function construction to construct Bloch functions associated with flat bands.
\item In Section \ref{sec:Chern}, we use results from Section \ref{sec:gen_simpl} and \ref{sec:Theta} to compute the Chern number of the flat bands for a variety of cases. 
\item In Section \ref{s:exp}, we study exponential squeezing of bands in the limit of small twisting angles.
\item In Appendix \ref{sec:ACL}, we study the anti-chiral limit ($\alpha\equiv 0$) of twisted trilayer graphene and show that it does not exhibit any flat bands.
\end{itemize}

\section{The twisted trilayer graphene Hamiltonian}\label{sec:derivation}
\subsection{Derivation}\label{ss:derivation}

We start with the continuum model for twisted trilayer graphene with coupling parameters $\beta,\tilde\beta \in \CC^2$,
\[\begin{split} &\mathscr H(\beta,\tilde \beta) = \begin{pmatrix} H_D & T_{12}(\zeta_1 z)  & 0\\   T_{12}(\zeta_1 z)^* & H_D & T_{23}(\zeta_2 z) \\ 0 & T_{23}(\zeta_2 z)^* & H_D \end{pmatrix} \text{ with }T_{ij}(z) = \begin{pmatrix} \tilde \beta_{ij} V(z) & (\beta_{ij} U(-z))^* \\ 
\beta_{ij} U(z) & \tilde \beta_{ij} V(z)
\end{pmatrix}. \end{split} \]
Here, $H_D$ is a Dirac operator $H_D = D_{x_1} \sigma_1 + D_{x_2} \sigma_2$, and $\mathscr H$ acts on a wavefunction $$\psi =(\psi_{\uparrow,A},\psi_{\uparrow,B},\psi_{\rightarrow,A},\psi_{\rightarrow,B},\psi_{\downarrow,A},\psi_{\downarrow,B})$$ with lattice $\uparrow, \rightarrow, \downarrow$ (upper, middle, lower) and sublattice index $A,B$.
Writing this operator in the basis in which $\psi =(\psi_{\uparrow,A},\psi_{\rightarrow,A},\psi_{\downarrow,A},\psi_{\uparrow,B},\psi_{\rightarrow,B},\psi_{\downarrow,B})$, we see it is equivalent to a Hamiltonian 
\begin{equation*}
\mathcal H(\beta,\tilde \beta,\zeta) = \begin{pmatrix} \mathcal W(\tilde \beta,\zeta)  & \mathcal D(\beta,\zeta)^* \\ \mathcal D(\beta,\zeta)&\mathcal W(\tilde \beta,\zeta) \end{pmatrix} 
\end{equation*}
with matrices
\begin{equation}
\label{eq:off-diag}
\begin{split} \mathcal D(\beta,\zeta) &= \begin{pmatrix} 2 D_{\bar z} &  \beta_{12} U( \zeta_1 z ) & 0 \\ 
  \beta_{12} U(-\zeta_1 z) &2 D_{\bar z} & \beta_{23} U(\zeta_2 z)\\ 
0 &   \beta_{23} U(-\zeta_2 z) &2 D_{\bar z} \end{pmatrix}\text{ and }\\
W(\tilde \beta,\zeta) &=\begin{pmatrix}0 & \tilde \beta_{12} V(\zeta_1 z) & 0 \\ 
(\tilde \beta_{12} V(\zeta_1 z))^* & 0 & \tilde \beta_{23} V(\zeta_2 z) \\
0 & (\tilde \beta_{23} V(\zeta_2 z))^* & 0 \end{pmatrix}, \end{split}
\end{equation}
where $\zeta_k=\theta_{k+1} - \theta_k$ in the small-angle approximation, with $\theta_1,\theta_2,\theta_3 \in (0,2\pi)$ denoting the rotation angle of the three layers, respectively.
We then choose $p$ in accordance with \eqref{eq:paq}, and make the change of variables $z \zeta_1/p \mapsto z$. By introducing new coupling parameters $\alpha_{ij}:= \frac{p}{\zeta_1}\beta_{ij}$ and $\tilde\alpha_{ij}:= \frac{p}{\zeta_1}\tilde \beta_{ij}$, we obtain a new equivalent Hamiltonian $H(\alpha,\tilde \alpha)$ (with rescaled energy scale) given by \eqref{eq:original}. As mentioned in the introduction, the chiral limit is obtained by setting $\tilde\alpha=0$.

\begin{rem}\label{rem:pnotzeromod3}
We may without loss of generality assume that $p\not\equiv0$ mod 3, since otherwise we can just flip the trilayer upside down to effectively replace $p$ with $\lvert q\rvert\not\equiv0$ mod 3 as follows: Assume that $p\equiv0$ mod 3. Then, in the notation of Assumption \ref{ass:angles}, we have $\zeta_2/\zeta_1=3^{j}r_1/r_2$ with $j<0$ and $q=r_1\not\equiv0$ mod 3. Conjugating $\mathcal D(\beta,\zeta)$ by a unitary matrix, we get
$$
\begin{pmatrix}0&0&1\\0&1&0\\1&0&0\end{pmatrix}
\mathcal D(\beta,\zeta)\begin{pmatrix}0&0&1\\0&1&0\\1&0&0\end{pmatrix}=
\begin{pmatrix} 2 D_{\bar z} & \beta_{23} U(-\zeta_2 z)  & 0 \\ 
 \beta_{23} U(\zeta_2 z)  &2 D_{\bar z} & \beta_{12} U(-\zeta_1 z)\\ 
0 &  \beta_{12} U( \zeta_1 z )  &2 D_{\bar z} \end{pmatrix},
$$
which swaps $\zeta_1$ for $-\zeta_2$ (as well as $\beta_{12}$ for $\beta_{23}$), and we have $\zeta_1/\zeta_2=3^{-j}r_2\operatorname{sgn}(r_1)/\lvert r_1\rvert$ with $\tilde j=-j>0$, $\operatorname{sgn}(r_1)r_2\in\ZZ$ and $\lvert r_1\rvert\in\mathbb N$. Following the convention in Assumption \ref{ass:angles}, we set $\tilde p:=\lvert r_1\rvert=\lvert q\rvert\not\equiv0$ mod 3, and $\tilde q=0$. After making the change of variables $-\zeta_2 z\mapsto \tilde pz$ and setting $(\alpha_{12},\alpha_{23})=(\beta_{23},\beta_{12})\tilde p/(-\zeta_2)$, we end up with the equivalent system
$$
\begin{pmatrix} 2 D_{\bar z} & \alpha_{12} U(\tilde p z)  & 0 \\ 
 \alpha_{12} U(-\tilde p z)  &2 D_{\bar z} & \alpha_{23} U(\tilde p\frac{\zeta_1}{\zeta_2} z)\\ 
0 &  \alpha_{23} U( -\tilde p\frac{\zeta_1}{\zeta_2} z )  &2 D_{\bar z} \end{pmatrix}
$$
where $\tilde p\not\equiv0$ mod 3 and $\tilde p\zeta_1/\zeta_2=3^{-j}\operatorname{sgn}(r_1)r_2\in 3\ZZ+\tilde q$, which proves the claim.
\end{rem}

The potentials $U$ and $V$ in \eqref{eq:original2} and \eqref{eq:off-diag} are smooth functions that satisfy \eqref{eq:UaV}--\eqref{eq:onlyV}. A characterization of such functions $U$ and $V$ is given in the following proposition. A straightforward proof (which we omit for brevity) can be obtained by using the fact that a function $u\in L^2(\CC/\Gamma;\CC)$ can be expanded in a Fourier series
with respect to the orthonormal basis 
\begin{equation}\label{eq:ONbasis}
e_{k}(z):=e^{\frac{i}{2} (z\bar k+\bar zk)}/\mathrm{Vol}(\mathbb C/\Gamma)^{\frac12},\quad k \in\Gamma^*,
\end{equation}
namely,
\begin{equation*}
u(z)=\sum_{k\in\Gamma^*} c_k e_k(z),\quad
c_k=\int_{\CC/\Gamma}u(z)\overline{e_k(z)}\,dA(z),
\end{equation*}
where $dA$ is the Haar measure on the torus $\mathbb C/\Gamma$.

\begin{prop}\label{prop:1}
Let $a=\frac{4\pi i}{3}(\omega a_1+\omega^2 a_2)\in\Gamma_3$. Let $j,\ell\in\{-1,0,1\}$. Then $u\in L^2(\C/\Gamma)$ satisfies 
\begin{equation}\label{eq:translation_coeffcond}
u(z+a)=\omega^{j(a_1+a_2)} u(z),\quad u(\omega z)=\bar\omega^\ell u(z),
\end{equation}
if and only if
\begin{equation}
\label{eq:potentialU}
u(z)=\sum_{n,m\in\Z} c_{nm}e^{\frac{i}2(-\sqrt 3(n+m)\re z-(2j+3(n-m))\im z )}
\end{equation}
where
$$
c_{nm}=(u,\phi_k)_{L^2(\CC/\Gamma)},\quad k={\tfrac1{\sqrt 3}(\omega^2(j+3n)-\omega(j-3m))},\quad n,m\in\ZZ,
$$
satisfies $c_{nm}=\bar\omega^{\ell} c_{(-m)(n-m+j)}=\bar\omega^{2\ell} c_{(m-n-j)(-n)}$ for $n,m\in\Z$. If in addition $\overline{u(\bar z)}=u(z)$ then
$\overline{c_{nm}}=c_{(-m)(-n)}
=\bar\omega^\ell c_{n(n-m+j)}=\bar\omega^{2\ell} c_{(m-n-j)m}$
for $n,m\in\Z$, and if also $u(\bar z)=u(-z)$ (which is only compatible with $\ell=0$) then $c_{nm}$ is real for all $n,m$.
\end{prop}

The conditions on $U$ in \eqref{eq:UaV} correspond to taking $j=\ell=-1$ in Proposition \ref{prop:1}.
We illustrate our results for $U$ obtained by choosing $c_{00}=1$, which then requires $c_{10}=\omega$ and $c_{0(-1)}=\omega^2$, and setting all other coefficients to 0. Noting that $z-\bar z=2i\im z$, while $z\bar \omega-\bar z\omega=-i\sqrt 3\re z-i\im z$ and $z\bar \omega^2-\bar z\omega^2=i\sqrt 3\re z-i\im z$, this potential is
given by \eqref{eq:standard_pot}.
More generally,
$$
f_{3n+1}(z)=\sum_{j=0}^2\omega^j e^{\frac{3n+1}2(z\bar\omega^j-\bar z\omega^j)}
$$
also satisfies \eqref{eq:translation_coeffcond} with $j=\ell=-1$, and so do the linear combinations
$U(z)=\sum_{n\in\Z} b_{3n+1}f_{3n+1}(z)$ that appeared in \cite{BEWZ22}.
Observe however that not all potentials in \eqref{eq:potentialU} for $j=\ell=-1$ are of this form.

\subsection{Symmetries of the full Hamiltonian}
\label{ss:symmetries}
We will now discuss the symmetries of the twisted trilayer graphene Hamiltonian \eqref{eq:original}.

First define the following twisted translations
\begin{equation*}
\mathscr L_{a}u(z) =  \operatorname{diag}(\omega^{p(a_1+a_2)},1,\bar\omega^{q(a_1+a_2)})u(z+a),\ \ a = \tfrac{4}{3}\pi i (\omega a_1+ \omega^2 a_2)\in \Gamma_3.    
\end{equation*}
We then have $[ D(\alpha), \mathscr L_{a}] =0$ for $D(\alpha)$ defined in \eqref{eq:original2}.
Consequently, for $u\in L^2_{\loc}(\CC; \CC^6)$, the Hamiltonian commutes with two symmetries
\[\begin{split}
 \mathscr Cu(z) &:= \operatorname{diag}(1,1,1,\bar \omega,\bar\omega,\bar \omega) u(\omega z) \\
  \mathscr L_au(z) &:= \operatorname{diag}(\omega^{p(a_1+a_2)},1,\bar\omega^{q(a_1+a_2)},\omega^{p(a_1+a_2)},1,\bar \omega^{q(a_1+a_2)}) u( z + a ),
 \end{split}\]
such that $\mathscr C \mathscr L_a = \mathscr L_{\bar \omega a} \mathscr C$. We slightly abuse notation by using $\mathscr L_a$ and $\mathscr C$ on $L^2(\CC;\CC^6)$ as well as on $L^2(\CC;\CC^3)$; for $\mathscr C$ this means identifying $u\in L^2(\CC;\CC^3)$ with $(u,0_{\CC^3})\in L^2(\CC;\CC^6)$.

We introduce
\begin{equation*}
L^2_k:=\{ u \in L^2_{\loc}(\CC;\CC^d): \mathscr L_a u = e^{i\langle a,k \rangle}  u , \ a \in \Gamma_3\},\ \ k \in \CC,
\end{equation*}
where $\langle a,k \rangle=\tfrac{1}2(a\bar k+\bar a k)$. Note that the spaces $L^2_k$ only depend on the equivalence class of $k\in \CC/\Gamma_3^*$, 
where the dual lattice $\Gamma_3^*$ is given by
\begin{equation}
\label{eq:dual_lattice}
\Gamma_3^*=\{{\sqrt 3}(\omega^2 k_1-\omega k_2), \  k_1,k_2\in\ZZ\}.
\end{equation}
We define a map between the spaces $\{L^2_k\}_{k\in\CC/\Gamma_3^*}$:
\begin{equation}
    \label{eq:tau}
    \tau(k): L^2_{k'} \to L^2_{k'+k}, \quad \tau(k)u(z) = e^{i\langle z,k\rangle}u(z).
\end{equation}

Among the spaces $\{L^2_k:k\in\CC\}$, we focus on three distinguished $k\in \mathcal{K} := \{\tfrac{r}{\sqrt 3}(\omega^2-\omega), r\in\ZZ_3\} = \{0,\pm i\} \subset\CC/\Gamma_3^*$, as they are the only possible locations of protected states of the Hamiltonian (cf.~Section \ref{sec:DiracPts}). For $k\in \mathcal{K}$, in view of the equation
\begin{equation*}
e^{i\langle a,k \rangle}=\omega^{r(a_1+a_2)}, \ \ k=\tfrac{r}{\sqrt 3}(\omega^2-\omega),\ \ r\in\ZZ_3,\ \ a\in\Gamma_3,
\end{equation*}
we define the subspaces of $L^2_k$ under the action of $\mathscr C$ which correspond to irreducible representations of the symmetry group (see \cite[Section 2.2]{BEWZ22}),
\[
   L^2_{{r,\ell}}=\{u\in L^2_{\loc}(\CC;\CC^3): \mathscr L_a u = \omega^{r(a_1+a_2)}u, \ \ \mathscr C u = \bar \omega^\ell u, \ \ a\in\Gamma_3\}, \ \ r,\ell \in \ZZ_3.
\]
Then for $k\in \mathcal{K}$, we have (see \cite[Section 2.1]{bhz2})
\begin{gather}
    \label{eq:decomp1}
L^2_k=\bigoplus_{\ell\in\ZZ_3} L^2_{{r,\ell}}, \ \ k=\tfrac{r}{\sqrt 3}(\omega^2-\omega),\ \ r\in\ZZ_3.
\end{gather}
\begin{rem}\label{r:correspondence}
    A simple computation gives 
    \begin{equation}\label{eq:distinguished} 
    \mathcal{K} = \{-ir: r\in\ZZ_3\},
    \end{equation}
    thus the correspondence \eqref{eq:decomp1} can be reformulated as
    \begin{equation}\label{eq:decomp2}
    L^2_{-ir} = \bigoplus_{\ell\in\ZZ_3} L^2_{{r,\ell}}, \ \ r\in\ZZ_3.\end{equation}
    Indeed, the points in $\mathcal{K}$ are fixed under rotation by $\omega$ modulo $\Gamma_3^*$. These three distinguished points in $\CC/\Gamma_3^*$ are the locations of protected states in Figure \ref{fig:protected_cones}\footnote{The figure is shown over a fundamental domain of $\CC/\Gamma_3^*$ represented in rectangular coordinates $y=(y_1,y_2)$ defined via $z=2i(\omega y_1+\omega^2 y_2)$, where the distinguished points $i$, $0$, and $-i$ correspond to $(-1,-1)$, $(0,0)$, and $(1,1)$ respectively.}.
\end{rem}

Define the involutive $\mathcal P \mathcal T$ symmetry and the involutive mirror symmetry $\mathscr M$ by
\begin{equation*}
(\mathcal P \mathcal Tu)(z) = \begin{pmatrix} 0 & \operatorname{id} \\ \operatorname{id} & 0 \end{pmatrix} \overline{u(-z)},\quad
\mathscr Mu(z)= \begin{pmatrix} 0 & \mathscr E \\  \mathscr E & 0 \end{pmatrix} u(-\bar z),
\end{equation*} 
with $\mathscr E = \operatorname{diag}(1,-1,1)$. 
Using the symmetry properties \eqref{eq:UaV}--\eqref{eq:onlyV} of $U$ and $V$ it is then easy to check (assuming $\tilde\alpha\in\RR^2$) that
\[  (\mathcal P \mathcal T) H(\alpha,\tilde \alpha) (\mathcal P \mathcal T) = H(\alpha,\tilde \alpha), \quad (\mathcal P \mathcal T)  \mathscr C = \omega \mathscr C (\mathcal P \mathcal T) \quad\text{and}\quad(\mathcal P \mathcal T) \mathscr L_a = \mathscr L_{-a} (\mathcal P \mathcal T)   \]
and
\[  \mathscr M H(\alpha,\tilde \alpha)\mathscr M = - H(\alpha,\tilde \alpha),\quad \mathscr M \mathscr C = \bar \omega \mathscr C^{-1} \mathscr M, \quad\text{and}\quad \mathscr M \mathscr L_a  =\mathscr L_{-\bar a} \mathscr M. \]
Then the joint application of mirror and $\mathcal P \mathcal T$ symmetry leave $L^2_{j,\ell}$ spaces invariant:
$$L^2_{j,\ell}\xrightarrow[]{{\mathscr M}}  L^2_{j,-\ell+1} \xrightarrow[]{{\mathcal P\mathcal T}} L^2_{j,\ell}.$$
Indeed, if $u\in L^2_{j,\ell}$ then $\mathscr Mu\in L^2_{j,-\ell+1}$ since
$$
\mathscr L_a\mathscr C (\mathscr M u)=
\bar\omega \mathscr M(\mathscr L_{-\bar a}\mathscr C^{-1}u)=\omega^{j(a_1+a_2)}\bar\omega^{-\ell+1} \mathscr Mu, 
$$ 
and if $v\in L^2_{j,-\ell+1}$ then $\mathcal P\mathcal Tv\in L^2_{j,\ell}$ since
$$
\mathscr L_a\mathscr C (\mathcal P\mathcal Tv)=
\bar\omega\mathcal P\mathcal T(\mathscr L_{-a}\mathscr Cv)=
\bar\omega\mathcal P\mathcal T(\omega^{-j(a_1+a_2)}\bar\omega^{-\ell+1}v)=
\omega^{j(a_1+a_2)}\bar\omega^{\ell}\mathcal P\mathcal Tv.
$$
In particular, defining the map 
\begin{equation}
\label{eq:mapQ}
Q: L_{\loc}^2(\CC;\CC^3) \to L_{\loc}^2(\CC;\CC^3)\text{ with }Qv(z) = \overline{v(-z)},
\end{equation} 
one can check that $QD(\alpha)Q = D(\alpha)^*$ and $Q:L^2_{j,\ell} \to L^2_{j,-\ell}$.
The symmetries can then give the symmetry of the spectrum of $H$ in each representation.
\begin{prop}\label{prop:chiralsymrestriction}
The densely defined operator $H(\alpha,\tilde \alpha):L^2_{j,\ell}\to L^2_{j,\ell}$, $j,\ell\in\ZZ_3$ satisfies 
$$
\Spec_{L^2_{j,\ell}} H(\alpha,\tilde \alpha)=-\Spec_{L^2_{j,\ell}} H(\alpha,\tilde \alpha).
$$
\end{prop}

\begin{proof}
This follows from conjugating $H(\alpha,\tilde \alpha)$ by $\mathscr M (\mathcal P \mathcal T).$
\end{proof}

\subsection{Floquet theory}\label{ss:Floquet}
Since $U$ is periodic with respect to $\Gamma,$ we can also perform the Bloch-Floquet transform on $H$ in \eqref{eq:original} with respect to $\Gamma$ and standard $\Gamma$-lattice translations. 
This way, we obtain the same fibre operators \eqref{eq:Floquet} but with periodic boundary conditions on $\CC/\Gamma.$
Thus, we have the equality of sets
\begin{equation}
\label{eq:nine}
\Spec_{L^2(\CC)}H(\alpha,\widetilde \alpha)=\bigcup_{k\in\CC/\Gamma_3^*}\Spec_{L^2_0}H_k(\alpha,\widetilde \alpha)=\bigcup_{k\in\CC/\Gamma^*}\Spec_{L^2(\CC/\Gamma)}H_k(\alpha,\widetilde \alpha).
\end{equation}
However, working on $\Gamma$ rather than $\Gamma_3$ has the disadvantage that, due to the larger lattice, it artificially inflates the number of bands by a factor $9$, since $\lvert \CC/\Gamma\rvert =  9\lvert \CC/\Gamma_3\rvert$. This is thoroughly explained in \cite[(2.6)]{bhz2}, see also \cite[Figure $4$]{bhz2}.
We shall thus mostly refrain from working on $\Gamma$ when referring to flat bands.  One thing we do want to mention is that one can argue analogously to \eqref{eq:magic_small} and conclude that
\begin{equation}
\label{eq:def:magicalphas2}
\alpha\in\mathcal A\quad\Longleftrightarrow\quad \Spec_{L^2(\CC/\Gamma)} (D(\alpha))=\CC.
\end{equation}

\section{Protected states and a Birman-Schwinger principle}
\label{sec:DiracPts}

In this section, we consider the Hamiltonian $H(\alpha)$ in the chiral limit $\tilde \alpha=0$ satisfying equation \eqref{eq:chiralsymmetry} on $L^2(\CC/\Gamma)$.  
The main purpose is to prove the existence of protected states at zero energy (see Figure \ref{fig:protected_cones}) in Proposition \ref{prop:uialpha}, and to establish a spectral characterisation in Theorem \ref{Bk} of the set of magic parameters $\mathcal A$ defined in \eqref{def:magicalphas}, as well as a characterisation of $\mathcal A$ in terms of the Wronskian of protected states, see Theorem \ref{theo:wronskian_intro}.

\subsection{Protected states at zero energy}
For $\alpha_{12}=\alpha_{23}=0$, we have $\operatorname{ker}_{L^2(\CC/\Gamma)}H(0)=\operatorname{Span}\{e_1,\ldots,e_6\}$ by Liouville's theorem, where $\{e_j\}$ is the standard basis on $\CC^6$. Moreover, for $a\in\Gamma_3$,
\[\begin{split} 
\mathscr L_{a} e_1 &= \omega^{p(a_1+a_2)} e_1 , \quad  \mathscr L_{a} e_2 = e_2 , \quad  \mathscr L_{a} e_3 =\bar \omega^{q(a_1+a_2)} e_3, \\
\mathscr L_{a} e_4 &= \omega^{p(a_1+a_2)} e_4 , \quad  \mathscr L_{a} e_5 = e_5 , \quad  \mathscr L_{a} e_6 =\bar \omega^{q(a_1+a_2)} e_6, \\
 \mathscr C e_1 &= e_1,\quad   \mathscr C e_2 = e_2, \quad  \mathscr Ce_3 = e_3, \quad
\mathscr C e_4 = \bar \omega e_4 , \quad  \mathscr C e_5 =  \bar \omega e_5 , \quad  \mathscr C e_6 = \bar \omega e_6,
\end{split}\] 
which implies that
\begin{equation*}
\begin{split}
    e_1\in L^2_{p,0},\ e_2\in L^2_{0,0},\ e_3\in L^2_{-q,0},\ 
    e_4\in L^2_{p,1},\ e_5\in L^2_{0,1}, \ e_6\in L^2_{-q,1}.
\end{split}
\end{equation*}
By Assumption \ref{ass:angles}, there is at least one element in $\{p,-q,0\}$ which is distinct from the others. Hence, 
there is an $n\in\{p,-q,0\}$ such that 
$$
\dim\operatorname{ker}_{L^2_{n,0}}H(0,0)=\dim\operatorname{ker}_{L^2_{n,1}}H(0,0)=1.
$$
It follows from Proposition \ref{prop:chiralsymrestriction} that
\begin{equation}\label{eq:kerHalphanotempty}
\operatorname{ker}_{L^2_{n,0}}H(\alpha,\widetilde \alpha) \neq \emptyset,\ \operatorname{ker}_{L^2_{n,1}}H(\alpha,\widetilde \alpha) \neq \emptyset,
\end{equation}
for all $\alpha\in\CC^2$, $\widetilde \alpha \in \RR^2$ as $H(\alpha)$ depends analytically on $\alpha_{12}, \alpha_{23}\in\CC$ so that the one-dimensional eigenspace in $\ker_{L^2_{n,0}}H(\alpha,\widetilde \alpha)$ and $\ker_{L^2_{n,1}}H(\alpha,\widetilde \alpha)$ cannot split for $\alpha \neq 0 $ or $\widetilde \alpha \neq 0$. 
In the chiral limit, $\widetilde \alpha=0$, since $\operatorname{ker}_{L^2_{r,\ell}}H(\alpha)=\operatorname{ker}_{L^2_{r,\ell}}D(\alpha)\oplus\{0_{\CC^3}\}+\{0_{\CC^3}\}\oplus \operatorname{ker}_{L^2_{r,\ell}}D(\alpha)^*$ we find that for this $n$ we have for all $\alpha\in\CC^2$ that
\begin{equation}\label{eq:kerDalphanotempty}
\operatorname{ker}_{L^2_{n,0}}D(\alpha)\ne\{0\}.
\end{equation}
Note also that \eqref{eq:kerHalphanotempty} implies the following:
\begin{corr}
\label{corr:simple_protected_state}
    For all $\alpha \in \mathbb C^2, \widetilde \alpha \in \mathbb R^2$, $0 \in \Spec_{L^2(\CC/\Gamma)}H(\alpha,\widetilde \alpha)$. In particular, $0 \in \Spec_{L^2(\CC/\Gamma)}H(\alpha).$
\end{corr} 

We now focus on the chiral limit. The following proposition gives the existence of three protected states in the corresponding representation spaces for all $\alpha\in \CC^2$.
\begin{prop}[Existence of protected states]
\label{prop:uialpha}
For any $\alpha\in \CC^2$ and $p,q$ as in Assumption \ref{ass:angles} with $p\not\equiv 0 \bmod 3$, we have three protected states with their location differing:
\begin{enumerate}[label=\text{Case \Roman*:},ref=\Roman*]
    \item 
        If $-q,p,0$ are not mutually different numbers in $\mathbb Z_3$, i.e., $\{-q,p,0\} = \{p,0\}$ as a set, then there exists $m\in\{p,0\}$ and $m\neq n\in\{p,0\}$ such that 
    \begin{equation*}
        \dim\ker_{L^2_{m,0}}D(\alpha) \geq 2, \ \ 
        \ker_{L^2_{n,0}}D(\alpha) \neq \{0\};
    \end{equation*}
    \item 
        If $-q,p,0$ are mutually different numbers in $\mathbb Z_3,$ then
    \begin{equation*}
        \ker_{L^2_{p,0}}D(\alpha)\neq \{0\}, \ \ \ker_{L^2_{0,0}}D(\alpha)\neq \{0\}, \ \ 
        \ker_{L^2_{-q,0}}D(\alpha)\neq \{0\}.
    \end{equation*}
\end{enumerate}
\end{prop}
\begin{proof}
Note that $H(\alpha)$ and $D(\alpha)$ depend analytically in $\alpha_{12}, \alpha_{23}\in \CC$. In Case \ref{case1}, the decomposition into representations $L^2_{r,\ell},\ r= p,0,-q$ and $\ell=0,1$ does not distinguish the elements in the nullspace, i.e., there are $n \neq m, \text{ where } m,n\in \ZZ_3$ such that
\[ \operatorname{dim}\ker_{L^2_{n,0}}H(0)=1 \text{ and } \operatorname{dim}\ker_{L^2_{m,0}}H(0)=2.\]
In view of \eqref{eq:kerDalphanotempty} we only have to show that $\dim\ker_{L^2_{m,0}}D(\alpha)\ge2$.

For $\alpha\notin\mathcal{A}$, by equation \eqref{eq:def:magicalphas2}, there exists $k\in\CC$ such that $(D(\alpha)-k)^{-1}: L^2(\CC/\Gamma)\to L^2(\CC/\Gamma)$ is compact so that $\Spec_{L^2(\CC/\Gamma)} D(\alpha)$ is discrete. The rotational symmetry 
\[D(\alpha)\mathscr C = \bar \omega \mathscr{C}D(\alpha)\]
yields that 
\[(D(\alpha)-k)\mathscr C u = \bar \omega \mathscr{C}(D(\alpha)-\omega k)u.\]
This means that $\Spec_{L^2(\CC/\Gamma)}D(\alpha)$ is invariant under multiplication by $\omega$. As $\dim \ker_{L^2(\CC/\Gamma)} D(0) = 3$ and one of the elements in $\ker_{L^2(\CC/\Gamma)} D(\alpha)$ is protected (as it corresponds to the protected state in $\ker_{L^2_{n,0}}D(\alpha)$), the $\omega$-rotational symmetry of $\Spec_{L^2(\CC/\Gamma)}D(\alpha)$ implies that $\dim \ker_{L^2_{m,0}(\CC/\Gamma)}D(\alpha)=2$ for $\alpha\notin \mathcal{A}$. Therefore, since $\mathcal A $ is discrete, we have $ \dim\ker_{L^2_{m,0}}D(\alpha) \geq 2$ for $\alpha\in\CC^2$.

For Case \ref{case2}, the translation operator $\mathscr L_a$ splits all functions $e_1,e_2,e_3$ in three different representation spaces $L^2_{p,0}, L^2_{0,0}, L^2_{-q,0}$, which gives 
$$\dim\ker_{L^2_{p,0}} H(0) = \dim\ker_{L^2_{0,0}} H(0) =\dim\ker_{L^2_{-q,0}} H(0) =1.$$
They correspond to three distinct points in Figure \ref{fig:protected_cones}. Since they are simple and the spectrum of the Hamiltonian $\Spec_{L^2_{r,0}}H(\alpha)$ is even for $r = p,0,-q$ due to Proposition \ref{prop:chiralsymrestriction}, the analyticity of $H(\alpha)$ in $\alpha_{12}, \alpha_{23}\in\CC$ yields that the eigenvalues are protected for all $\alpha\in \CC^2$. 
\end{proof}

\subsection{Spectral characterisation of magic parameters}\label{ss:spectral_char}
Following the strategy of \cite{BEWZ22}, we now give a spectral characterisation of the set of magic parameters using a compact operator $B_k$, which we will refer as the Birman-Schwinger operator.
For $\alpha=(\alpha_{12},\alpha_{23})$ with $\alpha_{12}\ne 0$, we shall for the moment consider the ratio $\alpha_{23}/\alpha_{12}$ to be fixed but arbitrary, and view 
\[ D(\alpha) = \begin{pmatrix} 2 D_{\bar z} &  \alpha_{12} U( pz ) & 0 \\ 
  \alpha_{12} U(- pz) &2 D_{\bar z} & \alpha_{12}\tfrac{\alpha_{23}}{\alpha_{12}} U(\tfrac{p\zeta_2}{\zeta_1} z)\\ 
0 &   \alpha_{12}\tfrac{\alpha_{23}}{\alpha_{12}} U(-\tfrac{p\zeta_2}{\zeta_1} z) &2 D_{\bar z} \end{pmatrix}\]
as a family of operators depending on a single complex parameter $\alpha_{12}\in\CC$.
The non-self-adjoint 
operator family $\CC\ni \alpha_{12}\mapsto   D(\alpha):H^1(\CC/\Gamma)\to L^2(\CC/\Gamma)$ is a holomorphic family of elliptic Fredholm operators of index 0.

Let $H_{k}(\alpha)$ be as in \eqref{eq:Floquet}, so that
\[ H_{k}(\alpha) = \begin{pmatrix} 0 & D(\alpha)^*+\bar k \\ D(\alpha)+k & 0 \end{pmatrix} : H^1(\CC/\Gamma;\CC^6) \to L^2(\CC/\Gamma;\CC^6). 
\]
Note that, with $e_k\in L^2(\CC/\Gamma;\CC)$ as in \eqref{eq:ONbasis} for $k\in\Gamma^*$, we have $2D_{\bar z} e_k=k e_k$, which implies that
$$
\Spec_{L^2(\CC/\Gamma)}(2D_{\bar z})=\Gamma^*.
$$
Therefore, we consider the identity
\begin{equation}
\label{eq:bs-id}
    D(\alpha)+k = (2D_{\bar z}+k)(\mathrm{Id}+\alpha_{12}T_k) \text{ for } k\notin \Gamma^*,
\end{equation}
where we have introduced the operator
\begin{equation}\label{eq:Tk}
    T_k =
    \begin{pmatrix} 
    0 &   R(k)U( pz ) & 0 \\ 
    R(k) U(- pz) & 0 &\frac{\alpha_{23}}{\alpha_{12}} R(k) U(p\tfrac{\zeta_2}{\zeta_1} z)\\
    0 & \frac{\alpha_{23}}{\alpha_{12}} R(k) U(-p\tfrac{\zeta_2}{\zeta_1} z) & 0
   \end{pmatrix}
\end{equation}
with $R(k) = (2D_{\bar z}+k)^{-1}.$
We conclude that $0 \in \bigcap_{k \in \CC} \Spec_{L^2(\CC/\Gamma)} H_k(\alpha)$ with $\alpha_{12} \in \mathbb C\setminus \{0\}$, $\alpha_{23} \in \mathbb C$ if and only if 
\[ -\alpha_{12}^{-1} \in \Spec_{L^2(\CC/\Gamma)}T_k\]
for some $k \notin \Gamma^*$.
Introduce the set
\[\mathcal{A}_k=\mathcal A_k(\tfrac{\alpha_{23}}{\alpha_{12}}) := 1/ (\Spec_{L^2(\CC/\Gamma)}T_k\backslash\{0\}) \text{ for } k\notin\Gamma^*.\]
It can be shown as in \cite[Proposition 3.1]{BEWZ22} (see also \cite[Proposition 2.2]{bhz2} for a more detailed argument) that the set $\mathcal A_k$ is independent of $k\notin\Gamma^*$, and that for any fixed $\tfrac{\alpha_{23}}{\alpha_{12}}$ we have 
\begin{equation}
  \label{eq:SpecD}
  \Spec_{ L^2 ( \CC / \Gamma ) }  ( D ( \alpha ) ) = 
\begin{cases} 
  \Gamma^*, & \alpha \notin \mathcal A, \\
  \CC, & \alpha \in \mathcal A , \end{cases}
\end{equation}
where $\mathcal{A}:=\{(\alpha_{12},\alpha_{23})\in\CC^2:\ \alpha_{12}\in \mathcal{A}(\tfrac{\alpha_{23}}{\alpha_{12}})\}.$ Note that this definition agrees with the definition of $\mathcal A$ in \eqref{def:magicalphas} in view of \eqref{eq:def:magicalphas2}.

\begin{rem}\label{rem:rescaling}
    If we go back to the rescaling done for \eqref{eq:off-diag}, (unrescaled) magic parameters are defined as
\begin{equation*}
\mathcal B:=\{\beta\in \mathbb C^2:  \Spec_{ L^2 ( \CC / \Gamma ) }  ( {\mathcal D} ( \beta ,\zeta) ) =\CC\}.
\end{equation*}
We note that the condition $\Spec_{ L^2 ( \CC / \Gamma ) }  ( D ( \alpha ) ) =\CC$ is invariant under rescaling of the $z$ variable. This means that the set $\mathcal B$ is obtained by a rescaling of the set $\mathcal A$, namely
\begin{equation*}
\mathcal B=\frac{\zeta_1}{p}\mathcal A.
\end{equation*}
In this article, we will for simplicity study the set of rescaled magic parameters $\mathcal A$, with the exception of section \S \ref{sec:continuity}, where we will investigate the continuity of the set of \emph{unrescaled} magic parameters with respect to the twisting angles $\zeta_1,\zeta_2$.
\end{rem}

The following proposition gives symmetries of the set $\mathcal A$ of magic $\alpha$'s.
\begin{prop}
We have
$$
\Spec D(\alpha)=\Spec D(-\alpha)=\Spec D(\bar\alpha)\quad \text{ and hence }\quad\mathcal A=-\mathcal A=\overline{\mathcal A}.
$$
\end{prop}

\begin{proof}
Conjugating $D(\alpha)$ by $\operatorname{diag}(1,-1,1)$ shows that $\Spec D(\alpha)=\Spec D(-\alpha)$.

Applying the map $Fv(z)=\overline{v(\bar z)}$ we find that $(D_{\bar z}Fv)(z)=(D_z\bar v)(\bar z)=-(\overline{D_{\bar z} v})(\bar z)=-(FD_{\bar z}v)(z)$ and
$$
D(\alpha)(Fv)=-F( D(-\bar\alpha)v)
$$
which implies that $\Spec(D(\alpha)) = -\overline{\Spec( D(-\bar\alpha))}$. By \eqref{eq:SpecD} we see that $\Spec D(\alpha)$ is invariant under additive inversion and complex conjugation, which gives
\[ \Spec(D(\alpha)) =  -\overline{\Spec( D(-\bar\alpha))}= \Spec( D(-\bar\alpha))=\Spec(D(\bar\alpha)) \]
where the last identity follows from the first part of the proof.
\end{proof}

For fixed $\frac{\alpha_{23}}{\alpha_{12}}$, we define the \emph{multiplicity} of $\alpha\in\mathcal{A}$ as
\[m(\alpha): = \min_{k\notin\Gamma^*}\{ m(\alpha^{-1}_{12}): \alpha^{-1}_{12}\in \Spec_{L^2_0}T_k\}.\]
In particular, when $\alpha^{-1}_{12}$ is a simple eigenvalue of $T_k$, we call $\alpha\in \mathcal{A}$ a \emph{simple} magic parameter. From the identity \eqref{eq:bs-id}, we obtain
\begin{prop}
\label{prop: multiplicity}
In the notation of \eqref{eq:specHk}, the following statements are equivalent:
\begin{enumerate}
    \item The magic parameter $\alpha\in\mathcal{A}$ has multiplicity $m$.
    \item $\min_{k\in\CC/\Gamma_3^*} \{\dim \ker_{L_0^2}(D(\alpha)+k)\} = m$.
    \item $E_{m}(k,\alpha) = 0$ for any $k\in\CC$ and $E_{m+1}(k_0,\alpha) \neq 0$ for some $k_0\in\CC$.
\end{enumerate}    
\end{prop}

We will now establish a spectral characterization of the set $\mathcal A$ of magic $\alpha$'s in terms of a scalar operator. Let $T_k$ be the operator in \eqref{eq:Tk} and note that
\begin{equation}
    \label{eq:tristan}
    \begin{split} T_k^2 &=A_k \oplus B_k \text{ with } B_k=(T_k)_{21} (T_k)_{12}+ (T_k)_{23}(T_k)_{32}\text{ and }\\
A_k& = \begin{pmatrix}(T_k)_{12}(T_k)_{21} & (T_k)_{12}(T_k)_{23} \\ (T_k)_{32}(T_k)_{21} & (T_k)_{32}(T_k)_{23} \end{pmatrix}  = \begin{pmatrix} (T_k)_{12} \\ (T_k)_{32} \end{pmatrix}  \begin{pmatrix} (T_k)_{21} & (T_k)_{23} \end{pmatrix} =:S^{(1)}_k  S^{(2)}_k.
\end{split}
\end{equation}
We observe that $B_k=S^{(2)}_kS^{(1)}_k$. Let $(\lambda,\varphi)$ be an eigenpair of $B_k$. Then  $(\sqrt{\lambda},((T_k)_{12}\varphi, \sqrt{\lambda}\varphi, (T_k)_{32} \varphi)^t)$  is an eigenpair of $T_k$. That this is no coincidence, follows from the Schur complement formula, which shows for $\lambda \neq 0$ that
\begin{align*} \begin{pmatrix} 1 & S^{(1)}_k/\lambda \\ S^{(2)}_k/\lambda & 1 \end{pmatrix} &= \begin{pmatrix} 1 & 0 \\ S^{(2)}_k/\lambda & 1 \end{pmatrix}  \begin{pmatrix} 1 & 0 \\ 0 & 1- S^{(2)}_k  S^{(1)}_k/\lambda^2\end{pmatrix} \begin{pmatrix} 1 & S^{(1)}_k/\lambda \\ 0 & 1 \end{pmatrix}\\ 
&=\begin{pmatrix} 1 & S^{(1)}_k/\lambda \\  0 & 1 \end{pmatrix}  \begin{pmatrix} 1- S^{(1)}_k  S^{(2)}_k/\lambda^2 & 0 \\ 0 & 1 \end{pmatrix} \begin{pmatrix} 1 & 0 \\ S^{(2)}_k/\lambda & 1 \end{pmatrix}.\end{align*}
This shows that $\Spec(A_k)\setminus\{0\} = \Spec(S^{(1)}_k  S^{(2)}_k)\setminus\{0\} = \Spec(S^{(2)}_kS^{(1)}_k )\setminus\{0\} = \Spec(B_k)\setminus\{0\}$ such that $\Spec(T_k)\setminus\{0\} =\pm \sqrt{\Spec(B_k)}\setminus\{0\}.$

Therefore, we have a spectral characterization of the set of magic parameters in terms of the Hilbert-Schmidt operator $B_k$.
\begin{theo}
\label{Bk}
   Let $\alpha \in \CC^2 \setminus  (\{0\} \times \CC)$, then with $B_k$ defined in \eqref{eq:Bk}, $\alpha$ is magic if and only if $\frac{1}{\alpha_{12}^2} \in \Spec(B_k(\tfrac{\alpha_{23}}{\alpha_{12}}))$ for $k \notin \Gamma^*.$
\end{theo}

Finally, we observe that
\[ \Spec_{L^2_{r}}(T_{k}) = \Spec_{L^2_{s}}(T_{k}),\quad k \notin \Gamma^*.\]
To see this, we note that by equation \eqref{eq:tau},
\begin{equation}
\label{eq:UUU}
\tau({s-r}): L^2_{r} \to L^2_{s}, 
\quad 
\tau(s-r)^* T_{k} \tau(s-r) = T_{k+s-r}.
\end{equation}
Hence, since the spectrum of $T_{k}$ is independent of $k$, it follows that if there exists a flat band, i.e., $(D(\alpha)+k)u_{k}=0$ with $0\ne u_{k} \in L^2_{ r}$ for all $k\in \mathbb C$, then there are also $0\ne v_{k} \in L^2_{s}$ such that $(D(\alpha)+k)v_{k}=0$ for all $k$.
\subsection{Theta function argument}
\label{ss:theta_arg}
To construct the entire flat band from a single Bloch function with a zero, one can use the structure of $D(\alpha)$ and employ theta functions for this construction.
To this end, we let $\theta$ be the Jacobi theta function
$ \theta ( \zeta ) :=  \theta_{1} ( \zeta | \omega ) := - \theta_{\frac12,\frac12} ( \zeta | \omega ) $ given by
\begin{equation*}
\label{eq:theta_f}
\begin{gathered} 
\theta ( \zeta )
= - \sum_{ n \in \mathbb Z } \exp ( \pi i (n+\tfrac12) ^2 \omega+ 2 \pi i ( n + \tfrac12 ) (\zeta + \tfrac
12 )  ) , \ \ \ \theta ( - \zeta ) = - \theta ( \zeta ), \\
\theta  ( \zeta + m ) = (-1)^m \theta   ( \zeta ) , \ \ \theta  ( \zeta + n \omega ) = 
(-1)^n e^{ - \pi i n^2 \omega - 2 \pi i \zeta  n } \theta   ( \zeta ) ,
\end{gathered}
\end{equation*}
so that $ \theta $ has simple zeros at $ \ZZ + \omega \ZZ $ (and no other zeros) -- see \cite{tata}. 

We now define 
\begin{equation*}
F_k(z) :=e^{-\frac{i}2( \bar \omega z +  \bar z  ) k } \frac{ \theta ( \frac{3z}{4\pi i \omega} + \frac{k}{\sqrt{3}\omega} ) }{
\theta\Big( \frac{3z}{4\pi i \omega}\Big) }.
\end{equation*}
A computation shows that $F_k(z+a) = F_k(z)$ for $a\in \Gamma_3$.
We also define
\[ G_k(z):=e^{\frac{i}{2}(\bar k-k\bar\omega)z}\frac{ \theta ( \frac{3z}{4\pi i \omega} + \frac{k}{\sqrt{3}\omega} ) }{
\theta\Big( \frac{3z}{4\pi i \omega}\Big) } \]
satisfying $G_k(z+a) = e^{i\langle a,k \rangle}G_k(z)$ for $a \in \Gamma_3$. Note that in particular we have $\tau(k)F_k(z) = G_k(z)$, where $\tau$ is defined in equation \eqref{eq:tau}. If $u\in L^2_{k'}$ is a solution to $D(\alpha) u = 0$ with $u(z_{*})=0$ (cf.~Lemma \ref{l:van}), then for any $k\in\CC$ we can construct
\begin{equation}
\label{eq:theta_a}
  (D(\alpha)+k)(F_{k}(z-z_*)u(z))=0 \text{ and }  D(\alpha)(G_{k}(z-z_*)u(z))=0.
\end{equation}
In particular, multiplication by $F_k$ leaves $L^2_{k'}$ invariant and multiplication by $G_k$ maps $L^2_{k'}$ to $L^2_{k'+k}.$
\begin{rem}\label{rmk:theta_a}
    We shall refer to equation \eqref{eq:theta_a} as the \emph{theta function argument}.
\end{rem}

\subsection{Flat bands via Wronskians}
\label{sec:another}
Here we prove Theorem \ref{theo:wronskian_intro}, which gives an alternative characterisation of magic parameters $\mathcal{A}$ as the zero set of the Wronskian $W(\alpha)$ in \eqref{eq:det} of the protected states.

We devote the rest of this section to proving the theorem. First note that the Wronskian $W(\alpha)$ is constant in $z$, as it is periodic and satisfies equation $D_{\bar{z}} W=0$. 
Thus we can consider the Wronskian $W(\alpha)$ at distinguished points which are invariant under rotation by $\omega$ modulo $\Gamma_3$:
\begin{equation*}
z_S := \tfrac13(\gamma_2-\gamma_1)=\tfrac{4\sqrt 3}9,\quad \gamma_j=\tfrac{4\pi i}3\omega^j\in\Gamma_3, \ \ j=1,2.
\end{equation*}
In fact, $\omega (\pm z_S)=\pm z_S \mp\gamma_2$ shows that $\{0,\pm z_S\}$ is fixed under rotation by $\omega$ as a set of elements in $\CC/\Gamma_3$, and we will refer to them as \emph{stacking points}.
These are precisely the points in $\CC/\Gamma_3$ corresponding to the distinguished points  $\mathcal K=\{-ir:r\in\ZZ_3\}\subset \CC/\Gamma_3^*$ given in \eqref{eq:distinguished}. Indeed, $\Gamma_3^*=\sqrt 3(\omega\ZZ\oplus\omega^2\ZZ)$, so $\CC/\Gamma_3^*$ is obtained from $\CC/\Gamma_3$ through multiplication by $3\sqrt 3/(4\pi i)$. Hence, $\pm z_S$ corresponds to
$$
\pm z_S\cdot\frac{3\sqrt 3}{4\pi i}=\mp i\in \CC/\Gamma_3^*.
$$

Now we consider $u\in L^2_{r,0}$ for some $r\in\ZZ_3$, then
\begin{equation}
    \label{eq:vanishing}
    \begin{split}
u(\pm z_S) &= \mathscr C u(\pm z_S) = u(\pm \omega z_S) = u(\pm z_S \mp \gamma_2)\\
&=\operatorname{diag}(\omega^{\pm p},1,\omega^{\mp q}) \mathscr L_{\mp \gamma_2} u(\pm z_S) = \operatorname{diag}(\omega^{\mp r \pm p },\omega^{\mp r},\omega^{\mp r \mp q}) u(\pm z_S).
\end{split} 
\end{equation}
In Case \ref{case1}, 
\begin{enumerate}
    \item if $p\equiv -q$, where we have $\varphi, \rho\in \ker_{L^2_{p,0}}D(\alpha)$, $\psi\in\ker_{L^2_{0,0}}D(\alpha)$, using \eqref{eq:vanishing} we obtain 
$$\varphi_2 (\pm z_S)=0,\ \rho_2 (\pm z_S) =0, \ \psi_1(\pm z_S)= \psi_3(\pm z_S)=0, $$
which yields $$W(\alpha) = \psi_2 (\varphi_1 \rho_3 - \varphi_3 \rho_1) (\pm z_S);$$
\item if $q\equiv 0$, where we have $\varphi \in \ker_{L^2_{p,0}}D(\alpha)$, $\psi, \rho \in\ker_{L^2_{0,0}}D(\alpha)$, we obtain 
$$\varphi_2 (\pm z_S) = \varphi_3 (\pm z_S) =0, \ \psi_1(\pm z_S)=0,\ \rho_1(\pm z_S)=0,$$
which yields 
$$W(\alpha) = \varphi_1 (\psi_2 \rho_3 - \psi_3 \rho_2) (\pm z_S).$$
\end{enumerate} 
In Case \ref{case2}, we obtain
\begin{equation*}
    \varphi_2(\pm z_S) = \varphi_3(\pm z_S)=0,\ \ \psi_1(\pm z_S) = \psi_3(\pm z_S)=0, \ \ 
    \rho_1(\pm z_S) = \rho_2(\pm z_S)=0,
\end{equation*}
therefore
$$W(\alpha) = \varphi_1 (\pm z_S) \psi_2 (\pm z_S) \rho_3(\pm z_S).$$
This leads to the following result.
\begin{prop}
\label{prop:wronskian}
    In the notation of equation \eqref{eq:3protect} and \eqref{eq:det}, we have
    \begin{equation*}
        W (\alpha) = 0 \quad\Longrightarrow\quad \alpha\in\mathcal{A}.
    \end{equation*}
\end{prop}
\begin{proof}
For Case \ref{case1}, we only discuss the subcase $p\equiv -q$ as the other subcase is similar. Assume $W(\alpha)=0$. Then if $\psi(\pm z_S)=0$, we can apply the theta function argument  (cf.~equation \eqref{eq:theta_a}) to $\psi$ to conclude that $\alpha\in\mathcal{A}$. Otherwise if $\varphi_1 \rho_3 - \varphi_3 \rho_1$ vanishes at $\pm z_S$, we have that $c_1\varphi + c_2 \rho \in \ker_{L^2_{p,0}}D(\alpha)$ vanishes at $\pm z_S$ for some $c_1,c_2\in\CC$ not both equal to zero. Applying the theta function argument \eqref{eq:theta_a} to $c_1\varphi + c_2 \rho$ we find that $\alpha\in\mathcal{A}$.

For Case \ref{case2}, we have 
$$W(\pm z_S) = \varphi_1(\pm z_S) \psi_2(\pm z_S) \rho_3(\pm z_S).$$
Therefore, $W=0$ implies that one of $\varphi, \psi, \rho$ vanishes at $ z_S$ and  one of $\varphi, \psi, \rho$ vanishes at $-z_S$. In either case, the theta function argument \eqref{eq:theta_a} gives $\alpha\in \mathcal{A}$.
\end{proof}

The kernel of $D( \alpha )$ provides an expression for the inverse of $ D (\alpha) - k $ with $ k \notin \Gamma^* $ and $ \alpha \notin \mathcal A$. In particular, the following result shows that $\alpha \in \mathcal A$ implies that $ W ( \alpha ) = 0$. Combined with Proposition \ref{prop:wronskian}, this proves Theorem \ref{theo:wronskian_intro}.
\begin{prop}
\label{p:inverse}
In view of the notation in \eqref{eq:3protect} and \eqref{eq:det}, we define a three-by-three matrix 
\begin{equation*}
\label{eq:defVa}   \mathbf W ( \alpha ) := [ \varphi , \psi, \rho ], \ \ W ( \alpha ) = \det \mathbf W ( \alpha ) . \end{equation*}
Then $ W(\alpha ) \neq 0 $ and 
$ k \notin \Gamma^* $ imply that{, with the cofactor matrix denoted by $\operatorname{adj}$,}
\begin{equation}
\label{eq:inv}
( D ( \alpha ) - k )^{-1} = \frac{1}{ W ( \alpha ) }  
\mathbf W ( \alpha )  ( 2 D_{\bar z } - k )^{-1} ({\rm{adj}} \mathbf W ( \alpha ) ) . 
\end{equation}
For a fixed $k\notin\Gamma^*$, $\alpha_{12} \mapsto ( D ( \alpha ) - k )^{-1} $ is a meromorphic family of compact operators with poles of finite rank at $ \alpha_{12} \in \mathcal A (\frac{\alpha_{23}}{\alpha_{12}}) $.
\end{prop}
\begin{proof}
If $ W ( \alpha ) \neq 0 $, then $ \mathbf W ( \alpha )^{-1} 
= (W ( \alpha ))^{-1} {\rm{adj}} \mathbf W ( \alpha ) $. 
Write $D(\alpha)= 2D_{\bar z} + \mathbf{U}_{\alpha}(z)$, and note that
$$2D_{\Bar{z}}[\mathbf{W}(\alpha)] = -\mathbf{U}_{\alpha}(z) \mathbf{W}(\alpha)\quad\Longleftrightarrow \quad\mathbf{W}(\alpha) = e^{-\frac{i}{2}\int\mathbf{U}_{\alpha}(z) d\bar{z}}.$$
Hence, $\mathbf{W}(\alpha)^{-1}$ is the integrating factor of the equation
$$2D_{\bar z}u + (\mathbf{U}_{\alpha}(z)- k) u  = v$$
so that
$$(2D_{\bar z}-k)(\mathbf{W}(\alpha)^{-1} u) = \mathbf{W}(\alpha)^{-1} v.$$
This yields equation \eqref{eq:inv}. In view of equation \eqref{eq:bs-id},
\[ ( D ( \alpha ) - k )^{-1} = ( I + \alpha_{12} T_{k})^{-1}
( D( 0 ) - k )^{-1} , \]
where, using analytic Fredholm theory (see for instance 
\cite[Theorem C.8]{res}), we see that $ \alpha_{12} \mapsto ( I + \alpha_{12} T_{k} )^{-1} $ is a meromorphic family of operators with poles of finite rank.
\end{proof}

\section{Trace formul\ae} 

\label{sec:TF} 
In this section, we use the Birman-Schwinger operator $B_k$, defined in equation \eqref{eq:Bk}, to analyse the set of magic parameters $\mathcal A(\alpha_{23}/\alpha_{12})$ for a fixed hopping ratio. More precisely, using Theorem \ref{theo:spectral_char}, we see that the traces of $B_k$ are actually equal to the sums of powers of magic parameters.
Following \cite[Theorem 4]{bhz1}, our goal is to provide a semi-explicit formula for 
$$\tr(B_k^{\ell})=\sum_{\alpha_{12}\in \mathcal A(\frac{\alpha_{23}}{\alpha_{12}})}\frac{1}{\alpha_{12}^{2\ell}},\quad \ell\geq 2. $$
We will prove the following result, which is a slight variation of \cite[Theorem 4]{bhz1}.
\begin{theo}
\label{traceresult}
Let  $B_k:L^2_{0}\to L^2_{0}$ be a meromorphic family of Hilbert-Schmidt operators defined for $k\notin \Gamma^*$. We suppose that $B_k$ satisfies the three properties stated in Lemma \ref{techni}. Then one has, for any $\ell\geq 2$,
\begin{equation}
\label{eq:taul}
\begin{split}
\tr(B_k^{\ell}) & =\frac{2i\pi\omega}{3\sqrt 3}\sum_{n\in \mathbb Z} n\left[\sum_{m\in \mathbb Z}\mathrm{Res}\big(\langle B_k^{\ell}e_{i},e_{i}\rangle_{L^2},\sqrt 3(m\omega^2-n\omega)\big)\right], \end{split} \end{equation} 
where all sums are finite sums.
\end{theo}

\begin{proof}
The only difference compared to the study of twisted bilayer graphene \cite[Theo.4]{bhz1} is that the set of forbidden values for the parameter $k$ is $\Gamma^*$ instead of $(3\Gamma^*+i)\cup (3\Gamma^*-i)$. This causes no harm to the proof as we can, by the residue theorem, translate the contour integral, used in the proof, by a small parameter $\epsilon>0$ without changing the final value.
\end{proof}
\begin{rem}
    Using the conjugation \eqref{eq:UUU}, we see that the trace, for the operator $B_k$ defined in \eqref{eq:Bk}, is the same on any space $L^2_{ir},$ where $r\in \ZZ_3$. The trace on the whole $L^2$ space is then $9$ times the result of \eqref{eq:taul}, see equation \eqref{eq:nine} and the discussion thereafter. Most of the results derived from these trace formulas will only use algebraic properties of the traces and will thus not be affected by the possible rescaling of \eqref{eq:taul}. 
\end{rem}
The proof now boils down to proving the following lemma, which is a slight adaptation of \cite[Lemma 2.2]{bhz1}.
\begin{lemm}
\label{techni}
Let $e_k(z)$ be as in \eqref{eq:ONbasis} for $k\in\Gamma^*$.
Consider a potential $U\in C^{\infty}(\mathbb C/\Gamma; \mathbb C)$ satisfying the first two symmetries of \eqref{eq:UaV} with a finite number of non-zero Fourier mode in its decomposition \eqref{eq:potentialU}. For $\ell\geq 2$, one has:
\begin{itemize}
\item The trace is constant in $k$, 
\begin{equation*}
\tr( B_k^{\ell})=\tau_{\ell} \text{ independent of }k\in \mathbb C \setminus \Gamma^*.
\end{equation*}
\item The function $\mathbb C \setminus \Gamma^* \ni k\mapsto \langle B_k^{\ell}e_{m},e_{m}\rangle_{L^2}$ is a finite sum of rational fractions on the complex plane $\mathbb C$ with degree equal to $-2\ell$ and with $($a finite number of$)$ poles contained in $\Gamma^*$.
\item For any $\gamma\in \Gamma^*$ and for any $k\notin \Gamma^*$, we have
$$ \langle B_k^{\ell}e_{3\gamma+i},e_{3\gamma+i}\rangle_{L^2}=\langle B_{k-3\gamma}^{\ell}e_{i},e_{i}\rangle_{L^2}.$$
\end{itemize}

\end{lemm}
\begin{proof}
The first point is a consequence of the fact that the spectrum of $B_k$ doesn't depend on $k$, which follows from Theorem \ref{theo:spectral_char}.
For the last two points, we prove by induction that $k\mapsto  B_k^{\ell}e_{3\gamma+i}$, where $\gamma\in \Gamma^*$, is of the form 
\begin{equation}
\label{eq:Induc}
B_k^{\ell}e_{3\gamma+i}=\sum_{\nu\in F}R_{\nu+3\gamma}(k)e_{\nu+3\gamma}, 
\end{equation}
where $F\subset \Gamma^*$ is a finite set and $R_{\nu}(k)$ is a sum of rational fractions of degree $-2\ell$ with poles located on $\Gamma^*$. Moreover, we will prove that one has the relation $R_{\nu+3\gamma}(k)=R_{\nu}(k-3\gamma).$

The result is clear for $\ell=0$. Suppose the result is true for $\ell$, and let's prove it holds for $\ell+1$.
The main observation is that multiplication by $U(\pm z)$ acts as a shift on the Fourier basis. The multiplication by $U(-z)$ sends $e_{\nu}$ to a linear combination of $e_{\ell}$ for $\ell \in \Gamma^*$. Then applying $(D(0)-k)^{-1}$ multiplies the coefficient of $e_{\ell}$ by $(\ell-k)^{-1}$. Multiplying by $U(z)$ gives back a linear combination of $e_{\nu}$ with $\nu\in \Gamma^*$. Finally, applying $(D(0)-k)^{-1}$ multiplies the coefficient of $e_{\nu}$ by $(\nu-k)^{-1}$. This means that, using the induction hypothesis \eqref{eq:Induc},
$$ B_k^{\ell+1}e_{3\gamma+i}=\sum_{\nu\in F}R_{\nu+3\gamma}(k)\sum_{\eta\in L}\sum_{\beta\in L}\frac{a_{\eta}}{k-(\nu+\beta-\eta+3\gamma)}\frac{a_{\beta}}{k-(\nu+\beta+3\gamma)}e_{\nu+\beta-\eta+3\gamma},$$
where $L\subset \Gamma^*$ is a finite subset that depends only on $U$ and $a_{\bullet}$ are constants. Thus, it is clear from this formula that the induction carries on to $\ell+1$. This concludes the proof of the lemma.
 \end{proof}

The trace formula is not fully explicit but it is sufficient to prove that the traces, for the standard potential, are of the form $q_{\ell}\pi/\sqrt 3$, for $q_l$ a rational number. This extraordinary rational condition seems to reflect some \emph{hidden integrability} of the set of magic parameters, which we do not fully understand yet. However, just like in \cite{bhz1}, this condition implies that the set of magic parameters is infinite. 
Now that we have the trace formula of Theorem \ref{traceresult}, the proofs of \cite[Theorems 5 and 6]{bhz1} carry over to our setting and we get

\begin{theo}
\label{rat}
Consider a potential $U\in C^{\infty}(\mathbb C/\Gamma; \mathbb C)$ satisfying the first two symmetries of \eqref{eq:UaV} with finitely many non-zero Fourier modes $c_{n} \in \mathbb Q(\omega/\sqrt 3)$ appearing in the decomposition \eqref{eq:potentialU}. Suppose moreover that $\alpha_{23}/\alpha_{12}\in \mathbb Q$. Then for any $\ell \geq 2$, one has $\tau_{\ell}\in \pi\mathbb Q(\omega/\sqrt 3) $. If $U$ also has the third symmetry of \eqref{eq:UaV} then the traces are real and thus $\tau_{\ell}\in  \pi \mathbb Q/\sqrt 3$. In particular, for all potentials satisfying all three symmetries in \eqref{eq:UaV}, one has
$$\forall \ell\geq 2,\quad \tr(B_k^{\ell})=\sum_{\alpha_{12}\in \mathcal A(U,\alpha_{23}/\alpha_{12})}\alpha_{12}^{-2\ell}=\frac{\pi}{\sqrt 3}q_{\ell},\;\;\; q_{\ell}\in \mathbb Q, $$
where $\mathcal A(U,\alpha_{23}/\alpha_{12})$ is the set of magic parameters counting multiplicity for a potential $U$ and fixed hopping ratio $\alpha_{23}/\alpha_{12}$.
\end{theo}
We deduce also the infiniteness of magic parameters, for our canonical choice of potential.
\begin{theo}
\label{theo:races_to_infinity}
Under the assumptions and with the same notation as in Theorem \ref{rat} one has the implication
$$|\mathcal A(U,\alpha_{23}/\alpha_{12})|>0\implies |\mathcal A(U,\alpha_{23}/\alpha_{12})|=+\infty.$$
Let $N\geq 0$, for a tuple $c=(c_{n})_{\{n: \Vert n \Vert_{\infty} \leq N\}}$, define $U$ to be the potential defined by \eqref{eq:potentialU}. Then the above implication holds for a generic $($in the sense of Baire$)$ set of coefficients $c=(c_{n})_{\{n:\Vert n \Vert_{\infty} \le N\}}\in \mathbb C^{(2N+1)^2}$ that contains $\mathbb Q(\omega/\sqrt{3})^{(2N+1)^2}.$
\end{theo}
Finally, we extend the result of \cite{bhz23} on the existence of non-simple magic parameters, see Figure \ref{fig:notall}.
\begin{figure}
       \includegraphics[width=7cm]{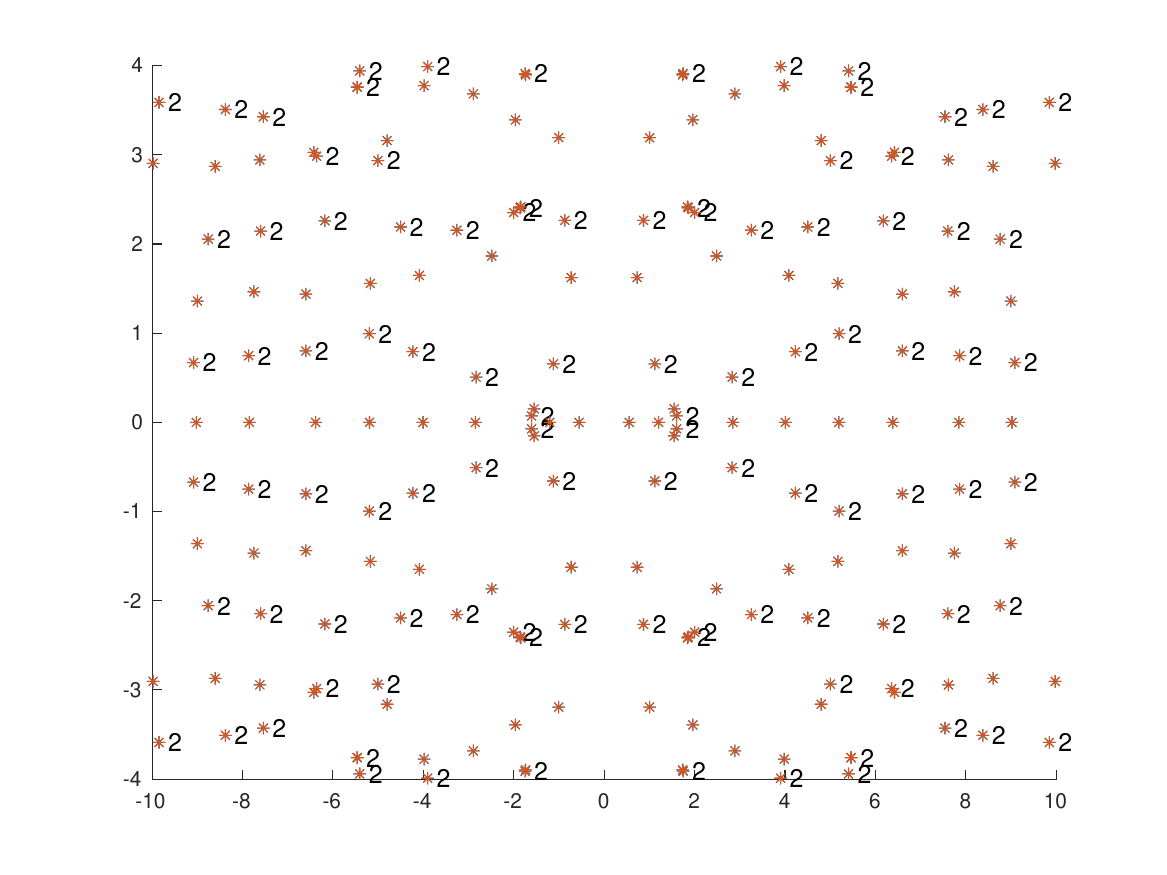}
       \includegraphics[width=7cm]{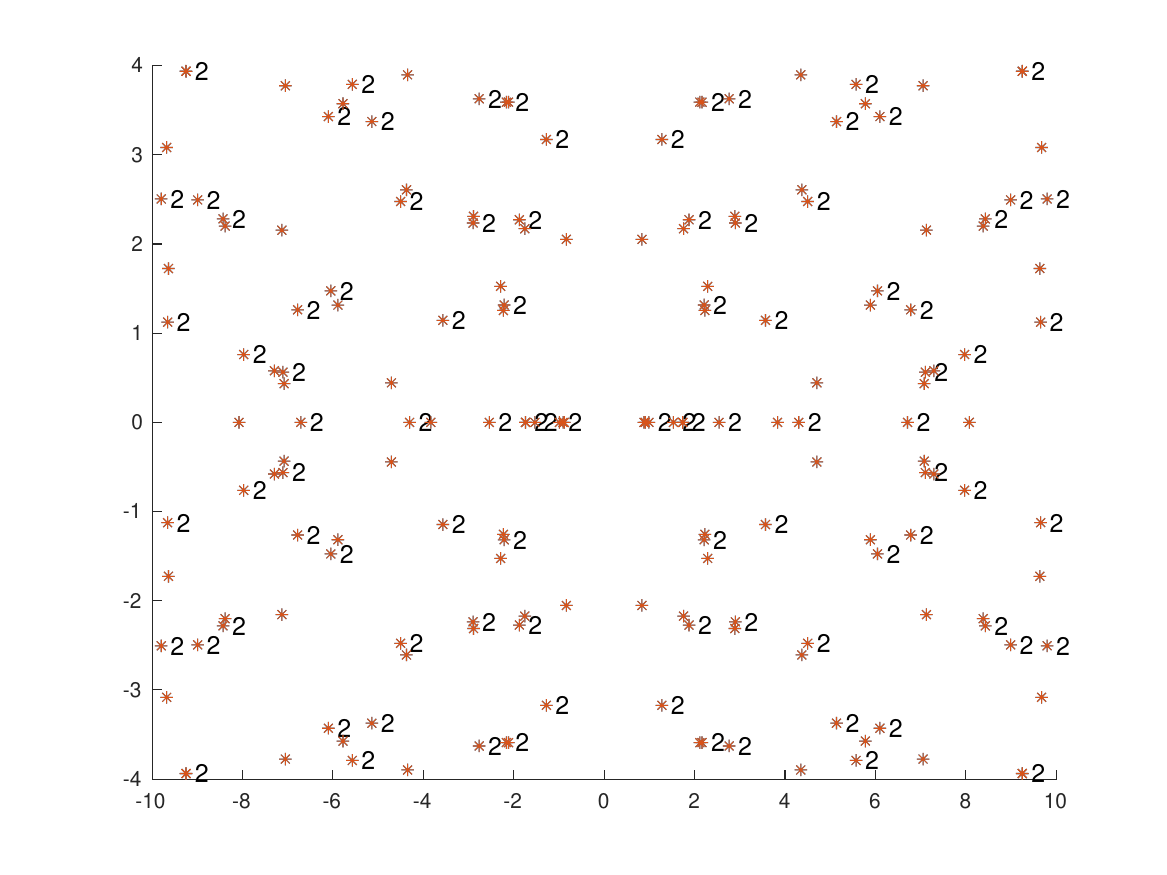}
          \caption{\label{fig:notall}Fixing $\alpha_{32} = \alpha_{12}$ and $\zeta = (1,3)$ (left) and $\zeta=(2,3)$ (right) we plot magic parameters $\alpha_{12}$. Non-simple eigenvalues are indicated by numbers. The multiplicity of the flat bands of the Hamiltonian is twice the multiplicity indicated in the figures. }
\end{figure}
For the rest of this subsection, we will assume we are in Case \ref{case1}, i.e., $p,-q,0$ are not mutually different modulo $3$. The argument in \cite{bhz23} was based on the possibility of taking $k=0$ on a specific translational invariant subspace $L^2_{-ir}$ for some $r\in \ZZ_3$. Then, using a theta function argument, it was shown 
(in \cite[Theorem 1]{bhz23}) that the eigenfunction of a simple magic parameter $\alpha$ on this subspace $L^2_{-ir}$ has to live in the smaller rotational invariant subspace $L^2_{r,2}$ (cf.~\eqref{eq:decomp2}). The existence of non-simple magic parameters was then shown by proving that the traces of the Birman-Schwinger operator $B_0$ at $k=0$ on $L^2_{r,1}$ did not all vanish. In the present case, we recall that the Birman-Schwinger operator, acting on $L^2_0$ of functions invariant by the translations $\mathscr L_a$\footnote{We define the action of $\mathscr L_a$ on scalar valued functions by projecting on the third component, see \eqref{eq:tristan}.}, is given by 
\begin{equation}
\label{eq:Bk_new}
     B_k(\tfrac{\alpha_{23}}{\alpha_{12}}):=R(k)_p^{-1}U(-pz) R(k)_0U(pz)+(\tfrac{\alpha_{23}}{\alpha_{12}})^2 R(k)_pU(p\tfrac{\zeta_2}{\zeta_1}z)R(k)_{p+p\frac{\zeta_2}{\zeta_1}}^{-1}U(-p\tfrac{\zeta_2}{\zeta_1}z). 
    \end{equation}
Here, we have denoted by $R(k)_p$ the restriction of the resolvent $R(k)=(2D_{\bar z}+k)^{-1}$ on the translational invariant subspace  $L^2_p$.
The presence of the factor $R(k)_0$ makes it impossible to put $k=0$. Note that  $p\frac{\zeta_2}{\zeta_1}\equiv q \bmod 3$, which means that $p,0, p+p\frac{\zeta_2}{\zeta_1}$ are not mutually different modulo $3$.  Recall that we defined
$$ L^2_{ir}:=\{u\in L^2(\mathbb C/\Gamma;\mathbb C):  \mathscr L_au=e^{i\langle ir,a\rangle} u=\bar{\omega}^{r}u, \ a\in \Gamma/3\},\quad r\in \mathbb Z_3.$$
The restrictions of $B_k$ on $L^2_{ip}$ are conjugated by the operator $\tau(\cdot)$ defined in  \eqref{eq:tau}. We see that, for $r\in \mathbb Z_3$, 
\begin{equation*}
( B_k)_{L^2_{ir}}:=R(k)_{p-r}^{-1}U(-pz) R(k)_{-r}U(pz)+(\tfrac{\alpha_{23}}{\alpha_{12}})^2 R(k)_{p-r}U(p\tfrac{\zeta_2}{\zeta_1}z)R(k)_{p+p\frac{\zeta_2}{\zeta_1}-r}^{-1}U(-p\tfrac{\zeta_2}{\zeta_1}z).
\end{equation*}
\begin{rem}
\label{rem:T0}
 In particular, there is a space $\mathcal H=L^2_{ir}$ such that one can take $k=0$ for $( B_k)_{\mathcal H}$. Because the spectrum of $T_k$ is the same on each of these spaces, we can suppose, without loss of generality that $p=-q\equiv -1$ modulo $3$. This means we can choose $r=1$ (see Remark \ref{r:correspondence}), i.e $\mathcal H=L^2_i$, and we study $( B_0)_{L^2_i}$,  and we use Proposition \ref{prop:zeros} to get
 $$\alpha \in \mathcal A \text{ is simple} \implies u_{i}\in \ker_{L^2_{i}}D(\alpha) \implies u_{i}\in \ker_{L^2_{-1,-1}}D(\alpha).$$ 
\end{rem} 
Just like in the proof of \cite[Theorem 1]{bhz23}, the existence of a non-simple magic parameter is obtained by showing that $\tr\big((B_0)^{\ell}_{L^2_{-1,1}}\big)$ is non-zero for a well chosen value of $\ell$. The proof is exactly the same as in the case of twisted bilayer graphene, where the argument was based on the fact that one can write
$$\tr\big((B_0)^{\ell}_{L^2_{-1,1}}\big)=\frac 13\tr\big((B_0)^{\ell}_{L^2_i}\big)+\mathcal R_{-1,1}^{\ell}, $$
where $\tr\big((B_0)^{\ell}_{L^2_i}\big)=q_{\ell}\pi/\sqrt 3$ for a non-zero rational number $q_{\ell}$, which means that the first term is transcendental. Now, we were able to show in \cite{bhz23} that the remainder $\mathcal R_{-1,1}^{\ell}$ was algebraic, thus proving that the trace is non-zero. We get the following refinement of Theorem \ref{theo:races_to_infinity}.
\begin{theo}
\label{theo:races_to_infinity2}
Suppose that we are in Case \ref{case1}. Then, under the assumptions and with the same notation as in Theorem \ref{rat}, one has the implication
$$|\mathcal A(U,\alpha_{23}/\alpha_{12})|>0\implies |\mathcal A_m(U,\alpha_{23}/\alpha_{12})|=+\infty.$$
Here, we denoted by $\mathcal A_m(U,\alpha_{23}/\alpha_{12})$ the set of non-simple magic parameters and by $\mathcal A (U,\alpha_{23}/\alpha_{12})$ the set of all magic parameters for the potential $U$.
In particular, the set of magic parameters for our canonical potential $U=U_0$ defined in \eqref{eq:standard_pot} is infinite.

Let $N\geq 0$, and for a tuple $a=(a_{n})_{\{n: \Vert n \Vert_{\infty} \leq N\}}$, define $U_a$ to be the potential defined by \eqref{eq:potentialU} with coefficients $a$. Then the above implication holds for a generic $($in the sense of Baire$)$ set of coefficients $a=(a_{n})_{\{n:\Vert n \Vert_{\infty} \le N\}}\in \mathbb C^{(2N+1)^2}$ that contains $(\mathbb Q(\omega/\sqrt{3}))^{(2N+1)^2}.$
\end{theo}
\begin{rem}
    In Case \ref{case2}, the above argument fails as we cannot consider $k=0$ on any of the rotational invariant subspaces. We note that numerical evidence suggest that there are choices of $\zeta_1,\zeta_2$ for which all magic parameters are simple, see Figure \ref{fig:allsimple}.
\end{rem}
\begin{figure}
\includegraphics[width=7cm]{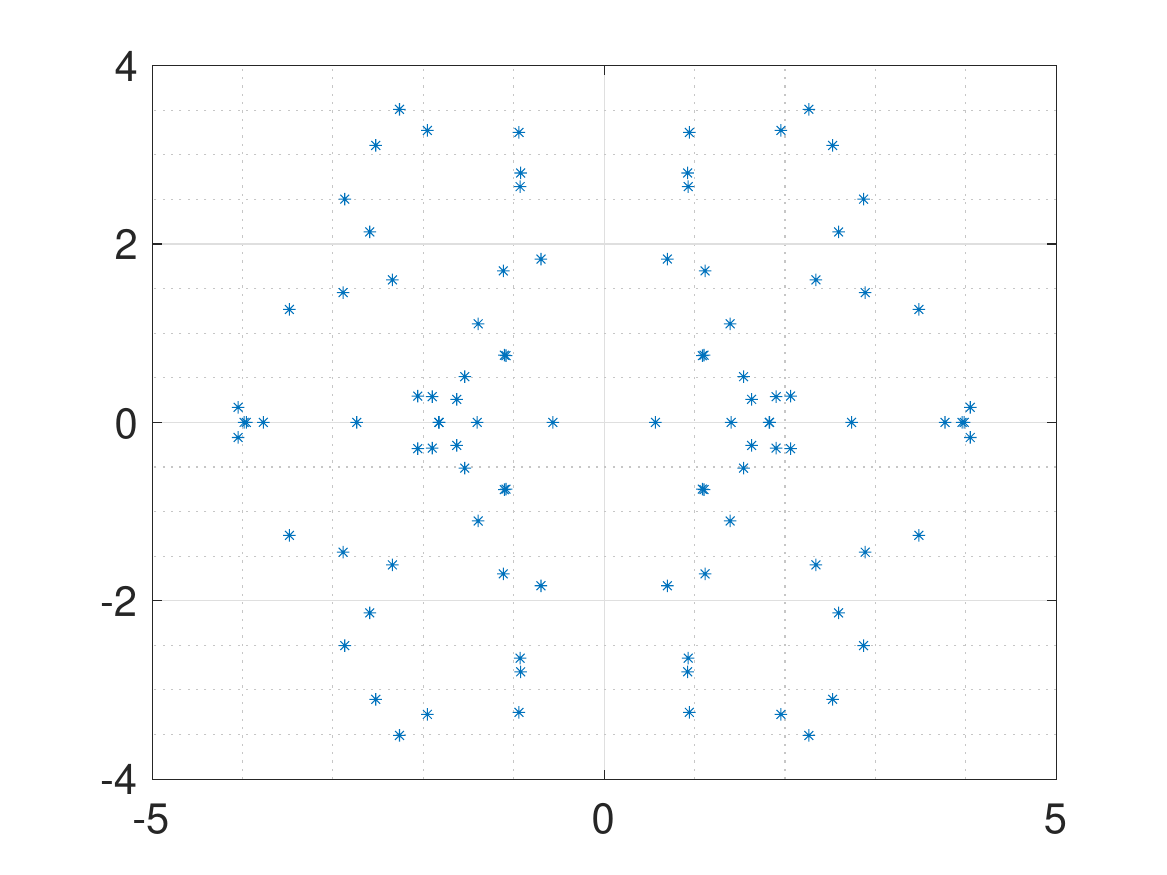}
\includegraphics[width=7cm]{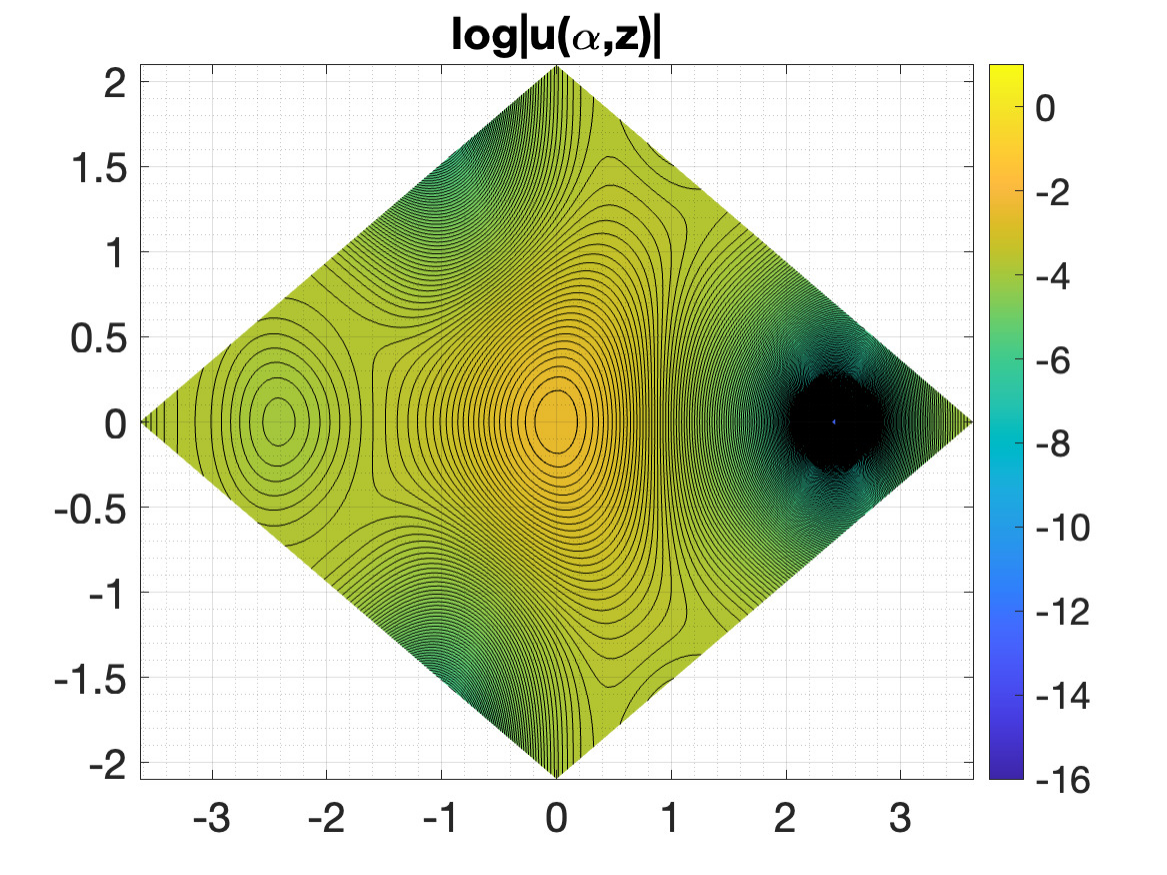}\\
\caption{\label{fig:allsimple} For $\zeta=(1,4)$, $\alpha_{12}=\alpha_{23}$, all magic parameters seem to be simple. For the largest (with positive real part) magic parameter a contour plot of the element of $\ker(D(\alpha))$ is shown.}
\end{figure}

\section{Traces for powers of order 2}
\label{sec:Tracesfor2}
Note that none of the results above proves the existence of a magic parameter. For the bilayer case, the existence, for the Bistritzer--MacDonald potential, was proven by explicitly computing the trace for $\ell=2$ and showing that it is non-zero. In the following section, we provide for $U=U_0$ explicit formulas for $\tr(B_k^2)$ that depend on the choice of parameters $\zeta_1,\zeta_2$, and $\alpha_{23}/\alpha_{12}$.

 For our numerics, it is convenient to use rectangular coordinates $z=2i(\omega y_1+\omega^2 y_2)$, see \cite[\S 3.3]{BEWZ22} for details. In these coordinates, we may introduce 
\begin{equation*}
\begin{gathered}    
{\mathcal D_{k }} :={ \omega^2  
(D_{y_1} + k_1) - \omega (D_{y_2} + k_2 ) }, \\
{\mathcal V ( y )} := \sqrt 3 ( e^{ -i ( y_1 + y_2 ) } 
+ \omega e^{ i (  2 y_1 - y_2 ) } + \omega^2 e^{ i ( - y_1 + 2 y_2 )}  ),
\end{gathered} 
\end{equation*}
with {\em periodic} periodic boundary conditions (for $ y \mapsto
y + 2 \pi  n  $, $  n  \in  \ZZ^2 $). {In the following, we shall write $\widehat{\mathcal V_{\pm}}(y):=\widehat{\mathcal V(\pm y)}.$}
The operator $\ B_k$, defined in \eqref{eq:Bk_new}, reads in the new coordinates
\[ \begin{gathered} 
 \hat B_k(\alpha_{23}/\alpha_{12}):=\widehat{\mathcal D}_k^{-1}\widehat{\mathcal V}_-(p)\widehat{\mathcal D}_k^{-1}\widehat{\mathcal V}_+(p)+(\alpha_{23}/\alpha_{12})^2 \left(\widehat{\mathcal D}_k^{-1}\widehat{\mathcal V}_+(\tilde p)\widehat{\mathcal D}_k^{-1}\widehat{\mathcal V}_-(\tilde p)\right),\end{gathered} \]
where we denoted by $\tilde p =p\zeta_2/\zeta_1$ and 
we introduced 
\begin{equation*} \begin{gathered}   \widehat{ \mathcal D}_{k } :={ \omega^2  
(D + k_1) - \omega (D + k_2 ) }, \text{ with } D= \operatorname{diag}(\ell)_{\ell \in \ZZ} \\
\widehat{\mathcal V}_{\pm} ( y ) :=\sqrt 3\left(  J^{\pm} \otimes J^{\pm} 
+ \omega J^{\mp 2}\otimes J^{\pm} + \omega^2 J^{\pm} \otimes J^{\mp 2}\right),
\end{gathered} 
\end{equation*}
where $J$ is the right-shift $J((a_n)_n) = (a_{n+1})_n $ -- see 
\cite[(3.17)]{BEWZ22}. The space  
$ L^2_{i}  ( \mathbb C/\Gamma ; \mathbb C ) $ corresponds to 
\[ 
\begin{gathered} \ell_{(1 ,1)}^2:=\{f\in \ell^2(\mathbb{Z}^2) : \forall \, n \notin \left(3\mathbb{Z}+1\right)\times\left(3\mathbb{Z}+1\right),\:\: f_{n}=0\}.
 \end{gathered} \]
As in \cite[\S 3.3]{BEWZ22}, we introduce auxiliary operators
$J^{p,q}:=J^p\otimes J^q,\:\: p,q\in \mathbb{Z}$. 
For a diagonal matrix $\Lambda=(\Lambda_{i,j})_{i,j\in \mathbb{Z}}$ acting on $\ell^2(\mathbb{Z}^2)$, we define a new diagonal matrix
$$\Lambda_{p,q}:=(\Lambda_{i+p,j+q})_{i,j\in \mathbb{Z}} .$$
We recall the following properties \cite[(3.24)]{BEWZ22}
\begin{equation*}
J^{p,q}\Lambda J^{p',q'}=\Lambda_{p,q}J^{p+p',q+q'}=J^{p+p',q+q'}\Lambda_{-p',-q'}. 
\end{equation*}
Denoting the inverse of $\widehat{ \mathcal D}_k^{-1}$ by
$$\Lambda=\Lambda_k:=\widehat{ \mathcal D}_k^{-1}, \ \ \  \Lambda_{m,n}=\frac{1}{\omega^2(m+k_1)-\omega(n+k_2)}, \ \ \  (k_1,k_2)\notin \mathbb{Z}^2, $$
we see that 
 \begin{align*}\frac 13 \widehat B_k
&=\Lambda\left(J^{-p,-p}+\omega J^{2 p,-p}+\omega^2J^{-p,2p}\right)\Lambda\left(J^{p,p}+\omega J^{-2p,p}+\omega^2J^{p,-2p}\right)
\\ &\quad+(\alpha_{23}/\alpha_{12})^2\Lambda\left(J^{\tilde p,\tilde p}+\omega J^{-2\tilde p,\tilde p}+\omega^2J^{\tilde p,-2\tilde p}\right)\Lambda\left(J^{-\tilde p,-\tilde p}+\omega J^{2\tilde p,-\tilde p}+\omega^2J^{-\tilde p,2\tilde p}\right).
 \end{align*}
This means that when expanding $\widehat B_k^{\ell},$ we get a diagonal term of the following form:
\begin{equation*}
\begin{gathered} 
((\widehat{B}_k^{\ell})_{ii})_{i \in \ZZ}=\sum_{\pi\in \Theta_{\ell}}a_{\pi}\prod_{i=1}^{\ell}\Lambda_{\tilde{\alpha}_i,\tilde{\beta}_i}\Lambda_{{\tilde{\gamma}_i},{\tilde{\delta}_i}},\\\pi: =\left[\big((\alpha_1,\beta_1),n_1\big),\big((\gamma_1,\delta_1),n_1\big),\big((\alpha_2,\beta_2),n_2\big),...,\big((\gamma_{\ell},\delta_{\ell}),n_l\big)\right], 
\end{gathered}
\end{equation*}
where
\begin{equation*}
\begin{gathered} \tilde{\alpha}_i=\sum_{j=1}^{i-1}n_j(\alpha_j+\gamma_j) \quad \tilde{\beta}_i=\sum_{j=1}^{i-1}n_j(\beta_j+\delta_j), \quad {\tilde{\gamma}_i}=n_i\alpha_i+\sum_{j=1}^{i-1}n_j(\alpha_j+\gamma_j),  \\
\tilde{\delta}_i=n_i\beta_i+\sum_{j=1}^{i-1}n_j(\beta_j+\delta_j). 
\end{gathered} 
\end{equation*}
Here, the tuple  $\pi$ is such that for each $1\leq k \leq \ell,$ we have either
$$n_k=p,\,\, \text{and }\, -(\alpha_k,\beta_k),(\gamma_k,\delta_k)\in \{(1,1),(-2,1),(1,-2)\}, $$
or we have
$$n_k=\tilde p,\,\, \text{and }\, (\alpha_k,\beta_k),-(\gamma_k,\delta_k)\in \{(1,1),(-2,1),(1,-2)\}. $$
We call such a tuple an admissible tuple. The sum is taken over the combinatorial set
\begin{equation}
\label{eq:Theta_l}
 \Theta_{\ell}:=\{\pi :  \pi \text{ is admissible  and }  \widetilde{\alpha_{\ell+1}}=\widetilde{\beta_{\ell+1}}=0\}.
\end{equation}
In other words, we consider all tuples of the form $$\pi: =\left[\big((\alpha_1,\beta_1),n_1\big),\big((\gamma_1,\delta_1),n_1\big),\big((\alpha_2,\beta_2),n_2\big),...,\big((\gamma_{\ell},\delta_{\ell}),n_l\big)\right], $$ admissible in the previous sense and such that the total sum of the (weighted) couples is zero. Finally, the coefficient $a_\pi$ is given by 
$$a_\pi=3^{\ell}(\alpha_{23}/\alpha_{12})^{2s_\pi}\omega^{m_\pi}, \quad m_\pi := \tfrac 23 \sum_{i=1}^{\ell} ( \gamma_i + \beta_i), $$
where $s_\pi$ denotes the number of $k$'s such that $n_k=\tilde p$.
We get the equivalent trace formula given in Theorem \ref{traceresult}, in these Fourier coordinates. The translation between the two coordinate systems is exactly the same as in \cite[Theorem 7]{bhz1}.
To state the result we introduce the notation 
$\gamma_{(a,b)}=\omega^2a-\omega b$, together with the orbits of 
\begin{equation}\label{eq:sigma}
    \sigma:\mathbb{Z}\times \mathbb{Z}\to \mathbb{Z},\:\:\: \sigma(m,n)=(-(n+m),m),
\end{equation}
which satisfies
$$\gamma_{ \sigma(m,n)}=\omega \gamma_{(m,n)}.$$
This implies that $\sigma^3=\operatorname{id}$ and that the orbits under $\sigma$ are of cardinality $3$. 
For a tuple $\pi=\left[((\alpha_1,\beta_1),n_1),((\gamma_1,\delta_1),m_1),((\alpha_2,\beta_2),n_2),((\gamma_{2},\delta_{2}),m_2)\right]\in \Theta_{2}$, we define 
    $$\sigma(\pi)=\left[(\sigma(\alpha_1,\beta_1),n_1),(\sigma(\gamma_1,\delta_1),m_1),(\sigma(\alpha_2,\beta_2),n_2),(\sigma(\gamma_{2},\delta_{2}),m_2)\right]. $$

\begin{theo}
\label{theo:technical_tristan}
Let ${\ell} \ge 2$ and $\Theta_{\ell}$ be as in \eqref{eq:Theta_l} with coefficients $\tilde \alpha_i,..,\tilde \delta_i, m_{\pi}$ as defined above. 
Then the traces are given by
$${\tr}\big(\widehat {B}_k^{\ell}\big)=\frac{-2i\omega \pi}{9}\sum_{\pi\in \Theta_{\ell}}(\alpha_{23}/\alpha_{12})^{2s_\pi}\mathcal C(\pi), $$
where we define the \emph{contribution} of the tuple $\pi$ as
\begin{equation*}
\mathcal C(\pi)=\sum_{(\eta_i,\epsilon_i)\in P(f_{\sigma(\pi)})}{\Res}(f_{\sigma(\pi)},-\gamma_{(\eta_i,\epsilon_i)})\epsilon_i. \end{equation*}
Here, 
the summation is taken on $P(f_{\pi}):=\{ -{({\tilde{\alpha}_i},{\tilde{\beta}_i})},-{({\tilde{\gamma}_i},{\tilde{\delta}_i})},1\leqslant i \leqslant 2\},$ the set of indices of poles of~$f_{\pi}$, with
$$f_{\pi}(k):=3^{\ell}\omega^{m_{\pi}}\prod_{i=1}^{\ell}\frac{1}{(k+\gamma_{({\tilde{\alpha}_i},{\tilde{\beta}_i})}+\mu)(k+\gamma_{({\tilde{\gamma}_i},{\tilde{\delta}_i})}+\mu)}, \quad \mu:=\omega^2-\omega.$$
\end{theo}
Thus, there are two components in the trace formula. The first one is the set $\Theta_{\ell}$. Note that this set only depends on $\zeta_2/\zeta_1$ but \emph{not} on $\alpha_{23}/\alpha_{12}$. The second component is given by the coefficients $a_{\pi}$ which depend on both parameters, with a polynomial dependence in $\alpha_{23}/\alpha_{12}$. The full computation of $\Theta_{\ell}$ is the main difficulty to make the trace formula from Theorem \ref{traceresult} fully explicit. This combinatorial set describes how the Jordan blocks in $\hat B_k$ interact with each other when computing the powers of $\hat B_k$ and thus reflects the \emph{non-normal} nature of the operator $B_k$. However, for $\ell=2$ and for $U=U_0$, the set $\Theta_{\ell}$ is sufficiently small to be computable by hand explicitly. We get different formulas for different values of $\zeta_2/\zeta_1$. We sum up the computations in the following theorem. When studying the continuity of the set of magic parameters with respect to parameters $\zeta_1,\zeta_2$ (see Section \ref{sec:continuity}), we will study the set of \emph{unrescaled magic parameters}. Using the rescaling of equation \eqref{eq:original}, we see that this set is defined to be
\begin{equation*}
\mathcal B=\frac{\zeta_1}p\mathcal A,
\end{equation*}
see Remark \ref{rem:rescaling}.
Using the same notation $\mathcal B(r)$, we deduce the following explicit formulas for
$$\mathcal S_{2\ell}(\beta_{23}/\beta_{12})=\sum_{\beta_{12}\in \mathcal B(\beta_{23}/\beta_{12})}\frac{1}{\beta_{12}^{2\ell}},\quad \ell\geq 2, $$
for $\ell=2$.
\begin{theo}
\label{explicit}
Consider $U=U_0$ defined in \eqref{eq:standard_pot}. Denote by $\mathcal B(\beta_{23}/\beta_{12}, \zeta_1,\zeta_2)$ the set of unrescaled magic parameters for the choice of parameters $\beta_{23}/\beta_{12}$, $\zeta_1$, $\zeta_2$, and let $\mathcal S_{2\ell}(\beta_{23}/\beta_{12},\zeta_1,\zeta_2)$ denote the corresponding trace. 
\begin{itemize}
\item If $\zeta_2/\zeta_1=-1,$ then we have
$$\mathcal S_4(\beta_{23}/\beta_{12},\zeta_1,-\zeta_1)=\left(1+\left(\frac{\beta_{23}}{\beta_{12}}\right)^2\right)^2\frac{4\pi p^2}{9\sqrt 3 \zeta_1^4 }. $$
\item If $\zeta_2/\zeta_1\neq\pm 1,$ then we have
$$\mathcal S_4(\beta_{23}/\beta_{12},\zeta_1,\zeta_2)=\frac{4\pi p^2}{9\sqrt 3 \zeta_1^4 }\left(\frac{1}{(\zeta_2/\zeta_1)^2}\left(\frac{\beta_{23}}{\beta_{12}}\right)^4+\frac{3}{1-\zeta_2/\zeta_1+(\zeta_2/\zeta_1)^2 }\left(\frac{\beta_{23}}{\beta_{12}}\right)^2+1\right). $$
\item If $\zeta_2/\zeta_1=1,$ then we have
$$\mathcal S_4(\beta_{23}/\beta_{12},\zeta_1,\zeta_1)=\left(1-\left(\frac{\beta_{23}}{\beta_{12}}\right)^2+\left(\frac{\beta_{23}}{\beta_{12}}\right)^4\right)\frac{4\pi p^2}{9\sqrt 3 \zeta_1^4}. $$
\end{itemize}
In particular, if the hopping ratio $\beta_{23}/\beta_{12}\in \mathbb R$, then $\mathcal S_4>0$ and there exist magic parameters.
\end{theo}
\begin{proof}
We will prove the equivalent formulas for $\mathcal A$ to keep the notation of the article consistent. First note that the case $$\zeta_2/\zeta_1=-1,$$ reduces to the bilayer case and the formula thus follows from \cite{BEWZ22}. We list a few tricks that make the calculation easier. 
\begin{enumerate}
\item 
    First, it is easy to see that $\Theta_2$ is stable under the map $\sigma$ in \eqref{eq:sigma}. Moreover, we have $m_{\sigma(\pi)}\equiv m_{\pi}-2\;\mathrm{mod}\; 3.$ Because we have the relation $\gamma_{ \sigma(m,n)}=\omega \gamma_{(m,n)}$, we see that
    $$ {\Res}(f_{\sigma(\pi)},\gamma_{\sigma(m,n)})=\omega^{-2}\omega^{2+1}{\Res}(f_{\pi},\gamma_{m,n})=\omega{\Res}(f_{\pi},\gamma_{m,n}).$$
    In other words, to compute the contributions on all the orbit, it is sufficient to compute only one residue out of the three, namely:
    $$\mathcal{C}(\pi)+\mathcal{C}(\sigma(\pi))+\mathcal{C}(\sigma^2(\pi))=-2i\pi\omega \!\!\sum_{(\eta_i,\epsilon_i)\in P(f_{\pi})}{\Res}(f_{\pi},-\gamma_{(\eta_i,\epsilon_i)})(\epsilon_i+\eta_i\omega-(\eta_i+\epsilon_i)\omega^2). $$
    We will use the notation $\mathcal{C}([\pi])$ for the contribution of the whole orbit.
    
\item The invariance by cyclic permutation of the trace implies that
    $$\mathcal{C}((\alpha_1,\beta_1),(\gamma_1,\delta_1),(\alpha_2,\beta_2),(\gamma_{2},\delta_{2}))=\mathcal{C}((\alpha_2,\beta_2),(\gamma_2,\delta_2),(\alpha_1,\beta_1),(\gamma_{1},\delta_{1})). $$
    \item Moreover, we also get the following relation 
     $$\mathcal{C}((\alpha_1,\beta_1),(\gamma_1,\delta_1),(\alpha_2,\beta_2),(\gamma_{2},\delta_{2}))=\mathcal{C}((-\gamma_1,-\delta_1),(-\alpha_1,-\beta_1),(-\gamma_{1},-\delta_{1})(-\alpha_2,-\beta_2)). $$
    Indeed, using $\gamma_{(-m,-n)}=-\gamma_{(m,n)}$, we see that we get a minus sign in the residue part of the contribution, but we also get a minus sign in the ``epsilon'' part so the resulting contribution is unchanged.
\end{enumerate}
All of this will reduce the number of residues one will have to compute. What is left is to compute the set $ \Theta_2$. Consider
$$\pi: =\left[\big((\alpha_1,\beta_1),n_1\big),\big((\gamma_1,\delta_1),n_1\big),\big((\alpha_2,\beta_2),n_2\big),\big((\gamma_2,\delta_2),n_2\big)\right]\in \Theta_2.$$   
Let's start by introducing some terminology. We'll call $A$, $B$ and $C$ the lines through the origin and with direction $(1,1)$, $(-2,1)$ and $(1,-2)$ respectively. We'll say that a couple $(\alpha,\beta)\in \mathbb Z^2$ is of type $A$, $B$ or $C$ if the corresponding point in the $\mathbb R^2$ plane is on the corresponding line. 
The sum of the four (weighted) couples is zero by definition. By the pigeonhole principle, at least two of the couples have the same type. Without loss of generality, let's suppose it is $A$. There are a few cases to cover.

\textbf{The case $p\neq \pm \tilde p$:}
\begin{itemize}
\item If $(\alpha_1,\beta_1)=-(\gamma_1,\delta_1)$ (i.e., $(\alpha_1,\beta_1)$ and  $(\gamma_1,\delta_1)$ have same type) then we must also have $(\alpha_2,\beta_2)=-(\gamma_2,\delta_2)$ and the converse also holds. This gives our first family of tuples:
\begin{align*}\pi=\left[\big((\alpha_1,\beta_1),p\big),\big((\gamma_1,\delta_1),p\big),\big((\alpha_2,\beta_2),\tilde p\big),\big((\gamma_2,\delta_2),\tilde p\big)\right],\,\,\,
\\ (\alpha_1,\beta_1)=-(\gamma_1,\delta_1),\,\,\,\,\,\ (\alpha_2,\beta_2)=-(\gamma_2,\delta_2).\quad\quad \quad
\end{align*}
\item If $(\alpha_1,\beta_1)=-(\alpha_2,\beta_2)=(1,1)$, then $p(\alpha_1,\beta_1)+\tilde p(\alpha_2,\beta_2)=(p-\tilde p)(1,1).$
This is a non-zero vector on the line $A$. If one of the last two couples is of type $A$, the so is the last one, so this case is included in the previous one. We can thus suppose that $(\gamma_1,\delta_1)$ is of another type, say $B$. Because $p\neq \tilde p$ and $(1,1)+(-2,1)+(1,-2)=0$, this implies that $p-\tilde p=-p=-\tilde p$ which implies $p=0$ and gives a contradiction. Thus, there are no additional tuple if $p\neq \tilde p$. 
\item The last case also does not produce any other tuples in the case $p\neq \tilde p$.
\end{itemize}
We know from the previous three points that we only need to consider the residue for 
\begin{align*}
\pi_1&=[((-1,-1),p), ((1,1),p),((1,1),\tilde p), ((-1,-1),\tilde p)],
\\
\pi_2&=[((-1,-1),p), ((1,1),p),((-2,1),\tilde p), ((2,-1),\tilde p)],  
\\ \pi_3&=[((-1,-1),p), ((1,1),p),((1,-2),\tilde p), ((-1,2),\tilde p)]. \end{align*}
A straightforward computation of the residues then gives
$$\tr(B_k^2)(p,\alpha_{23}/\alpha_{12},\zeta_2/\zeta_1)=\frac{4\pi}{9\sqrt 3 p^2}\left(\frac{1}{(\zeta_2/\zeta_1)^2}\left(\frac{\alpha_{23}}{\alpha_{12}}\right)^4+\frac{3}{1-\zeta_2/\zeta_1+(\zeta_2/\zeta_1)^2 }\left(\frac{\alpha_{23}}{\alpha_{12}}\right)^2+1\right). $$
\textbf{The case $p=\tilde p$:}
The condition becomes 
$$ (\alpha_1,\beta_1)+(\gamma_1,\delta_1)+(\alpha_2,\beta_2)+(\gamma_2,\delta_2)=0,$$
where we have
$$-(\alpha_1,\beta_1),(\gamma_1,\delta_1),(\alpha_2,\beta_2),-(\gamma_2,\delta_2)\in \{(1,1),(-2,1),(1,-2)\}.$$
The only solutions are given by
$$(\alpha_1,\beta_1)=-(\alpha_2,\beta_2),\, \, (\gamma_1,\delta_1)=-(\gamma_2,\delta_2), $$
or by 
$$(\alpha_1,\beta_1)=-(\gamma_1,\delta_1),\,\, (\gamma_2,\delta_2)=-(\alpha_2,\beta_2). $$
We can compute the residues, using again the three remarks, which gives the formula
$$\tr(B_k^2)(p,\alpha_{23}/\alpha_{12},1)=\left(1-\left(\frac{\alpha_{23}}{\alpha_{12}}\right)^2+\left(\frac{\alpha_{23}}{\alpha_{12}}\right)^4\right)\frac{4\pi}{9\sqrt 3 p^2}. $$
This completes the proof.
\end{proof}
By combining Theorems \ref{theo:races_to_infinity} and \ref{explicit} we obtain Corollary \ref{corr:degenerate} from the introduction.

\section{Continuity of bands \& discontinuity of magic parameters}
\label{sec:continuity}
In Proposition \ref{prop:Hoelder_cont} we show the H\"older 1/2-continuity of the spectrum of the Hamiltonian in Hausdorff-distance. This implies that the bands do not change dramatically when changing the twisting angles. This is in contrast to Theorem \ref{thm:discontinuity}, which says that the coupling coefficients giving rise to flat bands are rather ill-behaved when changing the twisting angles.
\begin{prop}
\label{prop:Hoelder_cont}
The map $(\beta,\tilde \beta,\zeta) \in \CC^{2}\times \CC^2 \times \mathbb R^2 \mapsto \Spec_{L^2(\CC;\CC^4)}(\mathcal H(\beta,\zeta))$ is locally H\"older 1/2 continuous in Hausdorff distance.
\end{prop}
\begin{proof}
Let $L>0$ to be fixed later on. If $d(\lambda,\Spec_{L^2(\CC;\CC^4)}(\mathcal H(\beta,\zeta))) <\varepsilon$ then there is $\varphi_{\varepsilon}$ normalized such that $\Vert (\mathcal H(\beta,\zeta)-\lambda)\varphi_{\varepsilon}\Vert \le \varepsilon \Vert \varphi_{\varepsilon}\Vert$. We then define $\eta_L(z):=\chi(z/L)$ with cut-off function $\chi \in C_c^{\infty}(\CC;[0,1])$ such that $\chi(z)=1$ for $\vert z \vert_{\infty} \le 1/2.$ 
We then set $\eta_{j,L}(z):=\eta_{L}(z-j)$ for $j \in \ZZ+i\ZZ$ and define $\alpha^{\uparrow}_L:=\sup_{z} \sum_{j \in \ZZ+i\ZZ} \vert\eta_{j,L}(z)\vert^2  = \mathcal O(L^2)$ and $\alpha_L^{\downarrow}:= \inf_{z} \sum_{j \in \ZZ+i\ZZ} \vert\eta_{j,L}(z)\vert^2  = \mathcal O(L^2)$ with $\alpha^{\uparrow}_{L}/\alpha_L^{\downarrow} = \mathcal O(1).$
Thus, we have that 
\[ \begin{split}
\sum_{j \in \ZZ+i\ZZ} \Vert \eta_{j,L}(\mathcal H(\beta,\zeta)-\lambda) \varphi_{\varepsilon} \Vert^2 
&\le \alpha^{\uparrow}_L \Vert (\mathcal H(\beta,\zeta)-\lambda) \varphi_{\varepsilon} \Vert^2 \lesssim \alpha_L^{\uparrow} \varepsilon^2 \Vert \varphi_{\varepsilon}\Vert^2 \lesssim \varepsilon^2 \sum_{j \in \ZZ+i\ZZ} \Vert \eta_{j,L}\varphi_{\varepsilon}\Vert^2. 
\end{split} \]

Let $\delta>0$, we then have that using $\Vert u+v\Vert^2 \le (1+\delta) \Vert v \Vert^2 + (1+\delta^{-1})\Vert w\Vert^2$
\[ \sum_{ j \in \ZZ+i\ZZ} \Vert (\mathcal H(\beta,\zeta)-\lambda)\eta_{j,L} \varphi_{\varepsilon} \Vert^2 \le (1+\delta^{-1}) \varepsilon^2 \sum_{j \in \ZZ+i\ZZ} \Vert \eta_{j,L}\varphi_{\varepsilon}\Vert^2 + (1+\delta) \sum_{j \in \ZZ+i\ZZ} \Vert [\eta_{j,L}, \mathcal H(\beta,\zeta)]\varphi_{\varepsilon}\Vert^2 .\]
Now, $\sum_{j \in \ZZ+i\ZZ}\Vert [\eta_{j,L}, \mathcal H(\beta,\zeta)] \varphi_{\varepsilon}\Vert^2 \lesssim \Vert \varphi_{\varepsilon}\Vert^2 \lesssim \frac{1}{L^2}\sum_{j \in \ZZ+i\ZZ} \Vert \eta_{j,L} \varphi_{\varepsilon} \Vert^2 .$
This implies that
\[ \sum_{ j \in \ZZ+i\ZZ} \Vert (\mathcal H(\beta,\zeta)-\lambda)\eta_{j,L} \varphi_{\varepsilon} \Vert^2 \lesssim \Big((1+\delta^{-1}) \varepsilon^2 + \tfrac{1+\delta}{L^2} \Big) \sum_{j \in \ZZ+i\ZZ} \Vert \eta_{j,L} \varphi_{\varepsilon} \Vert^2 .\]
Since this summed inequality holds, there is in particular some $j$ such that 
\[ \Vert (\mathcal H(\beta,\zeta)-\lambda)\eta_{j,L} \varphi_{\varepsilon} \Vert \lesssim  \sqrt{(1+\delta^{-1}) \varepsilon^2 + \tfrac{(1+\delta)}{L^2} } \Vert \eta_{j,L} \varphi_{\varepsilon} \Vert .\]
This implies that for new parameters $\beta',\zeta'$ 
\[ \Vert (\mathcal H(\beta',\zeta')-\lambda)\eta_{j,L} \varphi_{\varepsilon} \Vert \lesssim  \Big(\sqrt{(1+\delta^{-1}) \varepsilon^2 + \tfrac{1+\delta}{L^2} }  + (\vert \beta - \beta'\vert + \Vert U\Vert_{C^1} \vert \zeta-\zeta'\vert L) \Big)\Vert \eta_{j,L} \varphi_{\varepsilon} \Vert.\]
Choosing $L =(\vert \beta - \beta'\vert +  \Vert U \Vert_{C^1} \vert \alpha-\alpha'\vert)^{-1/2},$ the result follows. 
\end{proof}

While the previous implies that the spectrum is well-behaved (in particular also as we approach incommensurable angles), we now prove Theorem \ref{thm:discontinuity}, stated in the introduction saying that the set of magic parameters $\mathcal B$ is \emph{discontinuous }for varying parameters $\zeta_1$ and $\zeta_2$.

\begin{proof}[Proof of Theorem \ref{thm:discontinuity}]
We start by introducing
$ h:=\frac{\zeta_2}{\zeta_1}=3^k\frac{r_1}{r_2},\text{ with }r_1,r_2\not\equiv 0 \operatorname{mod} 3.$
We shall then choose a sequence $(\zeta_1^{(n)},\zeta_2^{(n)})$ such that $\frac{\zeta_1^{(n)}}{\zeta_2^{(n)}}$ yields integers $p_n$ in the trace formula diverging to infinity. More precisely, we choose
$$\zeta_1^{(n)}=\zeta_1,\quad \zeta_2^{(n)}=3^k\frac{r_1}{r_2}\times \frac{3^n}{3^n-1}\zeta_1^{(n)}=:h_n \zeta_1^{(n)} \ \Longrightarrow \ (\zeta_1^{(n)},\zeta_2^{(n)})\to (\zeta_1,\zeta_2),\quad p_n=r_2(3^n-1).  $$
Using Theorem \ref{explicit}, then for $n$ large enough, we have $\zeta_1^{(n)}\neq \zeta_2^{(n)}$ and thus
$$\mathcal S_4(r,\zeta_1^{(n)},\zeta_2^{(n)})=\frac{4\pi p_n^2}{9\sqrt 3 \zeta_1^4 }\left(\frac{1}{h_n^2}r^4+\frac{3}{1-h_n+h_n^2 }r^2+1\right). $$
This means that we have
\begin{equation*}
\mathcal S_4(r,\zeta_1^{(n)},\zeta_2^{(n)})\sim \underbrace{\frac{4\pi}{9\sqrt 3 \zeta_1^4 }\left(\frac{1}{h^2}r^4+\frac{3}{1-h+h^2 }r^2+1\right)}_{>0}p_n^2\to +\infty.\qedhere \end{equation*}
\end{proof}
Since the operator $B_k$ depends analytically on $\frac{\alpha_{23}}{\alpha_{12}}$, we immediately conclude:
\begin{prop}
    Let $\zeta_1/\zeta_2 \in \mathbb Q \setminus\{0\}$ fixed. The magic parameters depend continuously on the ratio $\frac{\alpha_{23}}{\alpha_{12}}$. In Case \ref{case1}, the same is true for magic parameters corresponding to eigenvalues of $B_k$ on subspaces $L^2_{-p,\ell}$ for $\ell \in \ZZ_3.$
\end{prop}   
The latter part implies that flat bands of multiplicity $2$ are not destroyed by varying the ratio $\zeta_1/\zeta_2.$ This condition is actually sharp, as numerical experiments suggest that magic parameters of multiplicity 2 in the case of equal twisting angles $\zeta_1=\zeta_2$ with $\alpha_{23}=\alpha_{12}$ are immediately destroyed when varying this ratio. 

\section{Generic simplicity}
\label{sec:gen_simpl}
In this section, we argue that for $\zeta_1,\zeta_2$ as in Case \ref{case1}, the magic parameters are either simple or two-fold degenerate for a generic set of tunnelling potentials satisfying the translational and rotational symmetry of the tunnelling potentials in the chiral limit. 

Since we assume Case \ref{case1}, Remark \ref{rem:T0} gives that $T_0$ is well-defined on $L^2_{ip}$.

We start by noticing that the existence of an eigenfunction of $T_0$ in a representation $L^2_{j,\ell}$ implies the existence of another eigenfunction in $L^2_{j,-\ell+1}$ for $\ell \in \{0,1\}.$ In particular, spectrum of $T_0$ in representations $L^2_{j,0}$ or $L^2_{j,1}$ always gives rise to degenerate flat bands. 
\begin{prop}
\label{prop:zeros}
Assume Case \ref{case1}, then $T_0: L^2_{-p,\ell} \to L^2_{-p,\ell}$ is well-defined for $\ell \in \ZZ_3$, and in addition, 
\[ \Spec_{L^2_{-p,\ell}}(T_0) =  \Spec_{L^2_{-p,-\ell+1}}(T_0)\] 
with equality of geometric multiplicity. In particular, if $1/\alpha \in \Spec_{L^2_{-p,\ell}}(T_0)$ then the multiplicity of the magic angle is at least $2$. Let $u_i \in \ker_{L^2_{-p,i}}(D(\alpha)) \setminus \{0\}$, then 
\[ u_0(z) = (z \pm z_S) w_0(z), \quad u_1(z) = z^2 w_1(z), \quad u_2(z) = z w_2(z) \quad\text{with }w_i \in C^{\omega}(\CC;\CC).\] 
If $u \in \ker_{L^2_{-p,1}}(D(\alpha)) \setminus \{0\}$ then $\wp(z, \frac{4}{3}\pi i w, \frac{4}{3}\pi i w^2)u \in \ker_{L^2_{-p,0}}(D(\alpha))$, where $\wp(z)$ is the $\Gamma_3$-periodic Weierstrass $p$-function with a double pole at zero. 
\end{prop}
\begin{proof}
We take $v \in L^2_{-p,\ell}$ with $T_0v = -\lambda v$ and $\lambda \neq 0$. Multiplying by $2D_{\bar z}$ we find that $D(1/\lambda)v = 0$. Using the anti-linear symmetry $Q$, defined in \eqref{eq:mapQ}, we have that $Qv \in L^2_{-p,-\ell}$ and $D(1/\lambda)^*Qv=0.$ This implies that $(1-1/\lambda T_0^*) (2D_{\bar z}Qv)=0,$ where $(2D_{\bar z}Qv) \in L^2_{-p,-\ell+1}.$

The expansion of the $u_i$ at the zeros follows as in Section $6$ of \cite{bhz23}, cf.~in particular \cite[Theorem $5$]{bhz23}. By the usual theta function argument, the number of flat bands is equal to two (cf.~Lemma \ref{lem:theta-indep}).
\end{proof}

We now provide an argument showing that the spectrum of $T_0$ in each representation $L^2_{j,\ell} $ is simple. 
\label{s:gensimp}

\subsection{Generalized potentials}
We consider the general class of potentials of the form
\begin{equation}
\label{eq:newU}
\begin{split}
 V &= \begin{pmatrix} 0 & U_+ & 0 \\ U_- & 0 & Y_+ \\ 0 & Y_- & 0 \end{pmatrix} \text{ where }
 U_{\pm} (\omega z) = \omega U_{\pm}(z), Y_{\pm}(\omega z)=\omega Y_{\pm}(z) \\
  U_{\pm}&(z + na) = \bar \omega^{ \pm p n(a_1+a_2)}U_{\pm }(z), \quad  Y_{\pm}(z + na) = \bar \omega^{ \pm \frac{p \zeta_2}{\zeta_1} n(a_1+a_2)}Y_{\pm }(z),
\end{split}
\end{equation}
with $a = \frac{4\pi i}{3}(a_1\omega+ a_2 \omega^2)\in\Gamma_3$.
 We do
{\em not} however assume $\overline{U_{\pm}(z)} = U_{\pm}(\bar z)$ and $\overline{Y_{\pm}(z)} = Y_{\pm}(\bar z)$.
For the topology defining the notion of \emph{genericity},
it is convenient to use the following Hilbert space of {\em real analytic} potentials defined
using the following norm: for fixed $ \delta > 0 $, 
\begin{equation*}
\| V \|_\delta^2 := 
\sum_{\pm} \sum_{  k \in \Gamma^* } (|a^\pm_{ k } |^2+|b^\pm_{ k } |^2) e^{ 2 |  k| \delta}, \   \ \ \ 
U_\pm ( z ) = \sum_{   k \in \Gamma^* } a^\pm_{ k } e^{ i \langle z ,  k \rangle } , \ \ \ 
Y_\pm ( z ) = \sum_{   k \in \Gamma^* } b^\pm_{ k } e^{ i \langle z ,  k \rangle } . 
\end{equation*}
Then we define $ \mathscr V = \mathscr V_\delta $ by 
\begin{equation}
\label{eq:Vscr}
V \in \mathscr V \ \Longleftrightarrow \ 
\text{ $ V$ satisfies \eqref{eq:newU}, } \ 
\| V \|_{\delta } < \infty . 
\end{equation}

We also define the more general potential
\begin{equation}
\label{eq:newU2}
 V =\begin{pmatrix} W_1 & U_+ & X_+ \\ U_- & W_2 & Y_+ \\ X_- & Y_- & W_3 \end{pmatrix}
 \end{equation}
with $U_{\pm},Y_{\pm}$ as before and 
\[ W_i (\omega z) = \omega W_i(z) \text{ and }X_{\pm}(\omega z) = \omega X_{\pm}(z)\]
as well as 
\[   W_{i}(z + na) =W_{i }(z) \text{ and }X_{\pm}(z + na) = \bar \omega^{ \pm p(1+\zeta_2/\zeta_1) n(a_1+a_2)}X_{\pm }(z).\]
The set of such $V$ with entries satisfying \eqref{eq:Vscr} are denoted by $\mathscr V_{\text{full}}.$
Such $V$ still satisfy the relevant symmetries $[\mathscr L_{a},V]=0 \text{ and }[\mathscr C,V]=0.$
The potentials $W_i$ describe magnetic potentials or strain fields and $X_{\pm}$ describes tunnelling between top and bottom layers. It is worth noticing that our argument to show generic simplicity relies on the matrix set $\mathscr V_{\text{full}}$ which requires more potentials than there are tunnelling potentials in chiral limit TTG. 

Before proceeding with the results, we shall explain the simplifications made along the way. We start with the simplifications needed already in the case of TBG \cite{bhz23}:
\begin{itemize}
    \item We give up on the symmetry $\overline{U_{\pm}(z)} = U_{\pm}(\bar z),$ since a potential satisfying this symmetry does not describe a dense set of functions in $L^2_{j,\ell}$ spaces.
    \item We do not assume that $U_{+}(z) = U_-(-z)$ as is assumed in the idealized chiral BM model. This choice would introduce a coupling between the components of the potential that is not general enough to split eigenvalues. 
\end{itemize}
Sticking to these simplifications in the case of TTG allows us to show that for a generic potential all eigenvalues in each representation are simple. Giving up on the symmetry $U_+(z) = U_-(-z)$ breaks the $C_2\mathcal T$ symmetry of the Hamiltonian which is given by the operator $Q$ in \eqref{eq:mapQ}. This is not an issue in the case of TBG, since one still has the symmetry
$(Av)(z):=\begin{pmatrix} 0 & 1 \\ -1 & 0 \end{pmatrix} \overline{v(z)}$ such that $AD_{\text{TBG}}(\alpha)A=-D_{\text{TBG}}(\alpha)^{*}$ relating $D_{\text{TBG}}(\alpha)$ to its adjoint. This symmetry is no longer available in the case of TTG. 

One can still show, along the lines of what has been done in TBG, that for a generic perturbation the eigenvalues of $T_k$ are simple in each subspace $L^2_{-p,\ell},$ but a splitting of representations is not possible. Indeed, one can only infer that if $T_0$ has a simple eigenvalue in each representation that under a generic perturbation they either split (leading to flat bands of Chern number -1) right away or move together but are related by $\theta$ functions. This is then enough to conclude that the Chern number has to be $-1$ under generic perturbations, see Corollary \ref{corr:Chernperturbations}.

Our main result of this section is the following theorem: 

\begin{theo}
\label{t:sim}
Assume Case \ref{case1}. There exists a generic subset (an intersection of open dense sets),
$ \mathscr V_0 \subset \mathscr V $, with $ \mathscr V $ defined in \eqref{eq:Vscr},
such that if $V \in \mathscr V_0$ then at any magic parameter $\alpha \in \CC^2$, the operator $T_0$ has at most a simple eigenvalue in each subspace $L^2_{-p,\ell}.$ For a generic set $ \mathscr V_0 \subset \mathscr V_{\text{full}} $, the flat bands are either simple bands in $L^2_{-p}$ such that \[\dim\ker_{L^2_{-p,2}}(D(\alpha)) =1 \text{ and }\dim\ker_{L^2_{-p,0}}(D(\alpha))=\dim\ker_{L^2_{-p,1}}(D(\alpha))=0, \]
two-fold degenerate flat bands with \[\dim\ker_{L^2_{-p,2}}(D(\alpha)) =0 \text{ and } \dim\ker_{L^2_{-p,0}}(D(\alpha))=\dim\ker_{L^2_{-p,1}}(D(\alpha))=1, \]
or three-fold degenerate flat bands 
\[\dim\ker_{L^2_{-p,2}}(D(\alpha)) = \dim\ker_{L^2_{-p,0}}(D(\alpha))=\dim\ker_{L^2_{-p,1}}(D(\alpha))=1, \]
where for $u \in \ker_{L^2_{-p,2}}(D(\alpha))$, and for precisely one of $\{\pm z_S\}$,
\[ u(z)=zw(z) \text{ and } u(z)=(z\pm z_S)^2 v(z) \text{ for }w,v \in C^{\omega}(\CC).\]
In addition, we have in the three-fold degenerate case 
\[ \ker_{L^2_{-p,1}}D(\alpha) = \mathbb C \wp(z \pm z_S, \tfrac{4}{3}\pi i w, \tfrac{4}{3}\pi i w^2) u(z)\]
and 
\[ \ker_{L^2_{-p,0}}D(\alpha) = \mathbb C \frac{u(z)}{\wp(z \mp z_S, \tfrac{4}{3}\pi i w, \tfrac{4}{3}\pi i w^2)}.\]
\end{theo}

\subsection{Proof of generic simplicity}

Our proof of Theorem \ref{t:sim} is an adaptation of the argument for generic simplicity of resonances
by Klopp--Zworski \cite{klop} -- see also \cite[\S 4.5.5]{res}.

Note that we are in Case \ref{case1}, i.e., $\{-q,p,0\}\equiv \{p,0\} \bmod \ZZ_3$. The key is to study simplicity of the eigenvalues of 
 $ T_{0 } $ on a rotational decomposition of $L^2_{j}$.
Hence, let 
\[  R_{k}   := ( 2 D_{\bar z } - k )^{-1} : L^2_{ip } 
\to  L^2_{ip } \cap H^1. \]
In particular, since $-p \notin \{p,0\}$, we may set $k=0$ and see that
\begin{equation}
\label{eq:defR}  R := R_{\mathbf 0 } : L^2_{ ip }
\to  L^2_{ ip } \cap H^1 , \end{equation}
is well defined together with $ T_{0 } = R V $. 
We need to decompose (cf.~equation \eqref{eq:decomp2})
\[ L^2_{ip } = \bigoplus_{\ell=0}^2 L^2_{{-p,\ell}} , \ \ \ 
L^2_{{-p,\ell} } \simeq L^2 ( F ) ,\]
where $ F $ is a fixed fundamental domain of $L^2_{-p,\ell}$, i.e., $\CC$ quotient by the joint group action $\mathscr L$ and $\mathscr C$. 
For $ V \in \mathscr V $ and $\ell \in \ZZ_3$ 
\[    R V :  L^2_{{-p,\ell} } \to L^2_{{-p,\ell} } . 
 \] 
The next lemma shows that the eigenvalues of $RV$ restricted
to each of the three representations $L^2_{-p,\ell}$ with $\ell \in \ZZ_3$ are simple. This argument is analogous to the proof of \cite[Lemma $5.2$]{bhz23} and in our proof we shall only point out the main differences to the argument provided there. 
\begin{lemm}
\label{l:sim}
In the notation of \eqref{eq:Vscr} and \eqref{eq:defR}, 
there exists a generic subset $ \mathscr V_{\ell}$ of $ \mathscr V $ such that
for $ V \in \mathscr V_{\ell} $,  the eigenvalues of  $ RV |_{ L^2_{{-p,\ell} } } $
are  simple for each $\ell \in \ZZ_3.$
\end{lemm}
\begin{proof} 
The proof follows the same string of arguments as \cite[Lemma $5.2$]{bhz23} in the case of twisted bilayer graphene: Consider operators $RV$ acting on
$\mathscr H := L^2_{{-p,\ell}} $.
The condition that one has to verify is, in the notation of \cite[(5.14),(5.16)]{bhz23}, that for $w \in \ker_{L^2_{-p,\ell}}(T_k-\lambda) \setminus\{0\}$ and $\widetilde w \in \ker_{L^2_{-p,\ell+1}}(T_k^*-\bar{\lambda}) \setminus\{0\}$ 
\[ \forall \,  V \in \mathscr V \ \  \langle V {w } ,
R^*  \widetilde { w } \rangle = 0 \implies w \text{ or }\widetilde w=0.\]
In our case, this leads to the component-wise constraint (with $ \circ_j $ denoting components of $ \bullet = w, \widetilde{ w } $)
\begin{equation}
\begin{split}
\label{eq:UpUm}  & \langle U_+  w_2 , R^* \widetilde w_1 \rangle_{L^2(F)} + 
\langle U_- w_1 , R^* \widetilde w_2 \rangle_{L^2(F)}  + \langle Y_+  w_3 , R^* \widetilde w_2 \rangle_{L^2(F)} + 
\langle Y_- w_2 , R^* \widetilde w_3 \rangle_{L^2(F)}
= 0 , 
\end{split}
\end{equation} 
where $ F $ is a fundamental domain. 
Since $ V $ is arbitrary in $F$, 
this implies that $ { w}$ or $  R^* \widetilde { w }  \equiv 0 $ since the elements $w,\widetilde w$ are, by ellipticity of the $D(-1/\lambda)$ operator, real-analytic functions with full support componentwise. Since it is easy to see that $\widetilde w$ cannot be constant, we readily conclude that indeed $w$ or $\widetilde w \equiv 0.$
\end{proof}

Next we show that eigenvalues within different representations can be split. For twisted bilayer graphene this has been addressed in \cite[Lemma $6.3$]{bhz23}.
\begin{lemm}
\label{l:split}
Assume Case \ref{case1}. For $W$ with $[\mathscr L_a,W]=0$ and $[\mathscr C,W]=0,$ we have
\[  \Spec_{ L^2_{{-p,\ell } } } ( R W ) \cap D( \lambda_0 , 2r )  = \{ \lambda_0 \}  , \ \ \  \ell \in \ZZ_3 ,  \   r > 0 ,\]
and $ \lambda_0 $ is a simple eigenvalue of all $ RW|_{L^2_{{-p,\ell } } } $. 
Then, for every $ \varepsilon > 0 $ there exists
$ V \in \mathscr V_{\text{full}} $, $ \| V \|_{\delta } < \varepsilon$,  such that 
\begin{equation}
\label{eq:split} 
\Spec_{ L^2_{{-p,\ell } } } ( R ( W + V ) ) \cap D( \lambda_0 , r )  = \{ \widetilde \lambda_{\ell} \},  \ \ 
\widetilde \lambda_0 = \widetilde \lambda_1  \neq \widetilde \lambda_2 \end{equation}
or 
\begin{equation*}
 \Spec_{ L^2_{{-p,1 } } } ( R ( W + V ) ) \cap D( \lambda_0 , r )  =  \Spec_{ L^2_{{-p,2 } } } ( R ( W + V ) ) \cap D( \lambda_0 , r ) = \{\widetilde \lambda \}
\end{equation*}
with $R(W+V)u = \widetilde \lambda u$ for $u \in L^2_{-p,2}$, where for precisely one of $\{\pm z_S\}$
\[ u(z)=zw(z) \text{ and } u(z)=(z\pm z_S)^2 v(z) \text{ for }w,v \in C^{\omega}(\CC).\]
In addition, we have 
\[ \ker_{L^2_{-p,1}}(R(W+\lambda)-\widetilde \lambda) = \mathbb C \wp(z \pm z_S, \tfrac{4}{3}\pi i w, \tfrac{4}{3}\pi i w^2)u(z)\]
and
\[ \ker_{L^2_{-p,0}}(R(W+\lambda)-\widetilde \lambda) = \mathbb C \frac{u(z)}{\wp(z \mp z_S, \tfrac{4}{3}\pi i w, \tfrac{4}{3}\pi i w^2)}.\]
\end{lemm}
\begin{proof}
We start by observing that
\begin{equation}
\label{eq:for_mengxuan}
 \begin{split}
w_2(z\pm z_S) &= \bar \omega w_2(\omega(z \pm z_S)) = \bar \omega w_2 (\omega z \pm z_S \mp \gamma_2) \\
& =\bar \omega \operatorname{diag}(\omega^{\pm p},1 , \omega^{\mp q}) \mathscr L_{\mp \gamma_2} w_2(\omega z \pm z_S) \\
&=\bar \omega  \operatorname{diag}(\omega^{\pm p\mp j},\omega^{\mp j},  \omega^{\mp q\mp j}) w_2(\omega z\pm z_S).\end{split}
\end{equation}
If $q=0$ and $p=1$ then 
\[ \bar \omega  \operatorname{diag}(\omega^{ 2p},\omega^{p},  \omega^{- q+p}) = (\omega, 1,1)\]
showing that the zero at $z_S$ is at least of second order while at $-z_S$ we get $(1,\omega,\omega).$ To summarize, if $w_2$ vanishes at either $\pm z_S$ then it vanishes to second order.

As in \cite[Lemma $6.3$]{bhz23} we have  $ {w}_k , 
\widetilde {w}_k \in L^2_{j,k} $, such that $ \langle w_k, \widetilde {w}_k \rangle = 1$, and 
\[ ( 2  \lambda_0  D_{\bar z }   - W ) {w} _k = 0 , \ \ \  ( 2 \bar \lambda_0 
 D_z  - W^* ) R^* \widetilde {w}_k = 0 .\]
We can split an eigenvalue with eigenvectors $ w_k $,  if we can find
$ V $ such that (see \eqref{eq:UpUm} for the notation)
\[  \begin{gathered} \langle V w_0 , R^*\widetilde  {   w}_0 \rangle_{L^2 ( F ) }  \neq 
 \langle V w_2 , R^*\widetilde   {w}_2\rangle_{L^2 ( F ) }.
\end{gathered} \]
If for all (analytic) $V $ the terms were equal it would follow that 
\[(w_0)_{\nu} \overline{(R^* \widetilde w_0)}_{\ell} = (w_2)_{\nu} \overline{(R^* \widetilde w_2)}_{\ell}.\]
The left hand side vanishes at both $\pm z_S$ and therefore so does the right-hand side for every $\nu,\ell.$
If one component $(w_2)_{\nu}$ of $w_2$ does not vanish at $z_S$, then $R^* \widetilde w_2(z_S)=0$ and vice versa.
We conclude that $w_2$ vanishes at either $z_S$ or $-z_S$ and if not, then $R^* \widetilde w_2$ vanishes at both. 

If $w_2$ or $R^* \widetilde w_2$ vanishes at both $\pm z_S$, then that function has $5$ zeros, and we have split the eigenvalue as indicated in \eqref{eq:split}. As the theta function argument would then allow us to construct $5$ flat bands, this would contradict the assumption that $\ker_{L^2_{ip}}(D(\alpha))$ is three-dimensional.

Thus, we assume that neither of the two has a zero at both $\pm z_S$. We shall focus here on $w_2$ having a zero at $z_S$ or $-z_S$ (but not both) as the case for $R^* \widetilde w_2$ is completely analogous. By the symmetry argument at the beginning of this proof, the entire vector $w_2$ then vanishes at $\pm z_S$ to second order.

If $w_2$ vanishes at either $\pm z_S$ to second order, we can define 
\[v_1(z) = w_2(z)\wp(z \pm z_S, \tfrac{4}{3}\pi i w, \tfrac{4}{3}\pi i w^2) \in \ker_{L^2_{-p,1}}(D(\alpha)).\]
This function vanishes at $0$ to second order and at $\mp z_S$ to first order. Finally, 
\[v_0(z) = v_1(z)\wp(z, \tfrac{4}{3}\pi i w, \tfrac{4}{3}\pi i w^2) \in \ker_{L^2_{-p,0}}(D(\alpha))\]
vanishes to first order at $\pm z_S$ and to second order at $\mp z_S.$
\end{proof} 
We find from Theorem \ref{t:sim} that the flat bands in $\ker(D(\alpha)+k)$ of multiplicity $m$ are the vector spaces $V_m$ given for $w_i \in \ker_{L^2_{0,i}}(D(\alpha))$ by
\begin{equation*}
\begin{split}
V_1(k) &= \{ \zeta F_{k}(z) w_2(z);\zeta \in \CC \}\\
V_2(k) &= \{ \zeta_1 F_{k_1}(z+z_S)w_0(z) + \zeta_2 F_{k_1}(z-z_S)w_0(z); \zeta_1,\zeta_2 \in \CC \} \text{ and }\\
V_3(k) &= \{\zeta_1 F_{k_1}(z-z_S)w_0(z)+\zeta_2  F_{k_2}(z)w_1(z) + \zeta_3 F_{k_3}(z) w_2(z); \zeta_1,\zeta_2,\zeta_3 \in \CC \}.
\end{split}
\end{equation*}

We can now finish
\begin{proof}[Proof of Theorem \ref{t:sim}]
We restrict us to explaining the second part of the proof, as the first part is even simpler since it does not require the splitting of representations in Lemma \ref{l:sim}. We have simplicity of the spectrum of $ T_{0} $ on 
$ L^2_{ip } $, modulo the necessary multiplicity of coupled spectrum in $L^2_{-p,0},L^2_{-p,1}$. Using then Lemma \ref{l:sim} (strictly speaking its proof) 
and Lemma \ref{l:split}, we conclude
that for every $ r > 0 $, the set
\[  \mathscr V_r := \{ V :    R (W+V)  \text{ satisfies the conclusion of Lemma }\ref{l:split} \text{ in }\CC \setminus D(0,r) \} \] 
is open and dense. We then obtain $ \mathscr V_0 $ by taking the intersection of 
$ \mathscr V_{1/n } $.
\end{proof}

\section{Band touching for simple magic parameters}
\label{sec:Theta}
Throughout this section, we shall assume without loss of generality that $p \notin 3 \mathbb Z$ (see Remark \ref{rem:pnotzeromod3}). We recall that $\alpha$ is a simple magic parameter if $E_1(k,\alpha)=0$ for any $k$ but $E_2(k_0,\alpha)\neq 0$ for some $k_0$. This leads naturally to the question of band touching: namely, does there exist $K_0\in\CC/\Gamma_3^*$ such that $E_2(K_0,\alpha)=0$? It turns out this is impossible for twisted bilayer graphene: it was shown in \cite[Theorem  4]{bhz23} that there was a spectral gap for a simple or doubly-degenerate magic angle. By this, we mean that the first non-zero band doesn't touch the flat bands at any point. Interestingly, in the case of twisted trilayer graphene, the first two bands always touch. Nevertheless, if we suppose that the magic parameter $\alpha$ is simple, the point where the two bands touch is unique. This is the statement of the next theorem and is the main result of this section.
\begin{theo}
\label{theo:onetouching}
    For $\alpha\in\mathcal{A}$ simple, there exists a unique $K_0\in\{-ip,0,iq\}$ such that 
    $$\dim\ker_{L^2_0}(D(\alpha)+K_0)\geq 2, \quad \dim\ker_{L^2_0}(D(\alpha)+k) = 1$$ 
    for any $k\notin K_0+\Gamma_3^*$. In other words, the first two bands touch at only one point $K_0\in\CC/\Gamma_3^*$. Moreover:
    \begin{itemize}
    \item For Case \ref{case1}, we have $K_0=-ip$ if $p=-q$ and $K_0=0$ if $q=0$, according to Proposition \ref{prop:vanishing1}.
    \item For Case \ref{case2}, we have $K_0\equiv -i(m-1)$ where the integer $m$ is characterized in Proposition \ref{prop:vanishing2}.
     \end{itemize}
\end{theo}

We will split the proofs between Case \ref{case1} and \ref{case2}, see the end of the section for the proofs of Propositions \ref{prop:vanishing1} and \ref{prop:vanishing2}. The main technique is the theta function argument, explained in Section \ref{ss:theta_arg}. To use it properly, we first need to study the zeroes of eigenfunctions for a simple magic parameter.

\subsection{Vanishing of eigenfunctions for simple magic parameters}
We start with a simple lemma characterizing the zeros of zero modes of the Hamiltonian:
  \begin{lemm}
 \label{l:van}
 Suppose that $   w \in C^\infty ( \mathbb C; \mathbb C^2 ) $
 and  that $ ( D ( \alpha ) +   k )   w = 0 $ for some $   k $ and 
 that $   w ( z_0 ) = 0 $. Then  
 $     w ( z ) = ( z - z_0 )   w_0 ( z ) $, where $ 
   w_0 \in C^\infty ( \mathbb C; \mathbb C^2 ) $. 
 \end{lemm}
\begin{proof}
 Since the solution to the elliptic equation is (real) analytic, it suffices to show that $ (2 D_{\bar z})^\ell   w ( z_0) = 0$ for 
 all $ \ell $. Then, writing
 $  (2 D_{\bar z })^\ell     w ( z )  =  (2 D_{\bar z })^{\ell-1}  \left[ \left(    U -   k 
 \right)    w \right]  ( z ) ,$
 the claim follows by induction on $ \ell $. 
\end{proof}

Using the previous lemma, we find the order of vanishing of elements in $L^2_{r,\ell}$ for $\ell=0,1,2$:
\begin{lemm}
   \label{l:van0_2}
    Let $u \in \ker_{L^2_{r,\ell}}(D(\alpha))$ for some $r\in\ZZ_3$.
    \begin{enumerate}
        \item If $\ell=2$, then $u(z)= z w(z)$ for some $w\in C^\infty ( \mathbb C; \mathbb C^3 )$.
        \item If $\ell=1$, then $u(z)= z^2 w(z)$ for some $w\in C^\infty ( \mathbb C; \mathbb C^3 )$.
        \item If $\ell=0$, and $u(0)=0$, then $u(z)= z^3 w(z)$ some $w\in C^\infty ( \mathbb C; \mathbb C^3 )$.
    \end{enumerate}
\end{lemm}
\begin{proof}
For the first case, since $u(\omega z) = \bar\omega^2 u(z) = \omega u(z)$ we find $u(0)=0$. The result then follows from Lemma \ref{l:van}.

For the second case, we also have $u(0)=0$ since $u(\omega z) = \bar \omega u(z)$. Since by Lemma \ref{l:van}, we have $u(z)=zu_1(z)$ and $D(\alpha)u=0$ implies $D(\alpha)u_1=0$. Using the identity
 $$\bar \omega (zu_1(z))=\bar \omega u(z)=u(\omega z)=(\omega z)u_1(\omega z),$$
we obtain $u_1(\omega z)=\omega u_1(z)$. This implies $u_1(0)=0$. As $D(\alpha)(zu_1) = z D(\alpha)u_1$, applying Lemma \ref{l:van} again yields $u_1(z)=zu_2(z)$ for some $u_2\in C^\infty ( \mathbb C; \mathbb C^3 )$.

The third case follows from an analogous iteration but with one more step than the proof of the second case. 
\end{proof}

We record another useful lemma. 
\begin{lemm}
\label{l:van_nonzero}
    Let $ u\in \ker_{L^2_{k}} D(\alpha)$ for some $k \in \{0,-ip, iq\}$, i.e., at points in $\CC/\Gamma^*_3$ where there exists protected states. Then $u(0)=0$ implies $\dim \ker_{L^2_{k}}D(\alpha) \geq 2$.
\end{lemm}
\begin{proof}
    Note that $L^2_k=\bigoplus_{\ell=0}^2 L^2_{{r,\ell}}$ for $r=ik\in\ZZ_3$. Let $ u\in \ker_{L^2_{k}} D(\alpha)$ with $u(0)=0$. Note that there exists a protected state $\varphi\in \ker_{L^2_{r,0}} D(\alpha)$ as $k \in \{0,-ip, iq\}$. If $u, \varphi$ are linearly independent, the lemma is obviously true. 
    
    Otherwise $u=c\varphi\in L^2_{r,0}$ with $u(0)=0$, so by Lemma \ref{l:van0_2} we obtain $u=z^3w(z)$ for some $w(z)\in C^\infty ( \mathbb C; \mathbb C^2 )$. This means that
    \begin{equation*}
        \Tilde{u}(z) = \wp'(z, \tfrac{4}{3}\pi i w, \tfrac{4}{3}\pi i w^2)u(z)\in L^2_{r,0}, \ \ \hat{u}(z) = \wp(z, \tfrac{4}{3}\pi i w, \tfrac{4}{3}\pi i w^2)u(z)\in L^2_{r,2}
    \end{equation*}
    satisfies $\tilde u, \hat u \in \ker_{L^2_{k}}D(\alpha)$, where $\wp(z)$ as before is the $\Gamma_3$-periodic Weierstrass $p$-function with a double pole at zero and $\wp'(z)$ denotes its derivative. This implies that $\dim \ker_{L^2_{k}}D(\alpha) \geq 3$ in this case.
\end{proof}

We give a simple but useful lemma, especially in constructing flat bands and dealing with multiplicity. This reduces the multiplicity of flat bands to the number of zeroes of eigenfunctions.
\begin{lemm}
\label{lem:theta-indep}
    For $\{z_i\}_{ 1\leq i\leq N}\subset \CC$, consider $\{F_k(z-z_i):  1\leq i\leq N\}$. If there exists $k_0\in\CC\setminus \Gamma_3^*$ such that $\{F_{k_0}(z-z_i):  1\leq i\leq N\}$ is linearly dependent, then there exist $i\neq j$ such that $z_i-z_j\in \Gamma_3$.
\end{lemm}
\begin{proof}
    Note that $F_k(z)\in L^2(\CC/\Gamma_3)$. By \cite{bhz2}, $F_k(z)$ gives Green kernel of $2D_{\bar z}+k$ with
    $$(2D_{\bar z}+k)F_k(z)= c(k)\delta_0(z),\ \ c(k) = 2\pi i\theta(z(k))/\theta'(0) \neq 0 \text{ when } k\notin \Gamma_3^*.$$
    If there exists $k_0\in\CC\setminus \Gamma_3^*$ such that $\sum_{i=1}^N c_iF_{k_0}(z-z_i) =0$ for some non-trivial $\{c_i\}$, applying $2D_{\bar z}+k$ yields
    $$ \sum_{i=1}^N c_ic(k_0)\delta_0(z-z_i)=0,$$
    which may only be true if there exist $i\neq j$ such that $z_i= z_j\in\CC/\Gamma_3$.
\end{proof}

\subsection{Proof of theorem \ref{theo:onetouching} }
We first show that the second band always touches the first one. It will be convenient to split the proof between Case \ref{case1} and Case \ref{case2}, beginning with Case \ref{case1}. We will use the various lemmas proved in the previous subsection, and we will also get some additional information on the vanishing point of the eigenfunctions. Our focus on Case \ref{case1} is no coincidence as in Case \ref{case2} the location of the zero is not uniquely determined, see Figure \ref{fig:10}.
\begin{prop}
\label{prop:vanishing1}
    Let $\alpha\in\mathcal{A}$ be simple and assume Case \ref{case1}.
    \begin{enumerate}
        \item If  $-q\equiv p$ $\mathrm{mod}\ 3$, then the unique protected state $\psi(z)\in \ker_{L^2_{0,0}}D(\alpha)$ has a simple zero at $-pz_S+\Gamma_3$; there exists $u\in E = \mathrm{span}\{\varphi,\rho\} \setminus \{0\}$ such that $u(z)$ has a simple zero at $pz_S+\Gamma_3$, while any $v\in E$ linearly independent of $u$ has no zero.
        \item If $q\equiv 0$ $\mathrm{mod}\ 3$, then the protected state $\varphi(z)\in \ker_{L^2_{p,0}}D(\alpha)$ only vanishes simply at $pz_S+\Gamma_3$; there exists $u\in E = \mathrm{span}\{\psi,\rho\}\setminus \{0\}$ such that $u$ only vanishes simply at $-pz_S+\Gamma_3$, while any $v\in E$ linearly independent of $u$ has no zero.
    \end{enumerate}
In both cases $u$ is unique (up to a rescaling).    
\end{prop}
\begin{proof}
    We only show the first case; the second case is essentially the same. By the symmetry under $\mathscr L_{a}, a\in\Gamma_3$, we only consider points in $\CC/\Gamma_3$. Recall that by the proof of Proposition \ref{prop:wronskian}, there exists $u(z), u'(z) \in E$ such that at least one of $\psi(z), u(z)$ vanishes at $z_S$ and at least one of $\psi(z), u'(z)$ vanishes at $-z_S$.

    Recall that by the argument provided in \cite[Theorem 3]{bhz2} we know that if $\alpha\in\mathcal{A}$ is simple, then no eigenstate in $\ker_{L^2_{r,\ell}}D(\alpha)$ with $r,\ell\in\ZZ_3$ can vanish at $z_0\notin \{\pm z_S,0\}.$

    We also recall that none of the above eigenstates can vanish at $0$, leaving only $\pm z_S$ as possible zeros. Indeed, assume $u\in L^2_{r,0}$ for some $r\in\ZZ_3$ and $u(0)=0$. By Lemma \ref{l:van0_2}, $u(z) = z^3 w(z)$ for some $w\in\CI(\CC;\CC^3)$. 
    Using the theta function argument we construct
    \[\{F_{k+ir}(z-z_S)\wp(z) \tau(ir) u(z), F_{k+ir}(z+z_S)\wp(z) \tau(ir) u(z)\}\subset \ker_{L^2_0}(D(\alpha)+k)\]
    where $\wp$ is the Weierstrass $p$-function with $\Gamma_3$-periodicity, which has a double pole at zero and simple zeroes at $\pm z_S$. The function $\tau(ir)$ above has the mapping property $ L^2_{-ir}\to L^2_{0}$ by equation \eqref{eq:tau}. By Lemma \ref{lem:theta-indep}, $\dim \ker_{L^2_0}(D(\alpha)+k)\geq 2$, which gives a contradiction to $\alpha\in\mathcal{A}$ being simple. 

    Now we show that $\psi$ cannot vanish at $pz_S$. Assume $\psi(pz_S)=0$. Write $\psi=(\psi_1,\psi_2,\psi_3)^T$. Since $\psi\in L^2_{0,0}$, we have
    \begin{equation*}
        \psi_1(z+a) = \bar{\omega}^{p(a_1+a_2)}\psi_1(z), \ 
        \psi_2(z+a) = \psi_2(z), \  
        \psi_3(z+a) = {\omega}^{q(a_1+a_2)}\psi_3(z).
    \end{equation*}
    Define $v(z):= \psi(z+pz_S)$. Then $v(0)=0$, and
    \begin{equation*}
    \begin{split}
        \mathscr{C} v(z) 
        &= v(\omega z) = \psi(\omega z +pz_S) = \psi(z+\bar{\omega}pz_S) \\
        &= \psi(z+pz_S+p\gamma_1) = v(z+p\gamma_1) =  (\bar{\omega}^{p^2}v_1, v_2, \omega^{pq} v_3)^T.
    \end{split}        
    \end{equation*}
    As in the proof of Lemma \ref{l:van0_2}, this implies that $v(z) = z^2 w(z)$ for some $w\in\CI(\CC)$, i.e., $\psi(z) = (z - pz_S)^2 w(z-pz_S)$ has a double zero at $pz_S$. Using the argument in the previous step we conclude that $\alpha\in\mathcal{A}$ has at least multiplicity two, which is a contradiction. 

    Therefore, $\psi$ can only vanish at $-pz_S$, and only of order one. Indeed, if it has a zero of multiplicity two, the theta function argument would imply that $\alpha\in\mathcal{A}$ has at least multiplicity two, which is again a contradiction. The same argument shows that there exists $u(z)\in E = \mathrm{span}\{\varphi(z),\rho(z)\}$ such that $u(z)$ only vanishes at $pz_S$ of order one. 

    It remains to show that any $v\in E$ linearly independent from $u$ has no zero. As $E\subset L^2_{p,0}$, we can show as before that $v(z)$ can at most admit a simple zero only at $pz_S$, as otherwise we get a multiplicity two $\alpha\in\mathcal{A}$ using the theta function argument. This is in fact also impossible, as otherwise we have
    \[\{F_{k+ip}(z-pz_S)\tau(ip)u(z), F_{k+ip}(z-pz_S)\tau(ip)v(z)\}\subset \ker_{L^2_0}(D(\alpha)+k)\]
    and thus $\dim \ker_{L^2_0}(D(\alpha)+k)\geq 2$, which contradicts the fact that $\alpha\in\mathcal{A}$ is simple.
\end{proof}

We now turn to Case \ref{case2}. The fact that the two bands touch then essentially follows from lemma \ref{l:van_nonzero}. The point $K_0=-i(m-1)$ at which the bands touch is characterized by which eigenfunction vanishes at the stacking point $-z_S$. This integer $m$ is not made explicit and numerical evidence suggest that it depends on the magic parameter $\alpha$, see Figure \ref{fig:10}.

\begin{prop}
\label{prop:vanishing2}
    If $\alpha\in\mathcal{A}$ simple, then there exists $m\in \ZZ_3$ for which we can find $u\in L^2_{m,0}$ among $\{\varphi, \psi, \rho\}$ with $u(-z_S)=0$ and $u(-z_S)\ne0$. Moreover, we have
    $\dim \ker_{L^2_{0}} (D(\alpha)-i(m-1)) \geq 2$.
\end{prop}

\begin{proof}
    By Proposition \ref{p:inverse}, for $\alpha\in\mathcal{A}$ we have at least one of the protected states $\varphi,\psi,\rho$ vanishes at $z_S$ and at least one of $\varphi,\psi,\rho$ vanishes at $-z_S$. 
    
    We first note that no protected state can vanish at $\pm z_S$ simultaneously. Indeed, if there exists a single protected state, say $\psi\in L^2_{0,0}$ (construction for other protected states are similar), that vanishes at $\pm z_S$ simultaneously, then, using the theta function argument, we can for all $k$ construct
     \begin{gather*}
        v_{1,k}= F_{k}(z-z_S)\psi(z),\ v_{2,k}= F_{k}(z+z_S)\psi(z) \in \ker_{L^2_0}(D(\alpha)+k)
    \end{gather*}
    such that
    $\dim \ker_{L^2_0}(D(\alpha)+k)\geq 2$
    by Lemma \ref{lem:theta-indep}, which contradicts the assumption that $\alpha$ is simple. Therefore, one of $\{\varphi,\psi,\rho\}$ has to vanish at $z_S$ and another one has to vanish at $-z_S$. This proves the first part of the result. 
    
    Assume $u\in \{\varphi,\psi,\rho\}$ with $u\in L^2_{m,0}$ and $u(-z_S) = 0$. Using the theta function argument in the form of $G_k(z)$ (cf.~equation \eqref{eq:theta_a}), we define
    $$u_{k}(z):={G}_{k+im}(z+z_S)u(z) \in \ker_{L^2_k}D(\alpha),\ \forall k\in\CC.$$
    In particular, for $k=-i(m-1)$ we have by \cite[Lemma 3.2]{bhz2} that 
    \begin{gather*}
        u_{-i(m-1)}\in \ker_{L^2_{m-1,2}}D(\alpha),\ \ u_{-i(m-1)}(0)=0.
    \end{gather*}
    By Proposition \ref{prop:uialpha} we have $\ker_{L^2_{m-1,0}}D(\alpha)\neq \{0\}$, so there is a $v_{-i(m-1)}\in L^2_{-i(m-1)}$ with $v_{-i(m-1)}\ne u_{-i(m-1)}$. From the definition of $\tau(k)$ in \eqref{eq:tau} follows that
    $$
    v\in L^2_k,\ \ D(\alpha)v=0\implies \tau(-k)v\in L^2_0,\ \ (D(\alpha)+k)v=0.
    $$
    Applying this for $k=-i(m-1)$ to $u_{-i(m-1)}$ and $v_{-i(m-1)}$ we obtain the second part of the result.
\end{proof}

\begin{figure}
 \includegraphics[width=5cm]{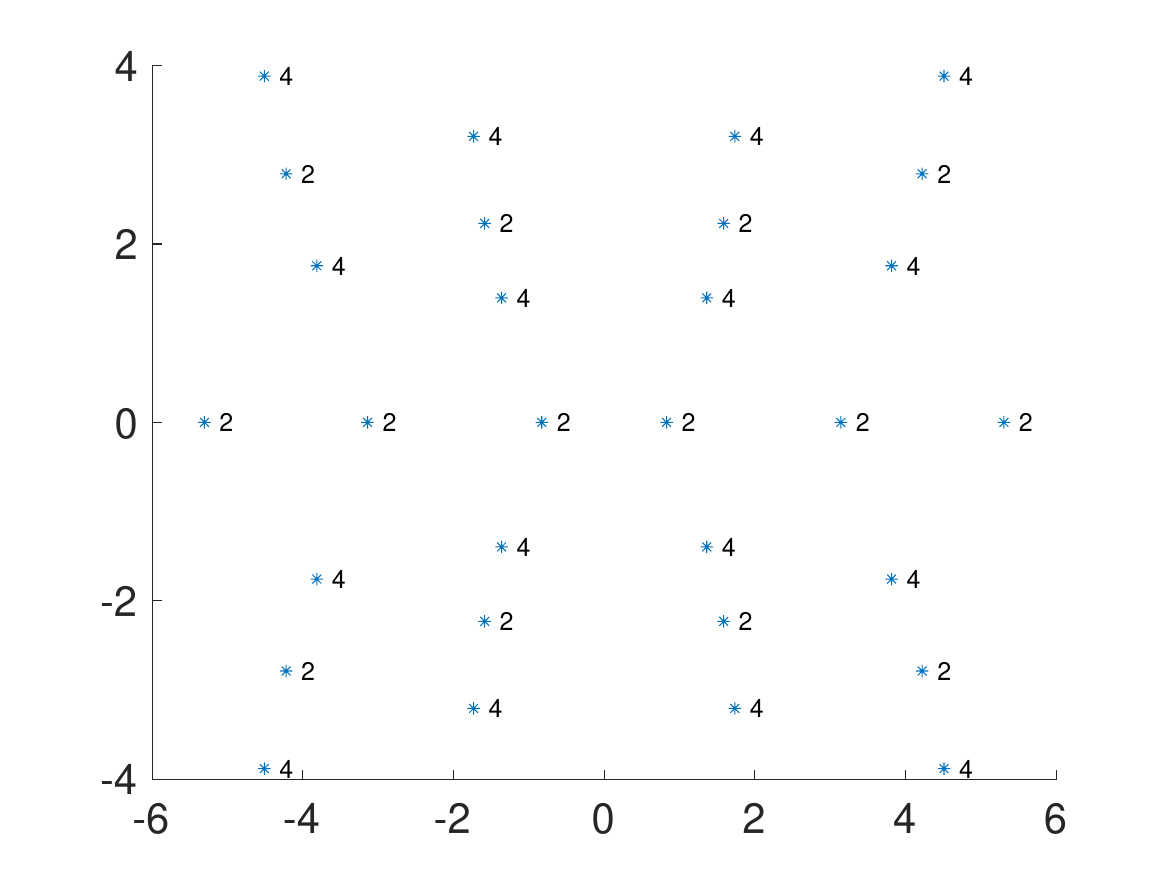}
  \includegraphics[width=5cm]{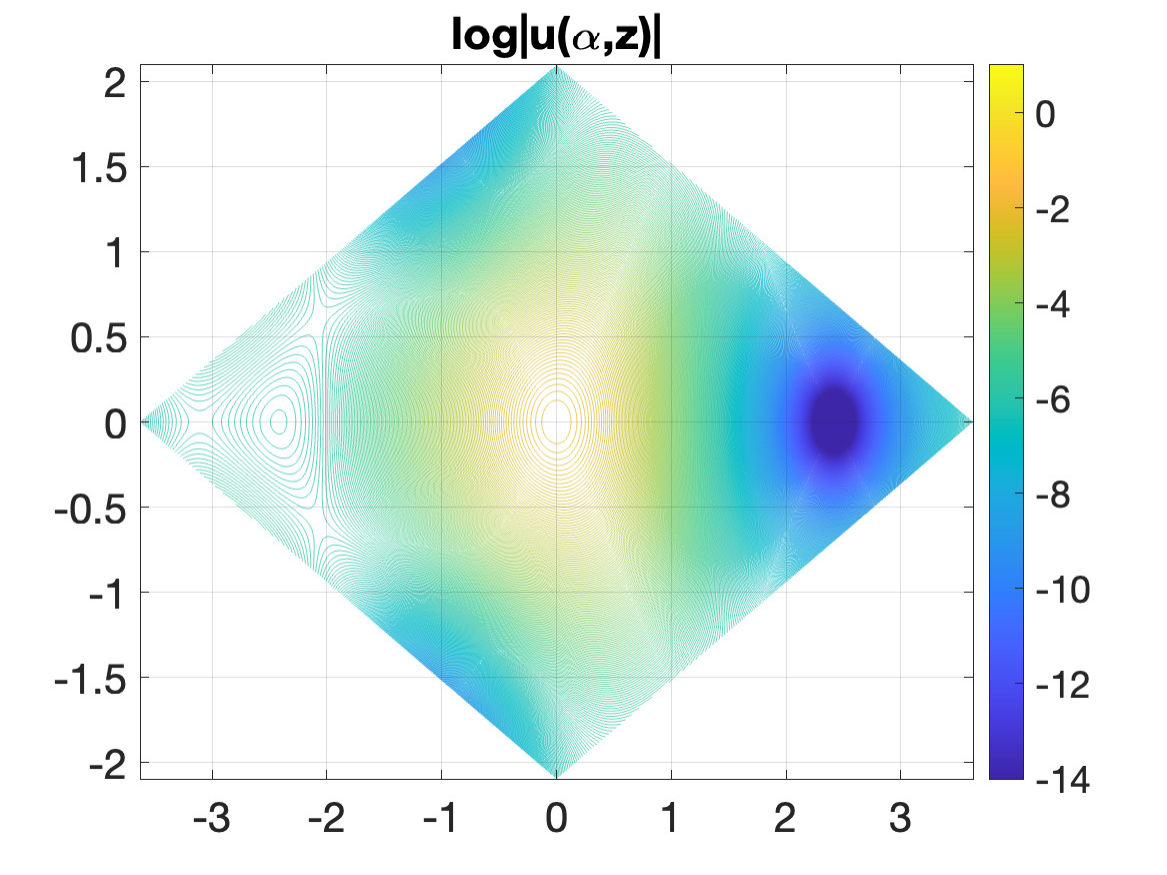}
 \includegraphics[width=5cm]{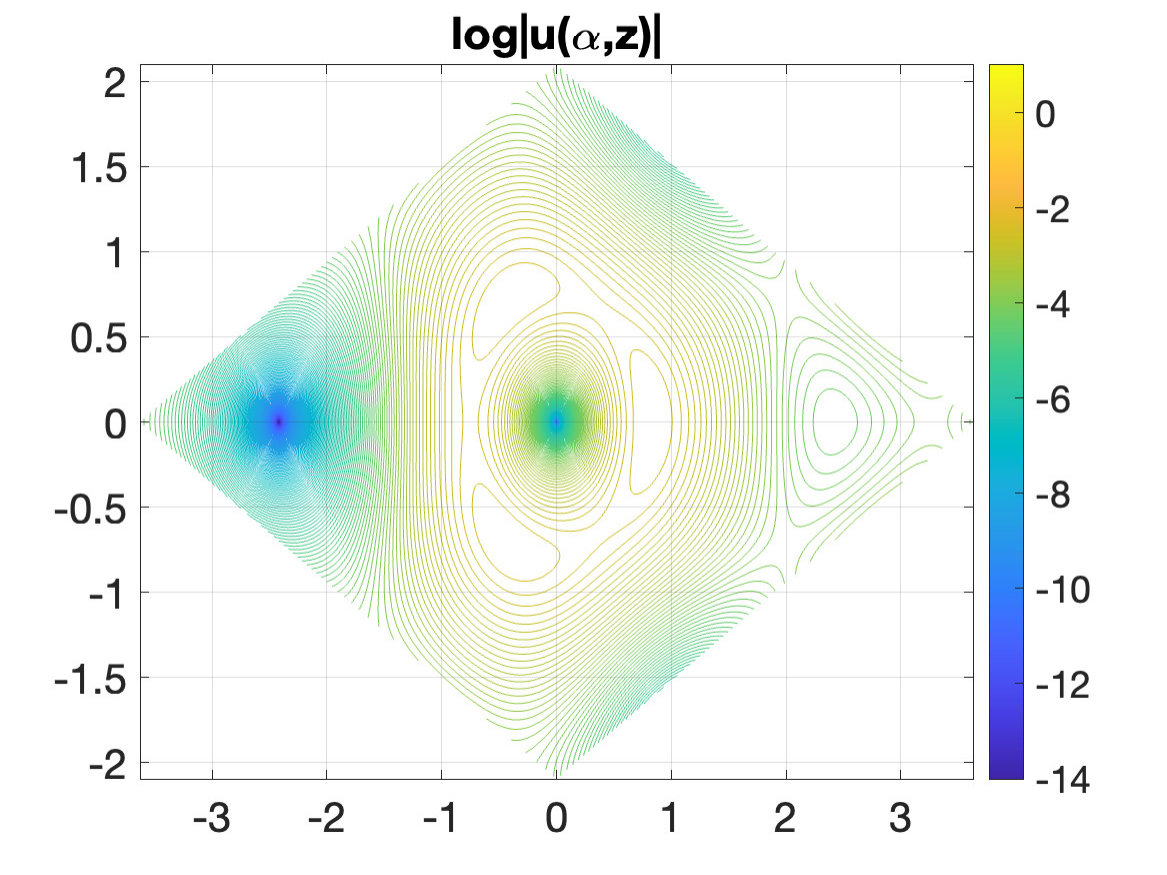}
  \caption{\label{fig:alldouble} For $\zeta=(1,1)$ and $\alpha_{12}=\alpha_{23}$ all magic parameters are at least two-fold degenerate. For the largest (with positive real part) magic parameter the two elements of $\ker_{L^2_{0,0}}(D(\alpha))$ and $\ker_{L^2_{0,2}}(D(\alpha))$ are shown.}
 \end{figure}

\begin{figure}
\includegraphics[width=7cm]{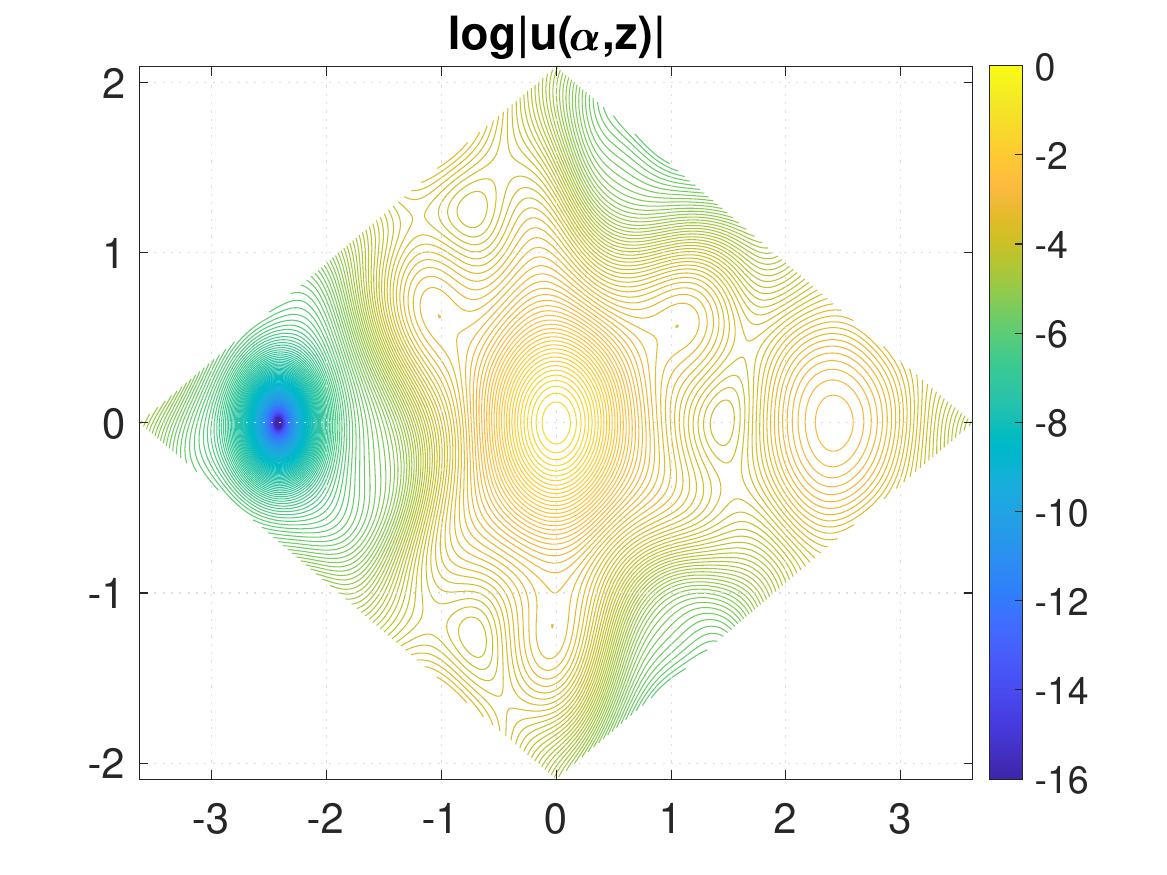}
\includegraphics[width=7cm]{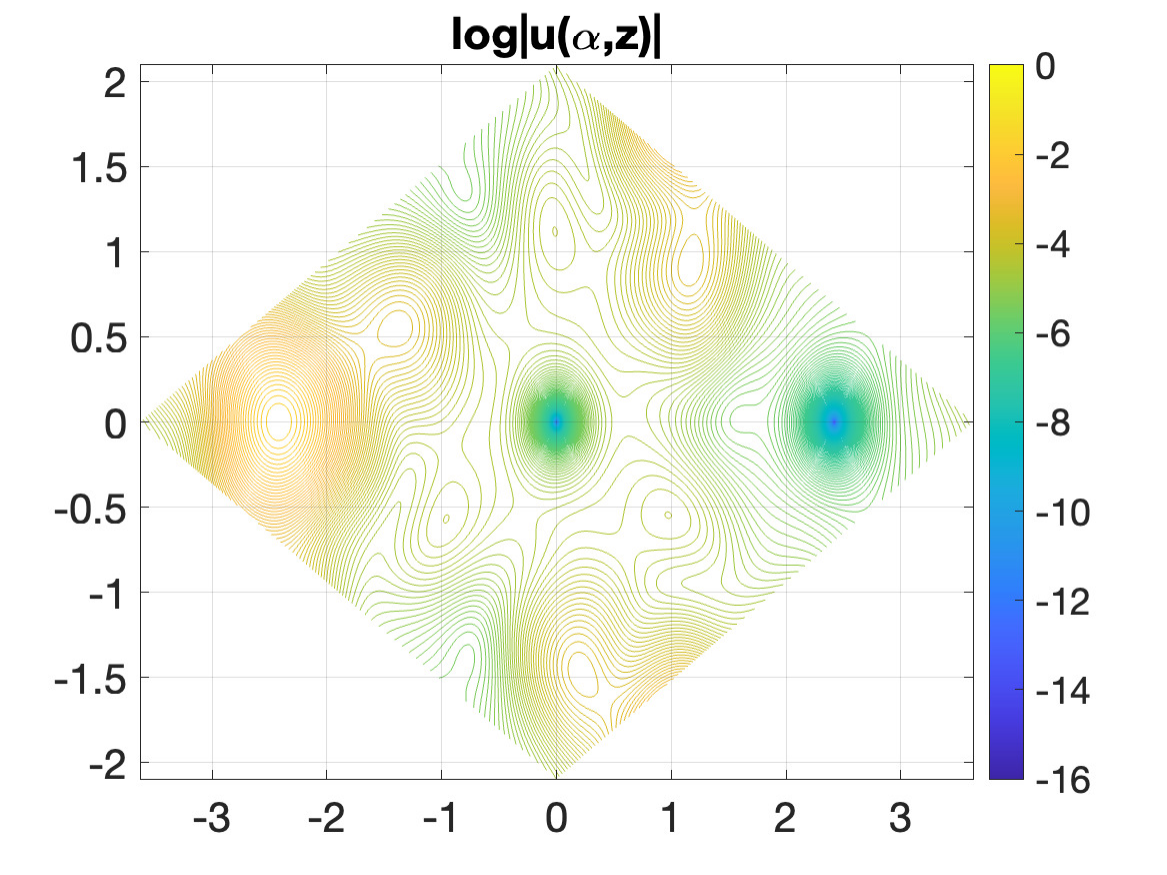}\\
\caption{Let $\alpha_{12}=\alpha_{23}.$ For $\zeta=(1,3)$ there are infinitely many magic parameters that are simple or two-fold degenerate, respectively. A contour plot of the largest two-fold degenerate zero modes with $\alpha \approx  1.12 + 0.65i$ in $\ker_{L^2_{0,0}}(D(\alpha))$ (left) and $\ker_{L^2_{0,2}}(D(\alpha))$ (right) is shown.}
\end{figure}

\begin{figure}
\label{fig:mnotconstant}
\includegraphics[width=7cm]{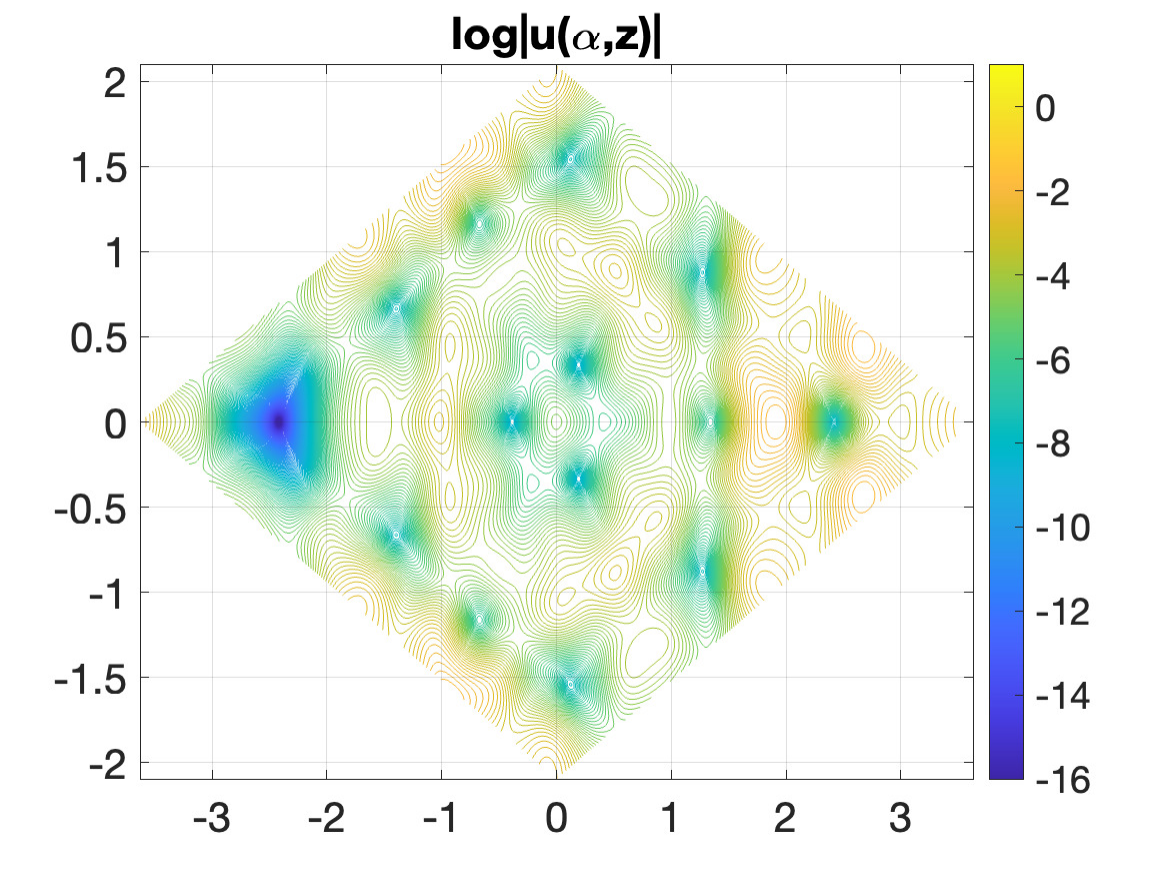}
\includegraphics[width=7cm]{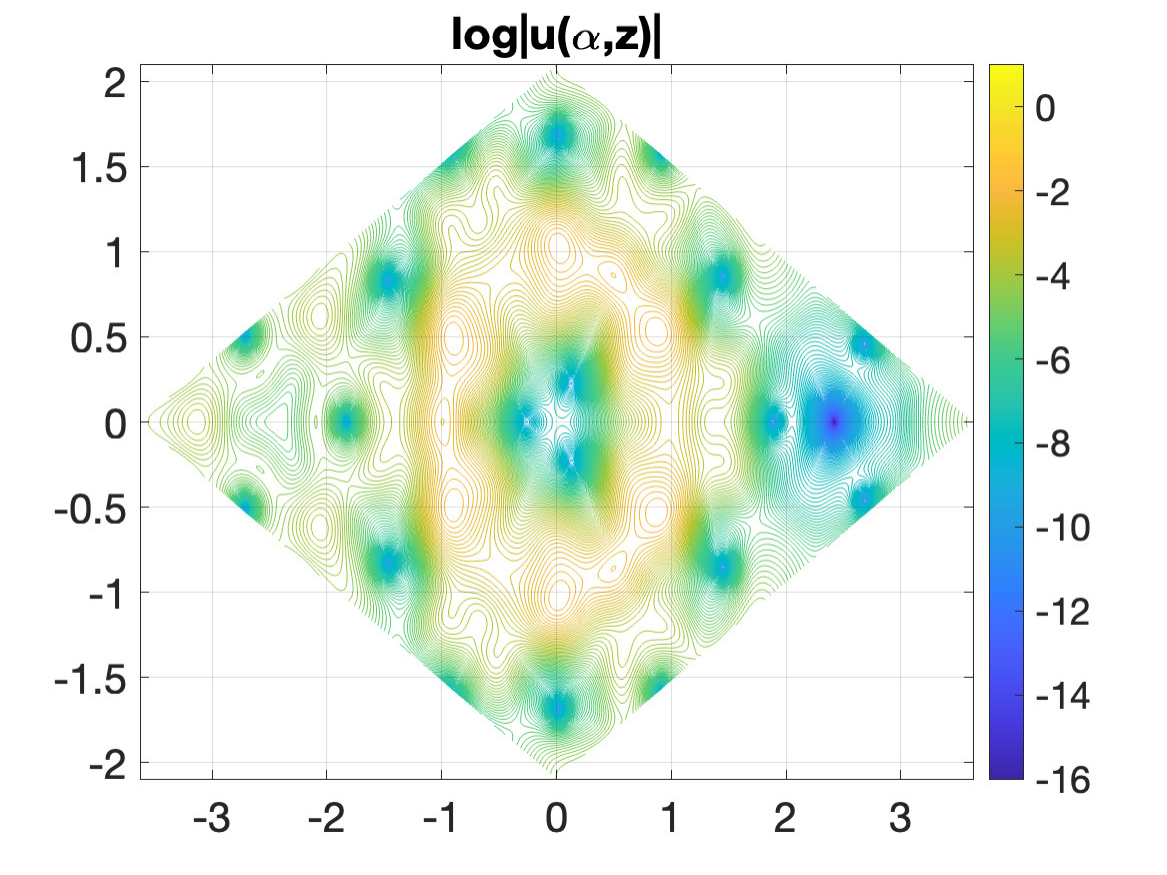}
\caption{\label{fig:10}Let $\alpha_{12}=\alpha_{23}$. For $\zeta_1 = 4$ and $\zeta_2 = 7$ we have simple magic parameters at $1.8999$ and $1.9288$ with simple element of $\ker_{L^2_{0,0}}(D(\alpha))$ illustrated. We observe that the unique zero appears at two different points $\pm z_S.$}
\end{figure}

\begin{figure}
\includegraphics[width=7cm]{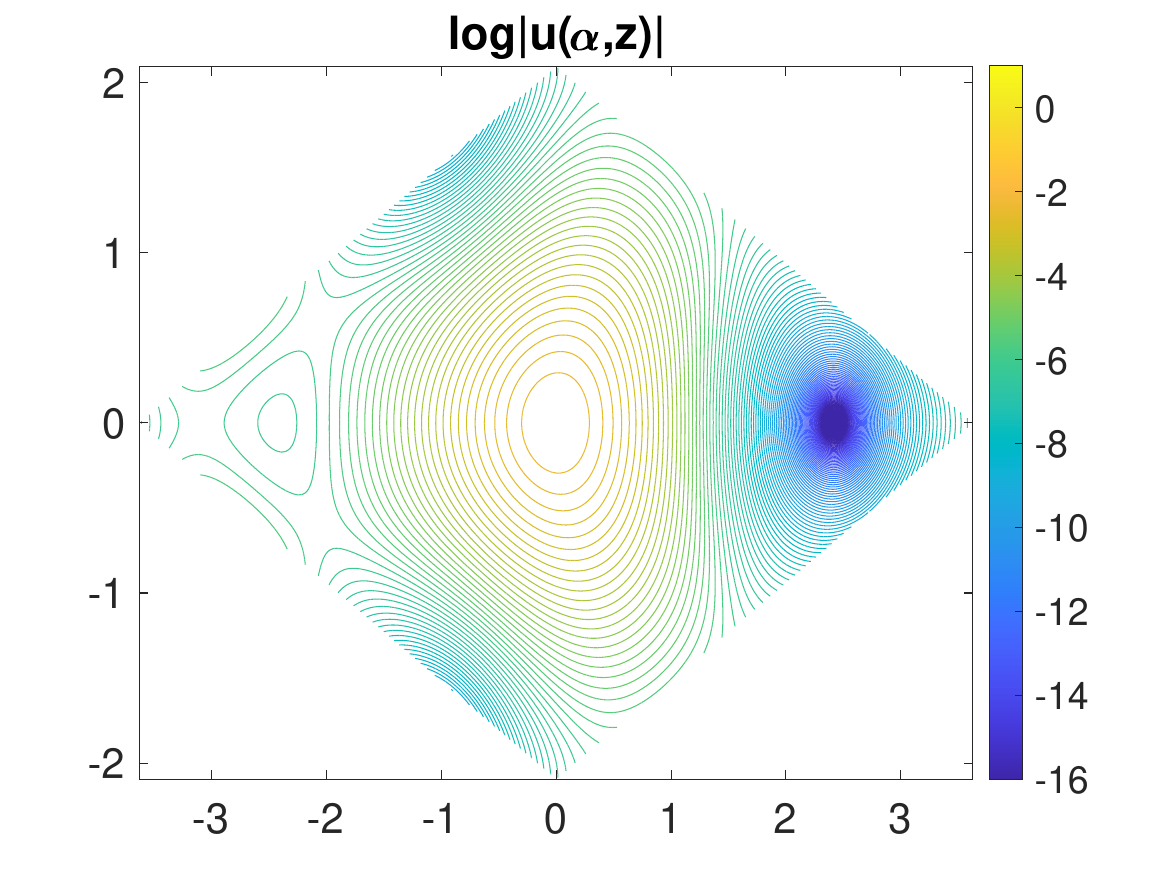}
\includegraphics[width=7cm]{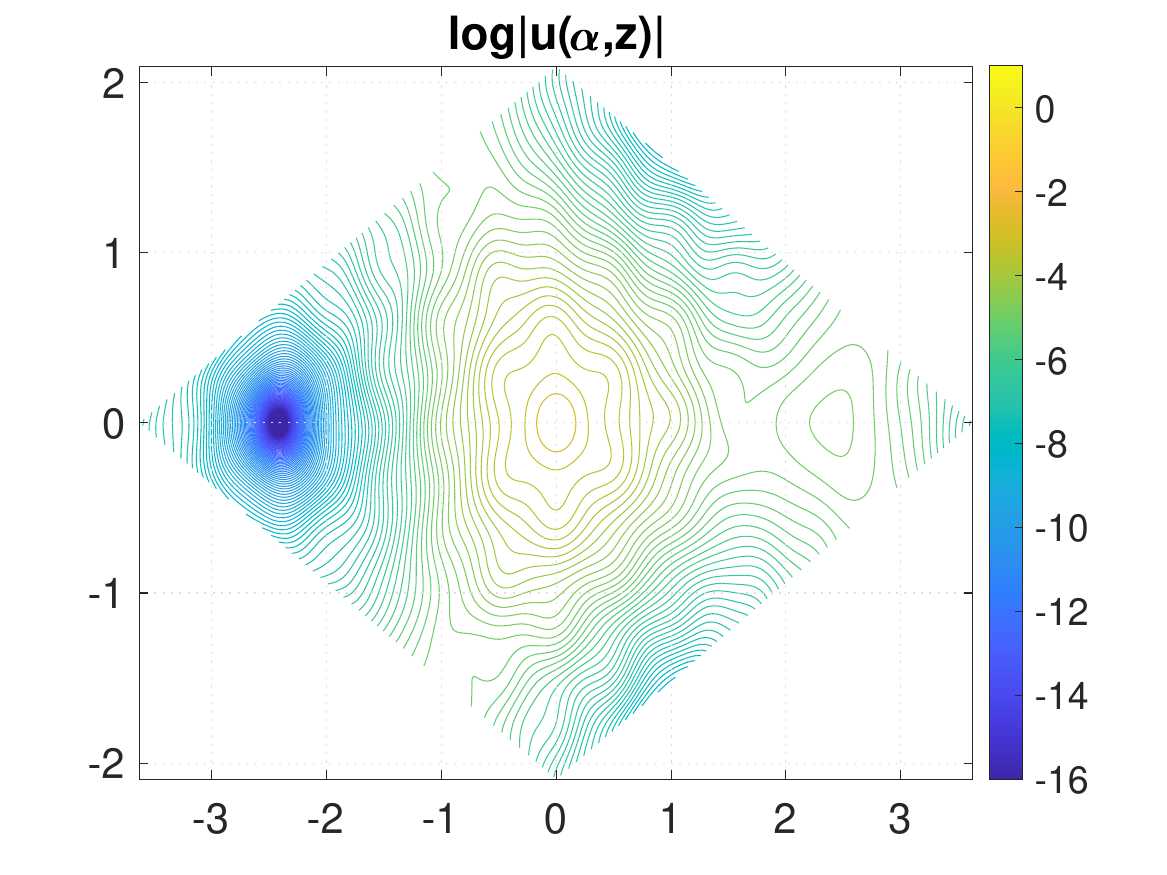}
\caption{Let $\alpha_{12}=\alpha_{23}$. For $\zeta = (1,1)$ (left) we have two-fold degenerate magic parameter at $\alpha \approx 0.82825$ with simple wavefunction in $\ker_{L^2_{0,0}}(D(\alpha))$ illustrated and for $\zeta = (1,7)$ (right) we have a two-fold degenerate magic parameter at  $\alpha \approx  1.0179+0.9281i$ with simple element of $\ker_{L^2_{0,0}}(D(\alpha))$ illustrated. There is a double zero at one of $\{\pm z_S\}.$}
\end{figure}

Proposition \ref{prop:uialpha} shows that the second band can touch the first band at at least one point and the location of this touching point is given by Proposition \ref{prop:vanishing1} and \ref{prop:vanishing2} for Case \ref{case1} and \ref{case2} respectively. We now show that in both Case \ref{case1} and \ref{case2}, when $\alpha\in\mathcal{A}$ is simple, the second band can only touch the flat band at a unique point among $\{-ip, 0, iq\}\subset \CC/\Gamma_3^*$ given by Proposition \ref{prop:vanishing1} and \ref{prop:vanishing2}.
\begin{prop}
\label{prop:uniqueness}
    For $\alpha\in\mathcal{A}$ simple, there exists a unique $K_0\in\{-ip,0,iq\}$ such that 
    $$\dim\ker_{L^2_0}(D(\alpha)+k) = 1$$ 
    for any $k\notin K_0+\Gamma_3^*$. 
\end{prop}

\begin{proof}
    Note that $\alpha\in \mathcal{A}$ is simple. Without loss of generality, we can work within $\CC/\Gamma_3^*$. By Proposition \ref{prop:vanishing1} and \ref{prop:vanishing2}, there exists $K_0 \in\{-ip,0,iq\}$ such that 
    $$\dim\ker_{L^2_0}(D(\alpha)+K_0) = 2.$$
    If there exists $k_0\neq K_0$ such that 
    $$\dim\ker_{L^2_0}(D(\alpha)+k_0) = 2,$$
    then we can take a basis $\{\varphi_1,\varphi_2\}$ of the space $\ker_{L^2_0}(D(\alpha)+k_0)$ such that $\varphi_1(z)$ has a unique zero $z_1\in \CC/\Gamma_3$ and $\varphi_2(z)$ does not vanish. By Proposition \ref{prop:vanishing1} and \ref{prop:vanishing2}, we can pick $\varphi_{K_0}\in \ker_{L^2_0}(D(\alpha)+K_0)$ such that $\varphi_{K_0}(z)$ does not vanish.
    
    We consider the Wronskian
    $$W = \det [\varphi_1, \varphi_2, \varphi_{K_0}].$$
    Note that the Wronskian satisfies
    \[ (2D_{\bar z}+2k_0+K_0)W = 0,\ \ W(z+\gamma) = e^{-i\langle \gamma, i(q-p)\rangle} W(z), \ \gamma = \tfrac{4\pi i}{3}(a_1\omega+ a_2 \omega^2).\]
    This yields $W\equiv 0$. Therefore there are $f_1(z), f_2(z)$ periodic with respect to $\Gamma_3$ such that 
    \begin{equation*}
        \varphi_{1} = f_1(z) \varphi_2 + f_2(z) \varphi_{K_0}. 
    \end{equation*}    
    Applying $D(\alpha)+k_0$ to the above equation we obtain
    $$D_{\bar z} f_1(z) \cdot \varphi_2 + (D_{\bar z}+k_0-K_0)f_2(z) \cdot \varphi_{K_0} = 0.$$
    If $\{\varphi_1, \varphi_2\}$ is linearly independent for all $z\in\CC/\Gamma_3$, then we must have have $D_{\bar z} f_1 = 0$ and  $(D_{\bar z}+k_0-K_0)f_2=0$. This yields $f_2\equiv 0$ as $D_{\bar z}+K_0-k_0: L^2(\CC/\Gamma_3) \to L^2(\CC/\Gamma_3)$ is invertible for $k_0\neq K_0$ and $f_1(z)$ is a meromorphic function over the torus $\Gamma_3$. As $\varphi_2$ has no zero and $\varphi_1$ is smooth, $f_1$ must be holomorphic therefore it is a constant. This contradicts to the linear independence assumption. Therefore, there exists some $z_*$ such that $c_1 \varphi_2(z_*)+ c_2 \varphi_{K_0}(z_*) =0$ for some $c_1, c_2$ non-zero, as neither $\varphi_2$ nor $\varphi_{K_0}$ vanishes.
    
    To obtain a contradiction, we can then define the flat band eigenfunctions $u_k\in\ker_{L^2_0}(D(\alpha)+k)$ as follows:
    \begin{equation*}
        u_k (z) = c_1 F_{k-k_0}(z-z_*)\theta(k-K_0)\varphi_2(z) + c_2 F_{k-K_0}(z-z_*)\theta(k-k_0)\varphi_{K_0}(z).
    \end{equation*}
    Note that this functions is smooth in $z$ even though $F_{k-k_\bullet}(z-z_*)$ has a pole at $z=z_*$. To see this, first similar to Lemma \ref{l:van}, we find that 
    $$c_1 \varphi_1(z)+ c_2 \varphi_2(z) = (z-z_*)w(z)$$
    for some $w\in\CI(\CC)$; then note the Laurent expansions
    \begin{gather*}
        G_{k}(z)\theta(k') = \frac{\theta(k)\theta(k')}{z\theta'(0)} + h_1(z), \ \  
        G_{k'}(z)\theta(k) = \frac{\theta(k)\theta(k')}{z\theta'(0)} + h_2(z)
    \end{gather*}
    where $h_1,h_2$ are holomorphic functions, which yields the same singularity at $z=0$. This will give rise to another flat band (with Chern number $-2$) in addition to the flat band (with Chern number $-1$) constructed from applying theta function argument to $\varphi$, which contradicts to $\alpha\in \mathcal{A}$ simple.
\end{proof}

\begin{rem}
    Along the same line of proofs of Propositions \ref{prop:vanishing1}, \ref{prop:vanishing2} and \ref{prop:uniqueness}, we can simplify the proof of \cite[Theorem 2, 3]{bhz2}.
\end{rem}

We conclude this section with a brief discussion about two-fold degenerated flat bands. The following rigidity theorem states that all two-fold degenerate magic angles give rise to flat band eigenfunctions in the same subspace. 
\begin{prop}
\label{prop:two-fold}
In Case \ref{case1}, if $\alpha$ is a two-fold degenerate magic parameter, then 
\begin{equation}
    \label{eq:two-fold}
    \dim\ker_{L^2_{-p,0}}(D(\alpha)),\dim\ker_{L^2_{-p,1}}(D(\alpha))\ge 1 \text{ and }\dim\ker_{L^2_{-p,2}}(D(\alpha))=0.
\end{equation}
\end{prop}
\begin{proof}
If $\operatorname{dim}\ker_{L^2_{-p,1}}(D(\alpha))=0,$ then $\operatorname{dim}\ker_{L^2_{-p,2}}(D(\alpha))\ge 2.$ Let $u_1,u_2$ be two independent elements in  $\ker_{L^2_{-p,2}}(D(\alpha))$. By Lemma \ref{l:van0_2}, they both vanish at $0$. Following \eqref{eq:for_mengxuan}, we can show that two of the three component of $u_1,u_2$ vanishes at $\pm z_S$. Therefore, we can choose a linear combination of the elements $u_1,u_2$ such that $w = \lambda_1 u_1 + \lambda_2 u_2$ vanishes at $0$ and $\pm z_S$. This allows us to define 
$$v(z) = w(z)\wp(z \pm z_S, \frac{4}{3}\pi i w, \frac{4}{3}\pi i w^2) \in \ker_{L^2_{-p,1}}(D(\alpha)).$$ 

This however implies that $v$ vanishes to second order at $\pm z_S$ and also at $\mp z_S$ by \eqref{eq:for_mengxuan}. This however implies that the flat band has at least multiplicity three by the usual theta function argument, which gives a contradiction to double degeneracy.
\end{proof}
\begin{prop}
\label{prop:gap}
In Case \ref{case1}, if $\alpha$ is a simple\footnote{In the case where $\alpha$ is simple, this proposition is a special case of Theorem \ref{theo:onetouching}. We keep this case in the result of the proposition as the proof is different and uses the operator $T_k$, and is thus closer to the spirit of this paper.} or two-fold degenerate magic parameter, then the flat bands can only touch other bands at the protected points $\{-ip, 0\}+\Gamma_3^*$. In particular, we have equality in equation \eqref{eq:two-fold} of Proposition \ref{prop:two-fold}.
\end{prop}
\begin{proof}
With out loss of generality, we consider $k\in\CC/\Gamma_3^*$. Assume that there is a band touch at $k_0\notin \{-ip, 0\}$. 
We start by giving this argument for simple magic parameters: 
This implies that the geometric multiplicity of the eigenvalue $-1/\alpha$ of $T_{k_0}$ is at least two. Recall that the algebraic multiplicity is independent of $k$ away from protected points. Indeed, $ k \mapsto T_{k}|_{L^2_{p }} $ is an analytic family of compact operators whose spectrum is independent of $k$. The algebraic multiplicity is thus independent of $k$ since it coincides with the rank of the spectral projection. This implies that it is at least two everywhere. 
Notice that the geometric multiplicity for $k=0$ on $L^2_{-ip}$ has to be one, as there would otherwise exist a flat band of higher multiplicity by Proposition \ref{prop:two-fold}. Hence, let $u \in \ker(T_0+1/\alpha)$ and $v$ the generalised eigenvector with $(T_0+1/\alpha )v =u .$ Thus, by multiplying by $2D_{\bar z}$ we find $D(\alpha)v(z)=2D_{\bar z}u(z).$ Defining $v_k(z)=F_k(z)v(z)$ and $u_k(z):=F_k(z)u(z) \in \ker(T_k+1/\alpha)$ we see that $D(\alpha)v_k(z)=2D_{\bar z}u_k(z)$ and thus $v_k \in \ker((T_k+1/\alpha)^2) \setminus \ker(T_k+1/\alpha).$ This implies that the algebraic multiplicity of $T_{k_0}$ is at least three at $-1/\alpha$. Iterating this argument, we see that the algebraic multiplicity is infinite which is a contradiction to the compactness of $T_k$. 

For two-fold degenerate magic parameters, the argument is similar. Touching of bands at some $k=k_0$ leads to a algebraic multiplicity of at least three. Specializing to $k=0$, it suffices now to argue again that this implies the existence of a generalized eigenvector in $L^2_0$. Since the band is assumed to be two-fold degenerate in the first place, we conclude that $\ker_{L^2_{0,2}}=\{0\}$ by arguing as in Proposition \ref{prop:two-fold}. If $\operatorname{dim}\ker_{L^2_{0,1}}(T_k+1/\alpha)=2.$ In this case, we have two independent elements $u_1,u_2 \in \ker_{L^2_{0,1}}(D(\alpha))$ that vanish to second order at $0$. By choosing a suitable linear combination we can thus construct a third solution $u_3$ that vanishes to second order at $0$ and to first order at $z_S.$ Applying the theta function argument, this implies that the flat band has higher multiplicity. A similar argument shows that $\dim\ker_{L^2_{0,0}}(D(\alpha))=1.$ We thus conclude that there exists a generalized eigenvector in $L^2_{0,0}$ and $L^2_{0,1}$. The rest of the argument proceeds then as above. 
\end{proof}

\section{Chern number}
 \label{sec:Chern}
To further study the structure of flat bands of the Hamiltonian \eqref{eq:original}, we compute here the Chern numbers of some vector bundles associated to magic parameters. For a general discussion on Chern connection, we refer to \cite[\S 8]{bhz23}. For our use, we recall that given a vector bundle $\pi:E\to X$, there is a notion of connection that is naturally defined if the bundle is holomorphic and equipped with a hermitian connection: the \emph{Chern connection} $D_C: C^{\infty}(X;E)\to C^{\infty}(X;E\otimes T^*X)$. For that we choose a local holomorphic trivialization $ U \subset X $, $ \pi^{-1} ( U ) \simeq U \times \mathbb C^n $, for which the hermitian metric is given by 
\begin{equation}
\label{eq:defGz}  \langle \zeta , \zeta\rangle_k = \langle G ( k ) \zeta , \zeta \rangle = \sum_{i,j=1}^n G_{ij}(k) \zeta_i \bar \zeta_j  \ \ \zeta \in \mathbb C^n , \ \ 
k \in U . \end{equation}
We see that if $ \{ u_1 ( k ) , \ldots , u_n ( k ) \} \subset \mathscr H $ is a basis of $ E_k $ for $ k \in U $, 
and $ U \ni k \to u_j ( k ) $ are holomorphic, then $ G ( k ) $ is the Gramian matrix:
\begin{equation}
\label{eq:gram} G ( k ) := \left( \langle u_i ( k ) , u_j ( k ) \rangle_{\mathscr H} \right)_{ 1\leq i, j \leq n } . 
\end{equation}
If $ s : X \to E $ is a section, then the Chern connection $
D_C : C^\infty (X;  E ) \to C^\infty (X ; E \otimes T^* X) $, over $ U $ is given by (using only 
the local trivialization and \eqref{eq:defGz})
\begin{equation*}
\begin{gathered}    D_C s ( k ) := d s ( k ) +\eta_C ( k )  s ( k )  , \\ \eta_C( k )  :=  G(k)^{-1} \partial_k G ( k ) \, dk \in
C^\infty ( U ,\Hom( \mathbb C^n , \mathbb C^n ) \otimes (T^* U)^{1,0} ) .
\end{gathered}
\end{equation*}
Here  $ \partial_k $ denotes the holomorphic derivative and the notation $ (T^* U)^{1,0} $ indicates that only $ dk $ and not $ d \bar k $ appear in the 
matrix valued 1-form $ \eta_C $,  $ \eta_C = \eta_C^{1,0} $.

The Chern curvature is then defined to be $\Theta:=D_c\circ D_C$ and the Chern class (which, in the case of a bundle over $\mathbb C/\Gamma_3^*$, is an integer that we will refer to as the \emph{Chern number}) is 
\begin{equation}
    \label{eq:Chern Number}
    c_1(E):=\frac{i}{2\pi}\int_{\mathbb C/\Gamma_3^*}\tr(\Theta)\in \mathbb Z.
\end{equation}

In the case of TBG, for a simple or doubly degenerate magic angle $\alpha$, \cite[Theorem 4]{bhz23} shows that there is a spectral gap. This means that if $\alpha$ has multiplicity $m\in\{1,2\}$, then $\forall k\in \mathbb C,$ one has $\mathrm{dim}(\mathrm{ker}_{L^2_0}(D(\alpha)+k))=m$. It is thus natural to consider the vector bundle $\pi: \tilde E\to \CC$ where
$$\tilde E:=\{(k,v): v\in V(k):=\mathrm{ker}_{L^2_0}(D(\alpha)+k)\}. $$
In this case, the Chern number can be explicitly computed and shown to be equal to $-1$ (see \cite[Theorem 4]{bhz2} for the simple case and \cite[Theorem 5]{bhz23} for the doubly degenerate case).

Since the flat bands in TTG are not uniformly gapped away from the rest of the spectrum, see for instance Theorem \ref{theo:onetouching}, it is not directly possible to associate a Chern number to
\[ V(k) = \ker(D(\alpha)+k)\]
as done in TBG, since $V$ changes its dimension for particular $k$.

We will use the general construction of a holomorphic vector bundle over a torus $\CC/\Gamma_3^*$ recalled in \cite[\S 8]{bhz23}. We will use the following two functions $e_a$ and $\tau(a)$ which are useful when working with theta functions (note that this $\tau$ is different from \eqref{eq:tau} but this should cause no confusion since \eqref{eq:tau} is not used from here on):
\begin{equation*}
\begin{gathered}
\left[ \tau(a) f  \right]( \zeta )  = e^{ - \frac{i}2(\bar \omega \zeta+\bar\zeta)a } e^{-2\pi i  a_1 \frac{3\zeta}{4\pi i \omega}} f ( \zeta ) , \ \ \ 
e_{ a } ( k ) = e^{ \pi i a_1^2 \omega + 2 \pi i  a_1 \frac{k}{\sqrt 3\omega} } (-1)^{a_1 -a_2 }  , \\
   a= \sqrt 3(\omega^2 a_1 - \omega a_2) , \ \ a_j \in \mathbb Z , \ \ k \in \mathbb C . 
\end{gathered}    
\end{equation*}
We observe that $\tau(a)$ is multiplication by an expression that has modulus $1$.
The importance of these two functions lies in the relation
\begin{equation}
\label{eq:nice_relation}
F_{k+a}(z)=\tau(a)e_a(k)^{-1}F_k(z). 
\end{equation}
 Moreover, we define the following equivalence relation, on $\CC\times L^2_{-ir}$: 
\begin{equation*}
   [ k , u]_\tau = [ k', u']_\tau \ \Leftrightarrow \
( k , u ) \sim_\tau ( k', u' ) \ \Leftrightarrow \  \exists \, a \in \Gamma_3^*,  \ \
k' = k + a , \ \  u' = \tau ( a ) u 
\end{equation*}
\begin{lemm}
\label{lemm:vector bundle}
Let $r\in \mathbb Z_3$ be fixed and let $\CC \ni k\mapsto V(k)\subset \mathrm{dim}_{L^2_{-ir}}(D(\alpha)+k) $. Here, the subspace $V(k)$ is given by
$$V(k):=\mathrm{Span}\{F_k(z+z_i)u:  1\leq i\leq m\}, \quad m\leq 3, $$
with $m$ being the number of zeros $z_i\in\CC/\Gamma^*$ of $u$. We define the trivial vector bundle over $\CC$:
$$\tilde E:=\{(k,v):  k\in \CC, \ v\in V(k)\}\subset \CC\times L^2_{-ir}. $$
Then we get an $m$-dimensional holomorphic vector bundle over the torus $\CC/\Gamma_3^*$ by defining
\begin{equation}
    \label{eq:E}
E:=\left\{ [ k , u ]_\tau  \in  \mathbb C \times V(k) /
\sim_\tau  \right\}=\tilde E/\sim_{\tau}\to \CC/\Gamma^*_3.
\end{equation}
Moreover, the associated family of multipliers is given by $k\mapsto e_a(k)$.
\end{lemm}

\begin{proof}
The proof is the same as \cite[Lemma 5.1]{bhz2}.
The action of the discrete group $ \Gamma_3^* $, with 
$ \Gamma_3^*\ni a : ( k , u ) \mapsto ( k + a , \tau ( a ) u ),$ on the trivial vector bundle $\tilde E$ is free and proper. Thus, its quotient by the action is a smooth complex
manifold. To show that it defines a holomorphic vector bundle, we use \cite[Appendix B]{bhz2} to reduce this to checking the relation
\begin{equation*}   
e_{a+a' } ( k ) = e_{ a' } ( k + a ) e_{a}( k) , \ \ 
a, a' \in \Gamma_3^*.
\end{equation*}
This last relation follows from an explicit computation.
\end{proof}
In the rest of the section, we aim to compute the Chern number associated to the vector bundle of flat bands in three cases: 
\begin{enumerate}
    \item The magic parameter $\alpha$ is simple.
    \item The magic parameter $\alpha$ is doubly-degenerate and we are in case \ref{case1}.
    \item The magic parameter $\alpha$ is three-fold degenerate in case \ref{case1} with each eigenfunction having three zeroes.
\end{enumerate}
This distinction appears in the generic statement of Theorem \ref{t:sim}.
\subsubsection{Simple magic angles}
\label{subsec:1}
In case of a simple magic parameter, it follows by Propositions  \ref{prop:vanishing2} 
 and \ref{prop:uniqueness} that there is an $m$ and $u \in \ker_{L^2_{m,0}}(D(\alpha))$ with a simple zero $u(-z_S)=0$ such that
\begin{equation*}
\label{eq:V}
    V(k) = \{\zeta F_k(z+z_S)u(z);\zeta \in \CC \}
\end{equation*} 
and $V(k) = \ker_{L^2_{-im}}(D(\alpha)+k)$ for $k \notin K_0+\Gamma_3^*.$ 

\subsubsection{Two-fold degenerate magic angles; Case \ref{case1}}
\label{subsec:2}
Assuming Case \ref{case1}, it follows from Propositions \ref{prop:two-fold} and \ref{prop:gap} that there are $u_0 \in L^2_{-p,0}$ and $u_1 \in L^2_{-p,1}$ such that for $k \notin \Gamma^*$,
\begin{equation*}
    V(k) = \{\zeta_0 F_k(z+z_S)u_0(z)+ \zeta_1 F_k(z-z_S)u_0(z);\zeta_0,\zeta_1 \in \CC \}
\end{equation*} 
and $V(k) = \ker_{L^2_{-ip}}(D(\alpha)+k)$ for $k \notin K_0+\Gamma_3^*.$ 

\subsubsection{Three-fold degenerate magic angle; Case \ref{case1}}
\label{subsec:3}
We now consider the case of $u_0 \in \ker_{L^2_{-p,0}}(D(\alpha))$, $u_1 \in \ker_{L^2_{-p,1}}(D(\alpha)),$ and $u_2 \in \ker_{L^2_{-p,2}}(D(\alpha))$ as described in Lemma \ref{l:split} with each one of them having three zeros (counting multiplicities). 
In this case, we may assume $u_2$ to have a zero of order $2$ at $z_S$ and a simple one at $0$.  We then define $v(z):=F_{-\tfrac{1}{4}}(z-z_S)F_{\tfrac{1}{4}}(z-z_S)u_2(z)$ with three disjoint simple zeros at $z_{\pm}:=z_S \pm \frac{\pi i}{3\sqrt{3}}$ and $0$ such that for $k \notin \Gamma^*$,
\[V(k)=\{  \zeta_0 F_k(z)v(z)+ \zeta_1 F_{k}(z-z_+)v(z) + \zeta_2 F_k(z-z_-)v(z); \zeta_0,\zeta_1,\zeta_2 \in \CC \}. \]

As done above in the choice of $v$, we can always split degenerate zeros using functions $F_k.$ Thus, we have the following general result.

\begin{theo}
\label{theo:Chern}
Let $V(k) = \operatorname{span}\{F_k(\bullet+z_i) u: 1 \le i \le m\}$ where $u \in \ker_{L^2_{-ir}}(D(\alpha))$ with $m =\operatorname{dim}(V(k))\leq 3$ zeros at distinct $-z_i \in \CC/\Gamma^*$, with $r \in  \ZZ_3$ fixed. Then the curvature $\tr(\Theta) = H(k) d\bar k \wedge dk$ satisfies $H\ge 0$ and the Chern number of the $m$-dimensional vector bundle defined in \eqref{eq:E} is $-1.$
\end{theo}
An immediate consequence of this result is
\begin{corr}\label{corr:Chernperturbations}
Let $\alpha$ be a simple magic angle or a two-fold degenerate magic angle with Case \ref{case1} as described in Paragraphs \ref{subsec:1} and \ref{subsec:2}, respectively. Then the Chern number of the associated flat band vector space $V(k)$ is equal to $-1.$

Moreover, consider Case \ref{case1}, then there is a generic set of potentials $\mathscr V_0\subset \mathscr V_{\text{full}}$, such that the flat bands of $D(\alpha)+k$ with tunnelling potential $V \in \mathscr V_0$, as described in Section \ref{sec:gen_simpl}, are described by vector spaces $V(k)$ outlined in Paragraphs \ref{subsec:1}-\ref{subsec:3} and the Chern number of the associated bundle is $-1$.
\end{corr}
\begin{proof}[Proof of Theorem \ref{theo:Chern}]
The computation for the case of simple bands is presented in \cite{bhz2} and for degenerate ones in \cite{bhz23}. Our argument here generalizes all these computations by fixing, without loss of generality, a Bloch function $u$ with $n$ disjoint simple zeros.
We define the Gramian matrix 
\[ G_{lm}(k) := \langle F_k(z+z_l)u,F_k(z+z_m)u \rangle. \]

Thus, for $a \in \Gamma_3^*$, using \eqref{eq:nice_relation} and the fact that multiplication is a complex unit, we get
\[ G_{lm}(k+a) = \langle e_a(k)^{-1} F_k(z+z_l) u, F_k(z+z_m)e_a(k)^{-1} u\rangle  = \vert e_a(k)\vert^{-2} \langle F_k(z+z_l) u, F_k(z+z_m)u \rangle .\]
We conclude that for $n$ the dimension of the Gramian matrix and $g(k):=\operatorname{det}(G(k))$, the Gramian determinant
\begin{equation}\label{eq:Gramiandeterminant}
    g(k+a)= \vert e_a(k)\vert^{-2n}g(k).
\end{equation}

Using \eqref{eq:Chern Number} and \eqref{eq:Gramiandeterminant} we find 
\[\begin{split} 
c_1 (V(k))
&= \frac{i}{2\pi} \int_{\CC/\Gamma_3^*} \tr \Theta \\
&= \frac{i}{2\pi} \int_{\partial \CC/\Gamma_3^*} \partial_k \log g(k) \ dk -\frac{i}{2\pi}\lim_{\varepsilon \downarrow 0} \int_{\partial D(0,\varepsilon)} \partial_k \log g(k) \ dk \\
&= \frac{ni}{2\pi}\left(\log e_{\gamma_2} (\gamma_1)- \log e_{\gamma_2} (0) - \log e_{\gamma_1} (\gamma_2) + \log e_{\gamma_1} (0)\right)-\frac{i}{2\pi} \lim_{\varepsilon \downarrow 0}\int_{\partial D(0,\varepsilon)} \partial_k \log g(k) \ dk\\
&=-n-\lim_{\varepsilon \downarrow 0}\frac{i}{2\pi} \int_{\partial D(0,\varepsilon)} \partial_k \log g(k) \ dk. \end{split}\]

It remains to evaluate the limit on the last line of the above equation. The determinant is given by 
\[g(k) = \det (G_{lm}(k)) = \det (\langle F_k(z+z_l)u,F_k(z+z_m)u \rangle).\] 
We will follow the proof of \cite[Theorem  5]{bhz23}. Suppose first we have proven that we have an expansion of the form
\begin{equation}
\label{eq:TT}
    g(k) = g_0 |k|^{2n-2} + \mathcal{O}(|k|^{2n-1})
    \end{equation}
with $g_0\neq 0$. It is then easy to evaluate the limit:
\begin{equation*}
\label{eq:limit}
 \begin{split} 
- \frac{i}{ 2 \pi } \lim_{\varepsilon \to 0 }  \int_{ \partial D ( 0 , \varepsilon ) }  \partial_k \log g ( k ) dk & = 
- \frac{i}{ 2 \pi } \lim_{\varepsilon \to 0 }  \int_{ \partial D ( 0 , \varepsilon ) }  \partial_k \log (g_0 |k|^{2n-2}+\mathcal{O}(|k|^{2n-1}) )  dk 
\\
& = - \frac{i}{ 2 \pi } \lim_{\varepsilon \to 0 } \int_{ \partial D ( 0 , \varepsilon ) } \frac{ (n-1) g_0 |k|^{2n-4} \bar k + \mathcal O ( |k|^{2n-2} ) }  { g_0 |k|^{2n-2}+\mathcal{O}(|k|^{2n-1})} dk
\\
& =  - \frac{i(n-1)}{ 2 \pi } \lim_{\varepsilon \to 0 }  \int_{ \partial D ( 0 , \varepsilon ) } ( k^{-1} 
+ \mathcal O ( 1 ) ) 
dk 
\\ & =  - \frac{i(n-1)}{ 2 \pi }  \lim_{\varepsilon \to 0 } ( 2 \pi i + \mathcal O ( \varepsilon ) ) = n-1.
\end{split}
\end{equation*}
Hence, $ c_1 ( V(k) ) = -1$ for all above cases, i.e., $n=1,2,3$. 

The proof will therefore be complete if we prove identity \eqref{eq:TT} for $n=3$, which we do below. The argument is essentially a generalisation of the proof of \cite[Theorem  5]{bhz23}. In that paper, it was remarked that the Chern connection coincided with the Berry connection, the latter one being convenient for our calculation as it gives our vector bundle a hermitian structure, see \cite[Proposition 9.1]{bhz23}.
\end{proof}

In our case, we can as in \eqref{eq:gram} take $ k \mapsto u_j ( k ) \in \mathscr H=L^2_{-ir} $, $ j = 1,2,  3 $,
a local holomorphic basis of $ E $. Then for 
\begin{equation*}
\label{eq:defPhi}  \Phi ( k ) := \wedge_{ j=1}^3 u_j ( k ) \in  \wedge^3 E_k \subset \wedge^3 \mathscr H , \end{equation*}
we have
\[  \| \Phi ( k ) \|^2_{\wedge^3 \mathscr H} = 
\det \left( (\langle u_j (k ), u_\ell ( k ) \rangle_{\mathscr H}  ) _{ 1\leq j,\ell \leq 3 }\right) = \det G ( k ) = g ( k )  . \]
We can now prove the following lemma, which implies \eqref{eq:TT}.
\begin{lemm}
    We have the following identities
\begin{equation}
    \label{eq:IDENTITY}
    \forall (j,\ell)\neq (2,2) \in (\mathbb N\cup \{0\})^2, \quad j+\ell \leq 4, \quad \partial_{\bar k}^j\partial_k^{\ell}g(k)_{|k=0}=\partial_{\bar k}^j\partial_k^{\ell}\| \Phi ( k ) \|^2_{|k=0}=0.
\end{equation}    
  But for $j=2,\ell=2$, we have $\partial_{\bar k}^2\partial_2^{\ell}g(k)_{|k=0}\neq 0.$  
\end{lemm}
\begin{proof}
We have $\Phi(k)= F_1 (k)  u_0  \wedge F_2 ( k )  u_0  \wedge F_3(k)u_0$ and is holomorphic in $k$, which means that
\begin{equation}
    \label{eq:III}
\partial_{\bar k}^j\partial_k^{\ell}g(k)_{|k=0}=\langle \partial_k^{\ell}\Phi(k)_{|k=0},\partial_{k}^j\Phi(k)_{|k=0}\rangle. 
\end{equation}
We then note that $(F_k)(z)\vert_{k=0}=1$, and thus, when computing the derivative $\partial_k^{\ell}\Phi(k)$, we see that if the derivative does not land on two of the terms, then we get a wedge sum with two terms being equal to $u_0$ when evaluating at $k=0$, which proves that the wegde product is zero. In particular, we see that under the condition of \eqref{eq:IDENTITY}, one of the two vectors in \eqref{eq:III} vanishes and so does the corresponding derivative of $g$. This proves the first part of the lemma. For the second point, we use the previous remark to get, with $\tilde F_i=\partial_k (F_i)_{k=0}$,
\begin{equation*}
    \partial_{\bar k}^2\partial_2^{\ell}g(k)_{|k=0}=\|\tilde F_1u_0\wedge \tilde F_2u_0 \wedge u_0+\tilde F_1u_0\wedge u_0 \wedge\tilde F_3u_0+ u_0\wedge \tilde F_2u_0 \wedge \tilde F_3u_0\|^2.
\end{equation*}
This last derivative vanishes if and only if we have the vector identity
\begin{equation}
    \label{eq:LHS}
\tilde F_1u_0\wedge (\tilde F_2u_0-\tilde F_3u_0) \wedge u_0= \tilde F_2u_0 \wedge \tilde F_3u_0\wedge u_0. 
\end{equation}
To see that this is impossible we recall that $u_0$ has, by assumption, $3$ simple zeroes $z_1,z_2,z_3$ and that $F_i$ has a pole of order $1$ at $z_i$. We then see that $\tilde F_i$ also has a simple pole at $z_i$ which gives 
\begin{equation}
    \label{eq: vanish}
    \forall j,i \in \{1,2,3\}, \quad \tilde F_i u_0(z_j)=0 \ \iff \ i\neq j.
\end{equation}
We can now fix the first variable to be equal to $z_1$ and thus, using \eqref{eq: vanish} and the determinant expression of the wedge product,
$$(\tilde F_2u_0 \wedge \tilde F_3u_0\wedge u_0)(z_1,x_2,x_3)=0, \quad  \forall x_2,x_3. $$
But we can now compute the left hand side of \eqref{eq:LHS} by expanding on the first line to get
$$ (\tilde F_1u_0\wedge (\tilde F_2u_0-\tilde F_3u_0) \wedge u_0)(z_1,x_2,x_3)=\underbrace{\tilde F_1u_0(z_1)}_{\neq 0}\times((\tilde F_2u_0-\tilde F_3u_0) \wedge u_0)(x_2,x_3).$$
This would then imply that $(\tilde F_2u_0-\tilde F_3u_0) \wedge u_0=0$, which is impossible by \cite[(9.25)]{bhz23}.
\end{proof}

\section{Exponential squeezing of bands}
\label{s:exp}

Here we study exponential squeezing of bands in the limit of small twisting angles. 
We shall assume that $\zeta_1=p\theta$ for fixed $p$ as $\theta\to0^+$, and that we have a constant angle ratio $\zeta_2/\zeta_1\equiv r$. 
We shall consider the Floquet Hamiltonian $H_k(\alpha)$ in \eqref{eq:Floquet_intro} defined for a general potential satisfying the symmetries
$U(z+a)=\bar\omega^{a_1+a_2} U(z)$ for $a=\frac{4\pi i}{3}(\omega a_1+\omega^2 a_2)\in\Gamma_3$, and $U(\omega z)=\omega U(z)$. By Proposition \ref{prop:1} we then have
\begin{equation}\label{eq:Usqueez}
U(z)=\sum_{n,m\in\Z} c_{nm}f_{nm}(z)
\end{equation}
where
\begin{equation*}
f_{nm}(z)=\exp\left(\frac{i}2(-\sqrt 3(n+m)\re z+(2-3(n-m))\im z) \right)
\end{equation*}
and
\begin{equation*}
c_{(m-n+1)(-n)}=\omega c_{nm},\qquad  c_{(-m)(n-m-1)}=\omega^{2}c_{nm} ,\qquad n,m\in\Z.
\end{equation*}
Note that each orbit
\begin{equation*}
(n,m)\to (m-n+1,-n)\to(-m,n-m-1)\to (n,m)
\end{equation*}
in $\ZZ^2$ is closed, which is in agreement with the fact that $\omega^3=1$.
Define an equivalence relation $\sim$ in $\Z^2$ by the condition that $(n,m)\sim(n',m')$ if $(n,m)$ and $(n',m')$ are in the same orbit, and let $S$ be the set of equivalence classes. Then $$U(z)=\sum_{[(n,m)]\in S} g_{[(n,m)]}(z),$$
where
\begin{align*}
g_{[(n,m)]}(z)&=c_{nm}f_{nm}(z)+c_{(m-n+1)(-n)}f_{(m-n+1)(-n)}(z)+
c_{(-m)(n-m-1)}f_{(-m)(n-m-1)}(z)\\&=c_{nm}\Big( e^{\frac{i}2(-\sqrt 3(n+m)\re z +(2-3(n-m))\im z)}
+\omega e^{\frac{i}2(-\sqrt 3(m-2n+1)\re z +(2-3(m+1))\im z)}\\
&\qquad\qquad+\omega^2 e^{\frac{i}2(-\sqrt 3(n-2m-1)\re z +(2-3(1-n))\im z)}\Big).
\end{align*}
A straightforward calculation shows that
\begin{equation}\label{eq:derivativesofg}
\begin{aligned}
\partial_z g_{[(n,m)]}(0)&=3 c_{nm}\partial_zf_{nm}(0)=\tfrac{3i}{4}c_{nm}[-\sqrt 3(n+m)-i(2-3(n-m))],\\
\partial_{\bar z} g_{[(n,m)]}(0)&=0.
\end{aligned}
\end{equation}
(This does not depend on the choice of representative in $[(n,m)]$.)
We shall require that
\begin{equation}\label{eq:realanalytic}
| c_{nm} | \leq c_0 e^{-c_1 (|n|+|m|) } ,    
\end{equation}
for some constants $ c_0, c_1 > 0 $, which is equivalent 
to real analyticity of $ U $. We also make the generic non-degeneracy assumption that
\begin{equation}\label{eq:nondeg}
0\ne\re(\partial_z U(0))=\sum_{[(n,m)]\in S} \re(\tfrac{3i}{4}c_{nm}[-\sqrt 3(n+m)-i(2-3(n-m))]).
\end{equation}
This is verified by the standard potential $U_0$ in \eqref{eq:standard_pot}, since $\partial_zU_0(0)=3/2$.
For such potentials we have the following result.

\begin{theo}
\label{t:squeezegenpot}
Let
$ H_k ( \alpha ) $ be given by 
\eqref{eq:Floquet_intro} with $ U $ defined in 
\eqref{eq:Usqueez}, and $\Spec_{ L^2 ( \CC/\Gamma_3 ) } H_k ( \alpha ) = \{  E_{\pm j} ( k , \alpha ) \}_{ j =1 }^\infty $ as in \eqref{eq:specHk}. Assume that $\zeta_1=p\theta$ with $p$ fixed as $\theta\to0^+$, and that $\zeta_2/\zeta_1\equiv r$.
If $U$ satisfies \eqref{eq:realanalytic} and \eqref{eq:nondeg}, then
 there exist positive constants $c_0$, $c_1$, and $c_2$ such that
for all $ k \in \CC $, 
\begin{equation*}
  E_j ( k, \alpha)  \leq c_0 e^{ - c_1 |\alpha| } , 
\ \  1\le  j \leq c_2 |\alpha|, \ \  \alpha=(1/\theta)\beta,\quad \theta > 0 .
\end{equation*}
\end{theo}

Since $\zeta_1=p\theta$ with $p$ fixed, we have $\alpha=(p/\zeta_1)\beta=(1/\theta)\beta$, so for the proof we take $h=\theta$ as a semiclassical parameter, and write
\begin{equation*} 
H_k(\alpha)=h^{-1}\begin{pmatrix} 0  & P(h)^*-h\bar k \\ P(h)-hk&0 \end{pmatrix},
\end{equation*}
where
\begin{equation*} 
 P(h) = \begin{pmatrix} 2h D_{\bar z} &  \beta_{12} U( pz ) & 0 \\ 
  \beta_{12} U(- pz) &2h D_{\bar z} & \beta_{23} U(pr z)\\ 
0 &   \beta_{23} U(-pr z) &2h D_{\bar z} \end{pmatrix}
\end{equation*}
with $r\equiv \zeta_2/\zeta_1$. We thus study the equation $(P(h)-hk)u(h)=0$ for small $h>0$.

The semiclassical principal symbol of $ P (h) - h k $ (see \cite[Proposition E.14]{res}) is given by 
\begin{equation*}
P_0(z,\bar z,\bar \zeta) = \begin{pmatrix} 2\bar \zeta &  \beta_{12} U( pz ) & 0 \\ 
  \beta_{12} U(- pz) &2\bar \zeta & \beta_{23} U(pr z)\\ 
0 &   \beta_{23} U(-pr z) &2\bar \zeta \end{pmatrix} , 
\end{equation*}
where we have used complex notation $ \zeta = \frac 12 ( \xi_1 - i \xi_2 ) $ and 
$ z = x_1 + i x_2 $. The Poisson bracket can then be expressed as 
\begin{equation}
\label{eq:defPo} \{ a , b \} = \sum_{ j=1}^2 \partial_{\xi_j } a \partial_{x_j} b -
\partial_{\xi_j } b \partial_{x_j} a = 
\partial_\zeta a \partial_z b  - \partial_\zeta b \partial_z a + \partial_{\bar \zeta } a \partial_{ 
\bar z } b - \partial_{\bar \zeta } b \partial_{ 
\bar z } a . \end{equation}

To prove Theorem \ref{t:squeezegenpot} we will use the analytic version \cite[Theorem 1.2]{dsz} of H\"ormander's quasimode construction based on the bracket condition. That is, if $ Q = \sum_{ |\alpha| \leq m } 
a_\alpha ( x, h ) ( h D)^\alpha $ is a differential operator such that 
$ x \mapsto a_\alpha ( x, h ) $ are real analytic near $ x_0 $, and
$ Q_0 ( x, \xi ) $ is the semiclassical principal symbol of $ Q $, then existence of $(x_0,\xi_0)$ such that
\begin{equation}
\label{eq:Pa}  Q_0 ( x_0 , \xi_0 ) = 0 , \ \  \{ Q_0 , \bar Q_0 \} ( x_0 , \xi_0 ) \neq 0 , 
\end{equation}
implies that there exists a family $ v_h \in C^\infty_{\rm{c}} ( \Omega ) $, $ \Omega $ a
neighbourhood of $ x_0 $, such that
\begin{equation}
\label{eq:quas}
| (h \partial)^\alpha_x Q  v_h (x ) | \leq C_\alpha e^{ - c / h } , \ \
\| v_h \|_{L^2}  = 1, \ \ | (h\partial_x)^\alpha v_h ( x ) | \leq C_\alpha e^{ - c | x- x_0|^2/ h } , 
\end{equation}
for some $ c > 0$. As noted in \cite{BEWZ22}, this formulation is different than in 
the statement of \cite[Theorem 1.2]{dsz}, but \eqref{eq:quas} follows
from the construction in \cite[\S 3]{dsz}.
We will use it to prove the following proposition.

\begin{prop}
\label{p:dsz}
There exists an open set $ \Omega \subset \CC $ and a constant $ c $
such that for any $ k \in \CC $ and $ z_0 \in \Omega $ there
exists a family  $ h \mapsto u_h \in C^\infty ( \CC/\Gamma ; \CC^3 ) $ such that for
$ 0 < h < h_0 $, 
\begin{equation*}
  | ( P ( h ) - h k ) u_h ( z ) | \leq e^{ - c /h }, 
\ \ 
\| u_h\|_{L^2}  = 1, \ \  | u_h ( z ) | \leq e^{ - c | z- z_0|^2/h } . \end{equation*}
\end{prop}
\begin{proof}
We follow the strategy for proving \cite[Proposition 4.1]{BEWZ22} and look for a point $z_0$ where $U(pz_0),U(prz_0)\ne0$, which will allow us to reduce to the case of a scalar equation so that we can apply \eqref{eq:quas}. If we set
\begin{align*}\label{eq:defvh} 
 Q  &:=  rU ( pz )U(prz) ( 2 h D_{\bar z} - h k ) \left( (rU (prz))^{-1} ( 2 h D_{\bar z} - h k ) \left(U(pz)^{-1}(2hD_{\bar z}-hk)\right) \right) \\ & \quad - 
r^2U(prz)U ( -prz)(2hD_{\bar z})-U(pz)U(prz)\left((2hD_{\bar z}-hk)( U (-p z)U(prz)^{-1}\bullet\right),
\end{align*}
then 
existence of $u_h $ follows from the existence of 
$ v_h \in C^\infty_{\rm{c}} ( \Omega'; \CC )  $, where $ \Omega' $ is a small 
neighbourhood of $ z_0 $ such that $U(pz),U(prz)\ne0$ for $z\in\Omega'$, such that
\begin{equation*}
Q v_h = \mathcal O ( e^{-c/h} ) , \ \ v_h (z_0 )  = 1, \ \ 
|v_h ( z ) | \leq e^{ - c | z-z_0|^2/h } ,
\end{equation*}
with estimates for derivatives as in \eqref{eq:quas}.
We then put
\[ u_h := \begin{pmatrix} v_h \\ -U ( pz)^{-1} ( 2 h D_{\bar z } - 
h k  ) v_h \\ (rU(prz))^{-1} [(2 h D_{\bar z })U(pz)^{-1}(2 h D_{\bar z }-hk) -U(-pz)]v_h    \end{pmatrix}  \]
and normalize to have $ \| u_h \|_{L^2}= 1$.
Since such $ v_h $ are supported in small neighbourhoods, this defines an 
element of $ C^\infty ( \CC/ \Gamma_3 , \CC^3 ) $.

The semiclassical 
principal symbol of $ Q $ is given by
\begin{equation}
\begin{aligned}
\label{eq:defV}
Q_0(z,\bar z, \bar \zeta) &:=\operatorname{det}(p(z,\bar z,\bar \zeta) )=2\bar \zeta(4 \bar \zeta^2 - V ( z, \bar z )), \\
V ( z, \bar z) & := \beta_{12}^2U ( pz ) U ( -p z )+\beta_{23}^2U ( prz ) U ( -pr z ). 
\end{aligned}
\end{equation}
To use \eqref{eq:quas} we need to check 
 H\"ormander's bracket condition \eqref{eq:Pa}: for $ z $ in an open
 neighbourhood of $ z_0 $ where $ U ( pz),U(prz ) \neq 0 $, there exists
 $ \zeta $ such that
\[  Q_0 ( z, \bar z , \bar \zeta ) = 0, \ \ \{ Q_0 , \bar Q_0 \} ( z,\zeta ) 
\neq 0 .\]
When $ Q_0 =0$ we have either $\zeta=0$ or $4 \bar \zeta^2 = V ( z) $.
A simple calculation using \eqref{eq:defPo} shows that $\{Q_0,\bar Q_0\}=0$ when $\zeta=0$. When $4 \bar \zeta^2 - V ( z)=0 $ we can take $ \zeta =  \frac12 \overline V^{\frac 12 } $ for either branch of the 
square root, which by
using \eqref{eq:defPo} gives
\begin{equation}\label{eq:bracketcomp}
  \begin{split} i \{ Q_0 , \bar Q_0 \} & = i ( (24\bar \zeta^2-2V)\bar \partial_z + 2\bar\zeta
\partial_z V \partial_\zeta ) (8 \zeta^3 - 2\zeta\overline V ) \\
&=4i((12\bar\zeta^2-V) \zeta\overline{\partial_zV}-(12\zeta^2-\overline{V})\bar \zeta\partial_zV)\\&=
 32 i (  \bar\zeta^2\zeta \partial_z V -\zeta^2 \overline{ \zeta \partial_z V }   )   
 = -  64  \Im ( \bar\zeta^2\zeta \partial_z V )  
 = - 8 \lvert V\rvert \Im ( {\overline{ V}}^{\frac12}   \partial_z V )  ,
\end{split}
\end{equation}
where we used $12\bar\zeta^2-V=8\bar\zeta^2$ when $4\bar\zeta^2-V=0$ in the second row.
We now just need to verify that the right-hand side is non-zero at some point $ z_0$.

We will do this by using Taylor's formula on $|V|\Im ( {\overline{ V}}^{\frac12}   \partial_z V ) $.  
By \eqref{eq:derivativesofg} we have
\begin{equation*}
U(pz)=p\partial_z U(0)z+O(\lvert z\rvert^2),\quad U(prz)=pr\partial_z U(0)z+O(\lvert z\rvert^2).
\end{equation*}
In view of the definition \eqref{eq:defV} of $V$ we see that $V(0)=\partial_z V(0)=\partial_{\bar z} V(0)=0$, and
$$
\partial_z^2 V(0)=-2p^2(\beta_{12}^2+r^2\beta_{23}^2)(\partial_z U(0))^2,\qquad \partial_z\partial_{\bar z} V(0)=\partial_{\bar z}^2 V(0)=0.
$$
With $C=C(p,r,\beta)=p^2(\beta_{12}^2+r^2\beta_{23}^2)$ it follows that
\begin{align*}
V(z)=-z^2C(\partial_z U(0))^2(1+O(\lvert z\rvert)),\quad
 \partial_zV(z)=-2zC(\partial_z U(0))^2(1+O(\lvert z\rvert)),
\end{align*}
which gives
\begin{align*}
{\overline{ V}}^{\frac12}(z)\partial_zV(z)&=\overline{\sqrt{-z^2C(\partial_z U(0))^2}}(-2zC(\partial_z U(0))^2)(1+O(\lvert z\rvert))\\
&=2i\lvert z\rvert^2C^{3/2}\lvert \partial_z U(0)\rvert^2\partial_z U(0)(1+O(\lvert z\rvert)).
\end{align*}
Since $C>0$ and $\Re \partial_z U(0)\ne0$ by the non-degeneracy assumption \eqref{eq:nondeg}, we see from this that $|V|\Im ( {\overline{ V}}^{\frac12}   \partial_z V )\ne0$ in a punctured neighborhood of the origin. This completes the proof.
\end{proof}

\begin{rmks}
1. The set where the bracket in \eqref{eq:bracketcomp} vanishes can be computed numerically, see Figure \ref{fig:vanishingbrackets}.

 2. The rate of decay of $E_j(k,\alpha)$ described by the constant $c_1$ in Theorem \ref{t:squeezegenpot} seems to depend on the ratio $\zeta_2/\zeta_1$, see Figure \ref{fig:decay}. The rate is significantly faster for the case of equal twisting angles $\zeta_2/\zeta_1=1$.
\end{rmks}

\begin{proof}[Proof of Theorem \ref{t:squeezegenpot}]
Having established Proposition \ref{p:dsz}, the result now follows by using the arguments in the proof of \cite[Theo.5]{BEWZ22} with the obvious adjustments, such as using Proposition \ref{p:dsz} instead of \cite[Proposition 4.1]{BEWZ22} and changing $\CC^2$ to $\CC^3$ in the construction. 
\end{proof}

\begin{figure}
\includegraphics[scale=.6,trim={.25cm 1cm .25cm 0cm},clip]{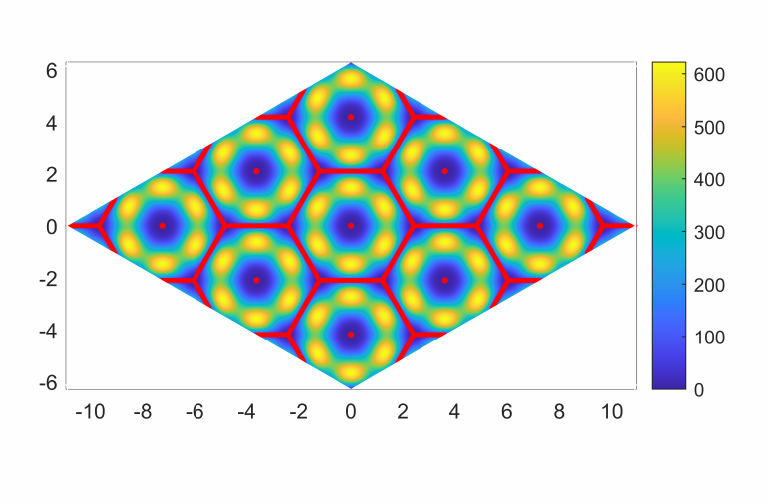}\quad
\includegraphics[scale=.6,trim={.25cm 1cm .25cm 0cm},clip]{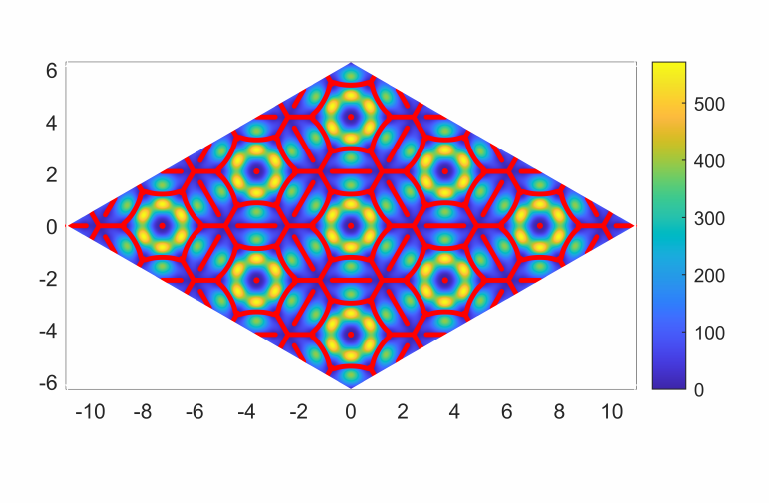}\\
\includegraphics[scale=.6,trim={.25cm 1cm .25cm 0cm},clip]{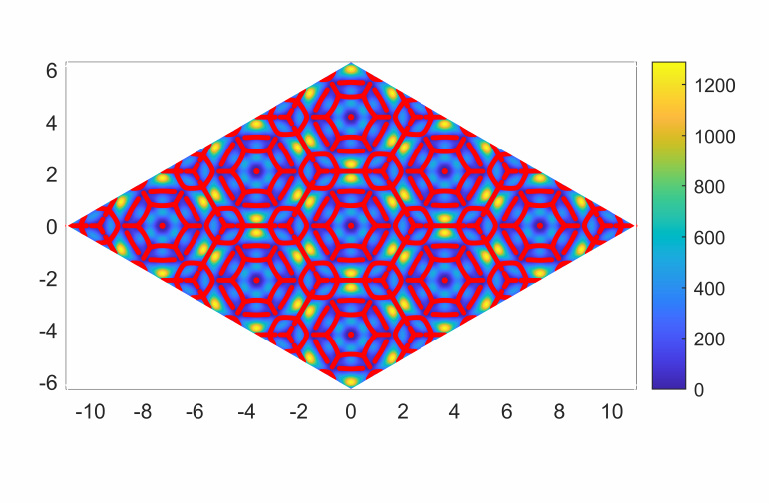}\quad
\includegraphics[scale=.6,trim={.25cm 1cm .25cm 0cm},clip]{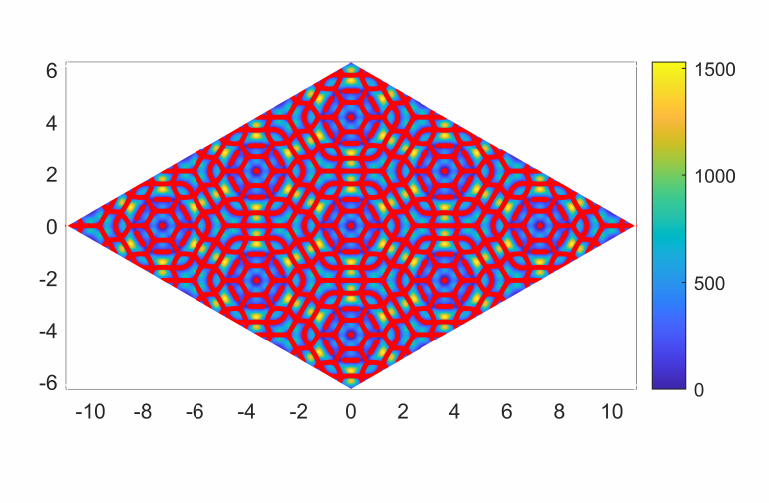}
\caption{\label{fig:vanishingbrackets} Plots of the values of $\lvert\{Q_0,\bar Q_0\}\rvert=8\lvert V\rvert\lvert \im(\overline{V}^\frac12\partial_zV)\rvert$ over a fundamental domain of $\CC/\Gamma$ (9 times larger than $\CC/\Gamma_3$) for hopping amplitudes $\beta=(\beta_{12},\beta_{23})=(1,1)$ and different values of $\zeta_1/\zeta_2$. The zero set is colored red, and we see that the bracket $\{Q_0,\bar Q_0\}$ is non-zero in a punctured neighborhood of the origin. In order top left, top right, bottom left, bottom right, we have $\zeta_1/\zeta_2=1,2,3,4$.}
\end{figure}

\begin{figure}[ht!]
    \centering
    \includegraphics[width=6cm]{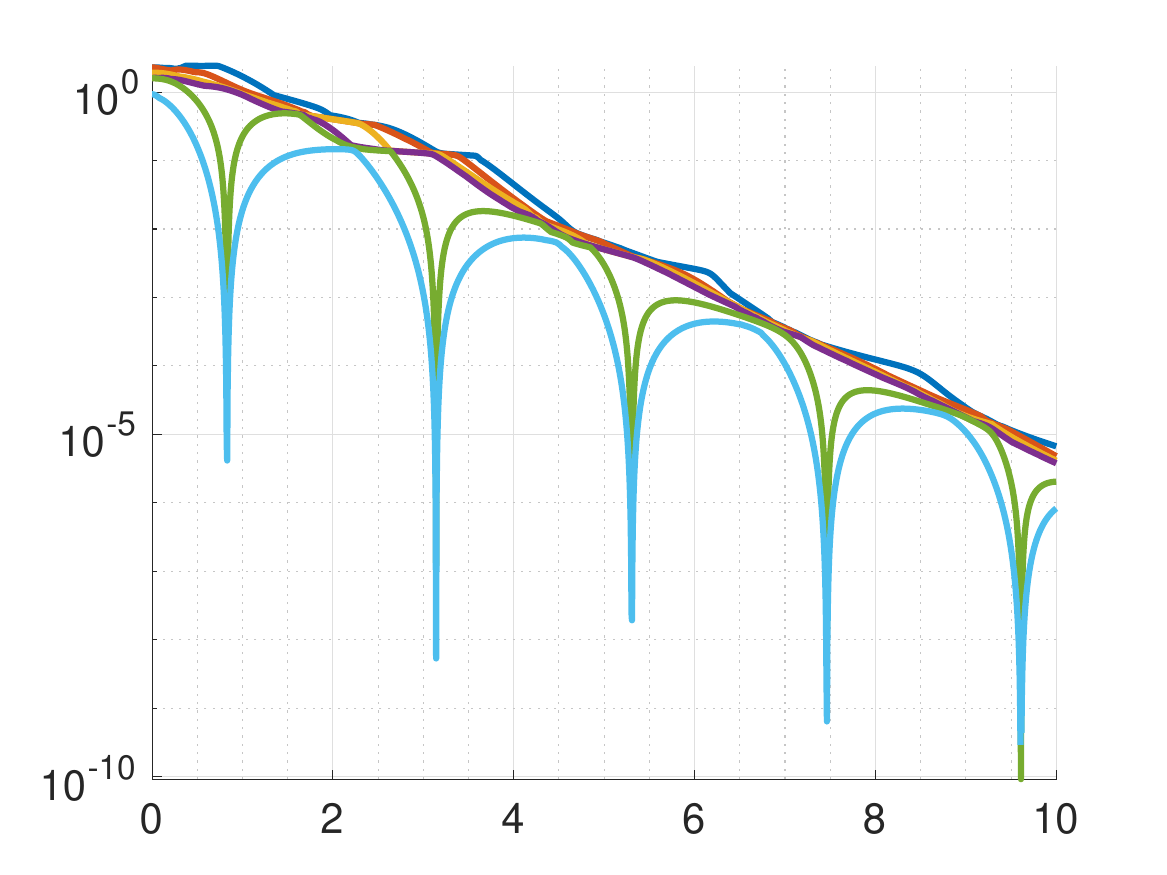}
    \includegraphics[width=6cm]{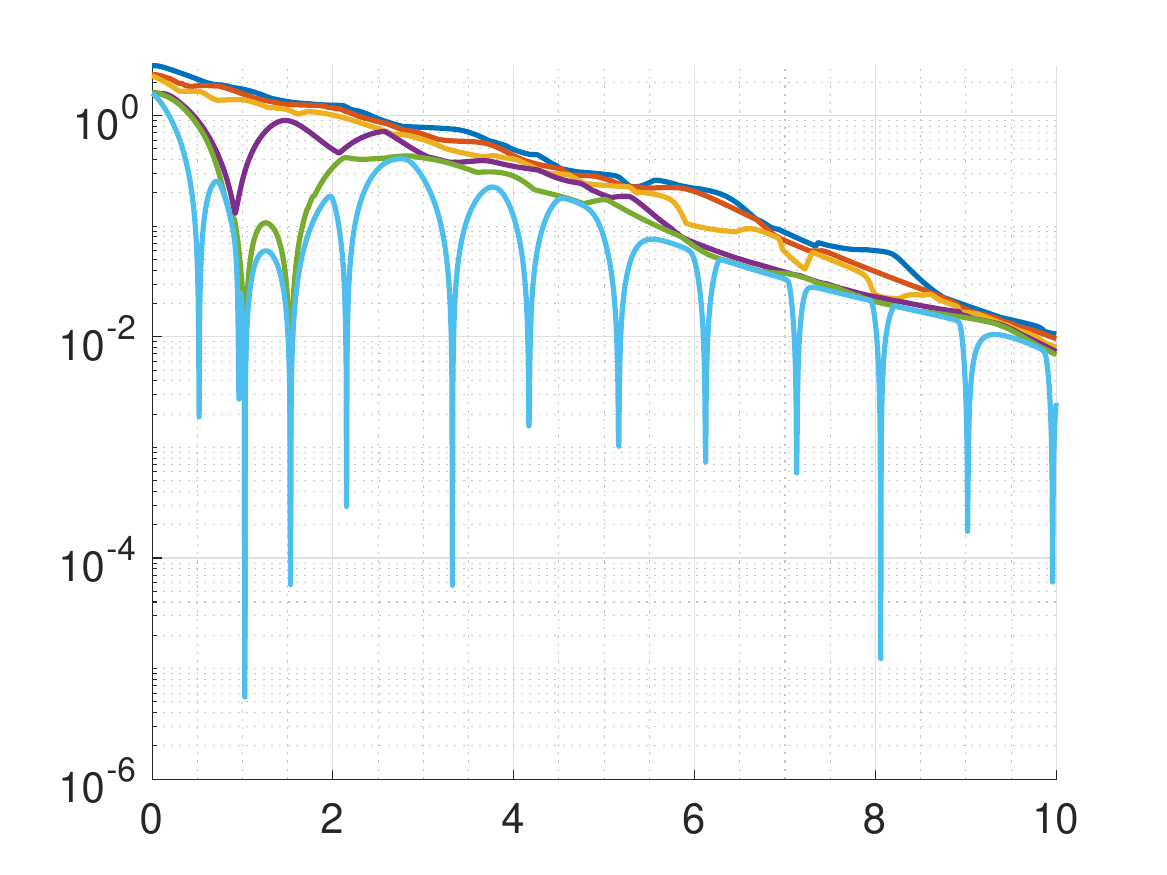} \\ 
    \includegraphics[width=6cm]{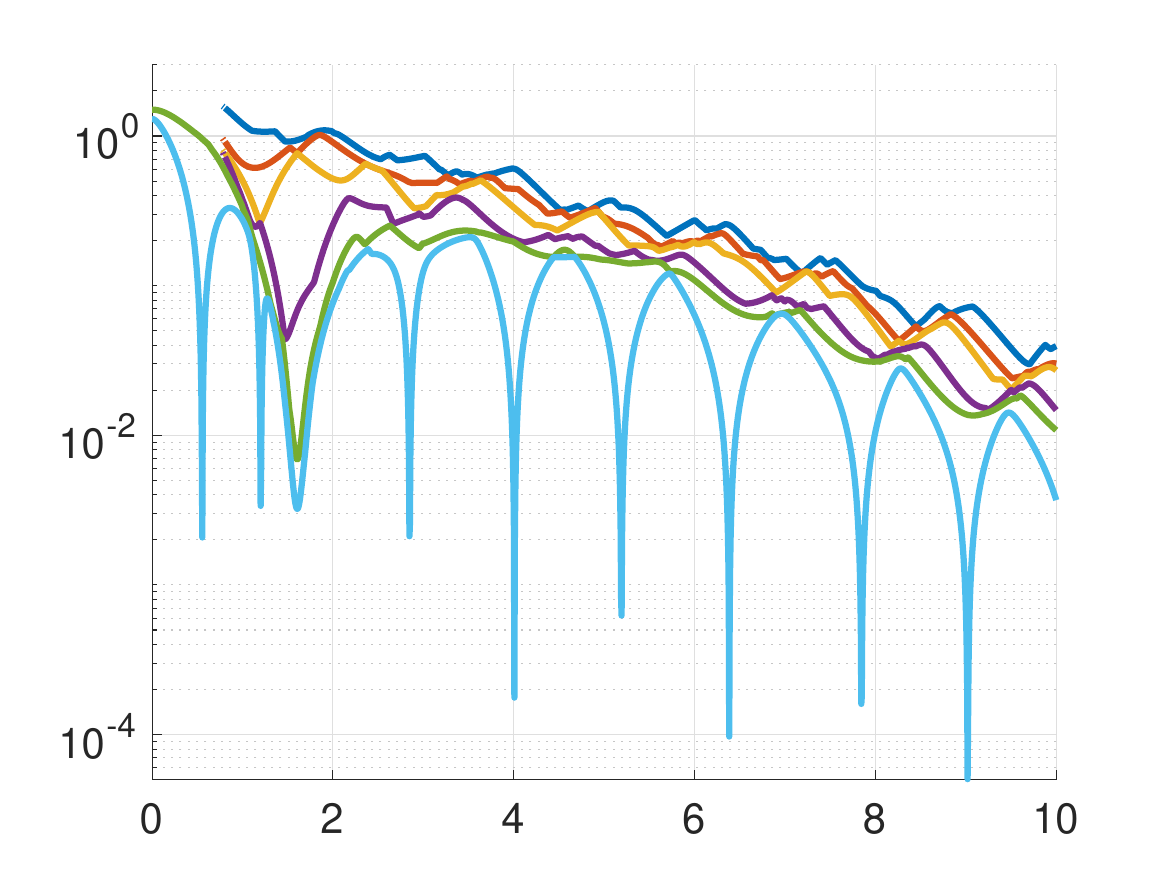}
    \includegraphics[width=6cm]{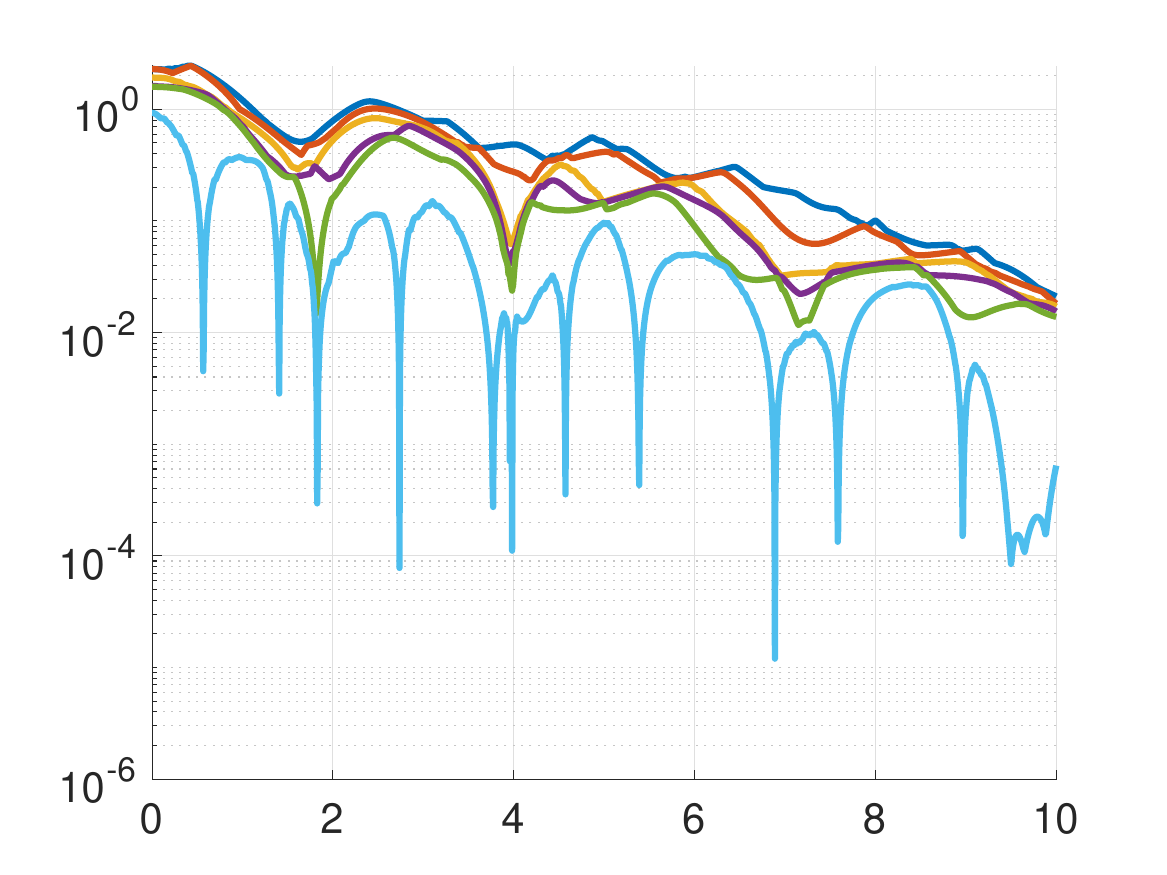}
    \caption{The 6 lowest positive eigenvalues ($y$-axis) of the chiral Hamiltonian for varying $\alpha_{12}=\alpha_{23}$ ($x$-axis) and $\zeta_1=\zeta_2=1$ (top left) and $\zeta_1/\zeta_2=2$ (top right); $\zeta_1/\zeta_2=3$ (bottom left) and $\zeta_1/\zeta_2=4$ (bottom right).}
    \label{fig:decay}
\end{figure}
\section*{Acknowledgements} We thank Fedor Popov and Grigory Tarnopolsky for interesting discussions. The authors are very grateful to Zhongkai Tao and Maciej Zworski for many helpful discussions on this project. The research of JW was supported by The Swedish Research Council grant 2019-04878. MY was partially supported by the National Science Foundation under the grant DMS-1901462 and Simons Fundation under the ``Moire Material Magic'' grant.

\appendix

 \section{Anti-chiral limit}
\label{sec:ACL}

The Hamiltonian in the anti-chiral limit $\alpha=0$ can be written, by conjugating with the unitary matrix  $\mathbf V = (e_1,e_4,e_3,e_5,e_2,e_6)$, in the off-diagonal form
\[ H _{\operatorname{ac}}=\mathbf VH(\tilde \alpha) \mathbf V^t= \begin{pmatrix} 0 & D_{\operatorname{ac}} \\ D_{\operatorname{ac}}^* & 0 \end{pmatrix} \text{ with } D_{\operatorname{ac}}= \begin{pmatrix}  \tilde \alpha_{12} V(p z) & 2D_{z} & 0 \\  2D_{\bar z} & (\tilde \alpha_{12} V(p z))^* & \tilde \alpha_{23} V(p\frac{\zeta_2}{\zeta_1} z) \\
(\tilde \alpha_{23} V(p\frac{\zeta_2}{\zeta_1} z))^* & 0 & 2D_{z} \\\end{pmatrix}. \]

\begin{theo}
The anti-chiral Hamiltonian $H_{\operatorname{ac}}$ does not have any flat bands at energy zero.
\end{theo}
\begin{proof}
We partition the off-diagonal part of the Bloch-Floquet transformed Hamiltonian
\[ D_{\operatorname{ac},k} = \begin{pmatrix} A_{11}(k) & A_{12} \\ A_{21}& A_{22}(k) \end{pmatrix}\text{ with } A_{11}(k) = \begin{pmatrix} \tilde \alpha_{12} V(p z) & 2D_{z} + \bar k \\ 2D_{\bar z} + k &  (\tilde \alpha_{12} V(p z))^*\end{pmatrix} \text{ and }A_{22}(k)= 2D_{z}+\bar k. \]
In particular, for $k=k_1+i k_2$ we have $$A_{11}(k) = \sum_{j=1}^2 (D_j-k_j)\sigma_j + \tilde\alpha_{12}(\Re(V(pz)) \operatorname{id}+i \Im(V(pz)) \sigma_3 ).$$
To show that there does not exist a flat band, it suffices to show that $Q_k$ is invertible for some $k.$
If $k \notin \Gamma^*$ then $A_{22}(k)^{-1}$ exists and by the block matrix inversion formula 
\[ \begin{split}
D_{\operatorname{ac},k}^{-1} &=\operatorname{diag}((A_{11}(k)-A_{12}A_{22}(k)^{-1}A_{21})^{-1},(A_{22}(k)-A_{21}A_{11}(k)^{-1}A_{12})^{-1}) \times \\
& \quad   \begin{pmatrix} \operatorname{id} & -A_{12}A_{22}(k)^{-1} \\ -A_{21}A_{11}(k)^{-1} & \operatorname{id}  \end{pmatrix}\\
&=\operatorname{diag}(A_{11}(k)^{-1},(A_{22}(k)-A_{21}A_{11}(k)^{-1}A_{12})^{-1})  \begin{pmatrix} \operatorname{id} & -A_{12}A_{22}(k)^{-1} \\ -A_{21}A_{11}(k)^{-1} & \operatorname{id}  \end{pmatrix}.
\end{split}
\]
The goal of the proof is to show that all terms on the right-hand side are well-defined for suitable $k$, which shows that $Q_k$ is invertible and thus the anti-chiral Hamiltonian does not exhibit a flat band at energy zero. The invertibility of $A_{11}(k)$ coincides with the expression found in the twisted bilayer case and thus one can see that $A_{11}(k)$ is invertible using \cite{BEWZ21}. Thus it suffices to show the invertibility of $A_{22}(k)-A_{21}A_{11}(k)^{-1}A_{12}$.

To do so we recall that for $k=k_1+ik_2$ we have $A_{11}(k) = \sum_{j=1}^2 (D_j-k_j)\sigma_j + \tilde\alpha_{12} (\Re(V(z))+ i \Im(V(z))\sigma_3)$ such that 
in terms of $H_S = (D_1 - k_1)^2 + (D_2 - k_2)^2$, we find
\[ A_{11}(k)(A_{11}(k)-\tilde \alpha_{12} \Re(V(z,\bar z))) = (\operatorname{id}+WH_S(k)^{-1}) H_S(k),\]
for some $W \in L^{\infty}(\CC,\CC^{2 \times 2}).$
We conclude that 
\[ A_{11}(k)^{-1} = (A_{11}(k)-\tilde \alpha_{12} \Re(V(z,\bar z)))H_S(k)^{-1}(\operatorname{id}+WH_S(k)^{-1}).\]

By complexifying $k_1 = \mu_1 + i \mu_2$ with $\mu_i \in \RR$
\[\Vert (A_{11}(k)-\tilde \alpha_{12} \Re(V(z,\bar z)))H_S(k)^{-1} \Vert \lesssim \langle \mu_2 \rangle^{-1}\]
we find that $\Vert A_{11}(k)^{-1} \Vert  \lesssim \langle \mu_2 \rangle^{-1}$ which when $\mu_2$ is taken large enough implies that also 
$A_{22}(k)-A_{21}A_{11}(k)^{-1}A_{12}$ is invertible, since $A_{22}(k)$ is normal and invertible for $\bar k \notin \Gamma^*.$ This completes the argument.
\end{proof}

\end{document}